%% file: main.tex
\documentclass[a4paper,USenglish,numberwithinsect]{lipics-v2018}

\usepackage[T1]{fontenc}

\usepackage{microtype}
\usepackage[utf8]{inputenc}
\usepackage{graphicx}
\usepackage{booktabs} 
\usepackage{amssymb}
\usepackage{stmaryrd}
\usepackage{paralist}
\usepackage{amsmath}
\usepackage{url}
\usepackage{xcolor}
\usepackage{colonequals}
\usepackage{tikz}
\tikzstyle{every picture}=[>={stealth}]
\usepackage{breakurl}

\hypersetup{
    colorlinks,
    linkcolor={red!50!black},
    citecolor={blue!50!black},
    urlcolor={blue!30!black},
    breaklinks=true
}

\input{macros}

\newcommand{\myparagraph}[1]{\subparagraph*{#1}}

\title{Connecting Knowledge Compilation Classes and Width Parameters}

\titlerunning{Connecting Knowledge Compilation Classes and Width Parameters}

\author{Antoine Amarilli}{%
LTCI, T\'el\'ecom ParisTech, Universit\'e Paris-Saclay, Paris}{%
antoine.amarilli@telecom-paristech.fr}{}{}
\author{Florent Capelli}{CRIStAL, Université de Lille, CNRS, Lille}{%
florent.capelli@univ-lille.fr}{}{}
\author{Mikaël Monet}{%
LTCI, T\'el\'ecom ParisTech, Universit\'e Paris-Saclay, Paris}{%
mikael.monet@telecom-paristech.fr}{}{}
\author{Pierre Senellart}{DI ENS, ENS, CNRS, PSL University, Paris}{%
pierre.senellart@ens.fr}{}{}

\authorrunning{A. Amarilli and F. Capelli and M. Monet and P. Senellart}

\Copyright{Antoine Amarilli and Florent Capelli and Mikaël Monet and Pierre Senellart}


\acknowledgements{We acknowledge Chandra Chekuri for his helpful comments at
\url{https://cstheory.stackexchange.com/a/38943/}, as well as Stefan Mengel for
pointing us to a notion of width for d-SDNNFs and suggesting a strengthening of
our complexity upper bound in Theorem~\ref{thm:upper_bound}.}

\EventEditors{John Q. Open and Joan R. Access}
\EventNoEds{2}
\EventLongTitle{42nd Conference on Very Important Topics (CVIT 2016)}
\EventShortTitle{CVIT 2016}
\EventAcronym{CVIT}
\EventYear{2016}
\EventDate{December 24--27, 2016}
\EventLocation{Little Whinging, United Kingdom}
\EventLogo{}
\SeriesVolume{42}
\ArticleNo{23}
\nolinenumbers 
\hideLIPIcs  

\renewcommand{\phi}{\varphi}
\renewcommand{\epsilon}{\varepsilon}
\renewcommand{\leq}{\leqslant}
\renewcommand{\geq}{\geqslant}

\theoremstyle{plain}
\newtheorem{proposition}[theorem]{Proposition}
\newtheorem{claim}[theorem]{Claim}
\newtheorem{observation}[theorem]{Observation}
\newtheorem{result}{Result}
\begin{document}

\maketitle

\begin{abstract}
  \input{abstract}
\end{abstract}

\section{Introduction}
\label{sec:introduction}
\input{introduction}

\section{Preliminaries}
\label{sec:preliminaries}
\input{preliminaries}

\section{Knowledge Compilation Classes: BDDs and DNNFs}
\label{sec:kc}
\input{kc}

\section{Upper Bound}
\label{sec:result}
\input{result}

\section{Proof of the Upper Bound}
\label{sec:proof}
\input{construction}

\section{Lower Bounds for Structured Classes: $\SCOV_n$ and $\SINT_n$}
\label{sec:set}
\input{set}

\section{Lower Bounds for Structured Classes: General Case}
\label{sec:structured}
\input{structured}

\section{Lower Bounds for Unstructured Classes}
\label{sec:unstructured}
\input{unstructured}

\section{Conclusion}
\label{sec:conclusion}
\input{conclusion}



\bibliography{main}

\end{document}

%% file: macros.tex
\newcommand{\dom}{\lambda}
\newcommand{\defeq}{\colonequals}

\newcommand{\card}[1]{\left|#1\right|}

\newcommand{\calF}{\mathcal{F}}

\newcommand\NOT{\lnot}

\newcommand\VARS{\mathsf{Vars}}
\newcommand\LEAVES{\mathsf{Leaves}}

\newcommand\SINT{\mathrm{SINT}}

\newcommand\SCOV{\mathrm{SCOV}}

\renewcommand\root{\mathrm{root}}
\newcommand\spl{\mathrm{Split}}

\newcommand\psw{\mathrm{psw}}
\newcommand\tsw{\mathrm{tsw}}

\newcommand\pw{\mathrm{pw}}
\newcommand\tw{\mathrm{tw}}

\newcommand\degree{\mathrm{degree}}

\newcommand\UNF{\mathrm{Unj}}

\newcommand{\NN}{\mathbb{N}}

\newcommand{\arity}{\mathrm{arity}}

\newcommand{\var}{\mathsf{var}}

\newcommand{\out}{\mathsf{output}}

\newcommand\restr[2]{{%
  \kern-\nulldelimiterspace %
  #1 %
  _{|#2} %
  }}

%% file: abstract.tex
The field of \emph{knowledge compilation} establishes the tractability of many
tasks by studying how to \emph{compile} them to Boolean circuit classes obeying some
requirements such as structuredness, decomposability, and determinism. However,
in other settings such as intensional query
evaluation on databases, we obtain Boolean circuits that satisfy some width bounds,
e.g., they have bounded treewidth or pathwidth. In this work, we give a systematic
picture of many circuit classes considered in knowledge compilation and show
how they can be systematically connected to width measures, through
upper and lower bounds. Our upper bounds show that bounded-treewidth circuits
can be constructively converted to d-SDNNFs, in time linear in the circuit size
and singly exponential in the treewidth; and that bounded-pathwidth circuits can
similarly be converted to uOBDDs. We show matching lower
bounds on the compilation of monotone DNF or CNF formulas to structured targets, assuming a constant bound on the
arity (size of clauses) and degree (number of occurrences of each
variable): any d-SDNNF (resp., SDNNF) for such a DNF (resp., CNF) must be of
exponential size in its treewidth, and the same holds for uOBDDs (resp., n-OBDDs)
when considering pathwidth. Unlike most previous work, our bounds apply to \emph{any}
formula of this class, not just a well-chosen family. Hence, we show that
pathwidth and treewidth respectively characterize the efficiency of compiling monotone DNFs to
uOBDDs and d-SDNNFs with compilation being singly exponential in the
corresponding width parameter. We also show that our lower bounds on CNFs extend
to unstructured compilation targets, with an exponential lower bound in the
treewidth (resp., pathwidth) when compiling monotone CNFs of constant arity and
degree to DNNFs (resp., nFBDDs).

%% file: introduction.tex
\emph{Knowledge compilation} studies how problems can be solved by compiling
them into classes of Boolean circuits or binary decision diagrams (BDDs) to which
general-purpose algorithms can be applied. This field has introduced numerous
such classes or \emph{compilation targets}, defined by various restrictions on
the circuits or BDDs, and
studied which operations can be solved on them; e.g., the class of
\emph{d-DNNFs} requires that negation is only applied at the leaves,
$\land$-gates are on disjoint variable subsets, and $\lor$-gates have mutually
exclusive inputs. However, a different way to define restricted classes is to
bound some graph-theoretic \emph{width parameters}, e.g., treewidth, which
measures how the data can be decomposed as a tree, or
pathwidth, the special case of treewidth with path-shaped decompositions. Such restrictions have been used in particular in the field of
database theory and probabilistic databases~\cite{suciu2011probabilistic} in the
so-called \emph{intensional approach} where we compute a \emph{lineage
circuit}~\cite{jha2012tractability} that represents the output of a query or the
possible worlds that make it true, and where these circuits can sometimes be
shown to have bounded treewidth~\cite{jha2010bridging,amarilli2015provenance,amarilli2017combined}.

At first glance, classes such as
\emph{bounded-treewidth circuits} seem very different from usual knowledge
compilation classes such as d-DNNF. Yet, for some tasks such as probability computation
(computing the probability of the circuit under an independent distribution on
the variables), both classes are known to be tractable: the problem can be
solved in linear time on d-DNNFs by definition of the
class~\cite{darwiche2001tractable}, and for bounded-treewidth circuits we can
use \emph{message passing}~\cite{lauritzen1988local} to solve probability
computation in time linear in the circuit and 
exponential in the treewidth. This hints at the existence of a connection
between traditional knowledge compilation classes and bounded-width classes.

This paper presents such a connection and shows that the width of circuits
is intimately linked to many well-known knowledge compilation classes.
Specifically, we show a link between the \emph{treewidth} of Boolean circuits
and the width of their representations in common circuit targets; and show a
similar link between the \emph{pathwidth} of Boolean circuits and the width of
their representation in BDD targets. We demonstrate this link by showing
\emph{upper bound} results on compilation targets, to show that
bounded-width circuits can be compiled to circuits or BDD targets in linear
time and with singly exponential complexity in the width parameter. We also show
corresponding \emph{lower bound} results that establish that these compilation
targets must be exponential in the width parameters, already for a restricted class
of Boolean formulas. We now present our contributions and results in more
detail.

The first contribution of this paper (in Section~\ref{sec:kc}) is to give a
systematic picture of the 12 knowledge compilation circuit classes that we
investigate. We classify them along three independent axes:
\begin{itemize}
  \item \emph{Conjunction}: we distinguish between \emph{BDD classes}, such as OBDDs (ordered binary decision diagrams
    \cite{bryant1992symbolic}), where logical conjunction is only used to test
    the value of a variable and where computation follows a \emph{path} in the
    structure; and \emph{circuit classes} which allow decomposable
    conjunctions and where computation follows a \emph{tree}.
  \item \emph{Structuredness}: we distinguish between \emph{structured classes}, where the circuit or
    BDD always decomposes the variables along the same order or
    v-tree~\cite{pipatsrisawat2008new}, and \emph{unstructured classes} where no such
    restriction is imposed except \emph{decomposability} (each variable must be
    read at most once).
  \item \emph{Determinism}: we distinguish between classes that feature \emph{no
    disjunctions} beyond decision on a variable value (OBDDs, FBDDs, and dec-DNNFs), classes that feature
    \emph{unambiguous} or \emph{deterministic disjunctions} (uOBDDs, uFBDDs, and d-DNNFs), and
    classes that feature \emph{arbitrary disjunctions} (nOBDDs, nFBDDs, and
    DNNFs).
\end{itemize}
This landscape is summarized in Fig.~\ref{fig:kc-classes}, and we review known
translations and separation results that describe the relative expressive power
of these features.

The second contribution of this paper (in Sections~\ref{sec:result}
and~\ref{sec:proof}) is to show an \emph{upper bound} on the compilation of
bounded-treewidth classes to d-SDNNFs, and of bounded-pathwidth classes to OBDD
variants. For \emph{pathwidth}, existing work had
already studied the compilation of bounded-pathwidth circuits to
OBDDs~\cite[Corollary~2.13]{jha2012tractability}, which can be made constructive
\cite[Lemma~6.9]{amarilli2016tractable}.
Specifically, they show that a circuit of pathwidth $\leq k$ can be converted in polynomial time into an OBDD of width $\leq 2^{(k+2)2^{k+2}}$.
Our first contribution is to show that, by using \emph{unambiguous OBDDs}
(\emph{uOBDDs}), we can do the same but with linear time complexity, and with
the size of the uOBDD as well as its \emph{width} (in the classical knowledge
compilation sense) being singly exponential in the pathwidth. Specifically:

\begin{result}[see Theorem~\ref{thm:upper_bound_pw}]
\label{res:upper}
  Given as input a Boolean circuit $C$ of pathwidth~$k$ on $n$ variables,
  we can compute
  in time $O(|C| \times f(k))$
  a complete uOBDD equivalent to $C$ of width $\leq f(k)$
  and size $O(n \times f(k))$, where $f$ is singly exponential.
\end{result}

For \emph{treewidth}, we show that 
bounded-treewidth circuits can be compiled to the class of \emph{d-SDNNF
circuits}:
\begin{result}[see Corollary~\ref{cor:upper_bound}]
\label{res:upper2}
  Given as input a Boolean circuit $C$ of treewidth~$k$ on $n$ variables,
  we can compute in time $O(|C| \times f(k))$ a complete d-SDNNF equivalent to
  $C$ of width $\leq f(k)$ and size $O(n \times f(k))$, where $f$ is singly exponential.
\end{result}
The proof of Result~\ref{res:upper2}, and its variant that shows
Result~\ref{res:upper}, is quite challenging: we 
transform the input circuit bottom-up by
considering all possible valuations of the gates in
each bag of the tree decomposition, and keeping track of additional information to
remember which guessed values have been substantiated by a corresponding input.
Result~\ref{res:upper2} generalizes a recent theorem of Bova and
Szeider in~\cite{bova2017circuit} which we improve in two ways.
First, our result is constructive, whereas~\cite{bova2017circuit} only shows a bound on the
size of the d-SDNNF, without bounding the complexity of effectively computing
it. Second, our bound is singly exponential in~$k$,
whereas~\cite{bova2017circuit} is doubly exponential;
this allows us to 
be competitive with message passing (also singly exponential in~$k$),
and we believe it can be useful for practical applications.
We also explain how Result~\ref{res:upper2} implies the tractability
of several tasks on bounded-treewidth circuits, e.g.,
probabilistic query evaluation,
enumeration~\cite{amarilli2017circuit}, 
quantification~\cite{capelli2019tractable},
MAP
inference~\cite{fierens2015inference}, etc.

The third contribution of this paper is to show \emph{lower bounds} on how efficiently we
can convert from width-based classes to the compilation targets that we study.
Our bounds already apply to a weaker formalism of width-based circuits, namely,
monotone formulas in CNF (conjunctive normal form) or DNF (disjunctive normal
form). Our first two bounds (in Sections~\ref{sec:set} and~\ref{sec:structured})
are shown for \emph{structured} compilation targets, i.e., OBDDs, where we
follow a fixed order on variables, and SDNNFs, where we follow a fixed v-tree;
and they apply to arbitrary monotone CNFs and DNFs. The first lower bound 
concerns pathwidth and OBDD representations: we show that,
up to factors in the formula arity (maximal size of clauses) and degree (maximal number
of variable occurrences), any OBDD for a monotone CNF or DNF must
be of width exponential in the pathwidth $\pw(\phi)$ of the
formula~$\phi$. Formally:

\begin{result}[Corollary~\ref{cor:obdd_omega}]
  \label{res:lower_obdd}
  For any monotone CNF $\phi$ (resp., monotone DNF~$\phi$) of constant arity and
  degree, the size of the
  smallest nOBDD (resp., uOBDD) computing~$\phi$ is $2^{\Omega(\pw(\phi))}$.
\end{result}
This result generalizes several existing lower bounds in knowledge
compilation that 
exponentially separate CNFs from OBDDs,
such as~\cite{devadas93comparing} and~\cite[Theorem~19]{bova2017compiling}. 

Our second lower bound shows the analogue of Result~\ref{res:lower_obdd}
for the treewidth $\tw(\phi)$ of the formula~$\phi$ and (d-)SDNNFs:
\begin{result}[Corollary~\ref{cor:SDNNF_omega}]
  \label{res:lower_dsdnnf}
  For any monotone CNF $\phi$ (resp., monotone DNF~$\phi$) of constant arity and
  degree, the size of the
  smallest SDNNF (resp., d-SDNNF) computing~$\phi$ is $2^{\Omega(\tw(\phi))}$.
\end{result}

These two lower bounds contribute to a vast landscape of knowledge
compilation results
giving lower bounds on compiling 
specific Boolean functions to
restricted
circuits classes, e.g.,
\cite{devadas93comparing,razgon2014obdds,bova2017compiling} to OBDDs,
\cite{cali2017non} to dec-SDNNF,
\cite{beame2015new} to \emph{sentential decision diagrams} (SDDs),
\cite{pipatsrisawat2010lower,bova2016knowledge} to d-SDNNF,
\cite{bova2016knowledge,capelli2016structural,capelli2017understanding} to d-DNNFs and DNNFs.
However, all those lower bounds (with the exception of some results
in~\cite{capelli2016structural,capelli2017understanding}) apply to
well-chosen families of Boolean functions (usually CNF),
whereas Result~\ref{res:lower_obdd} and~\ref{res:lower_dsdnnf} 
apply to \emph{any} monotone CNF and DNF. Together with
Result~\ref{res:upper},
these generic lower bounds point to a strong relationship
between width parameters and structure representations, on monotone CNFs and
DNFs of constant arity and degree. Specifically, the smallest width of OBDD
representations of any such formula $\phi$ is in~$2^{\Theta(\pw(\phi))}$, i.e.,
precisely singly exponential in the pathwidth; and
an analogous bound applies to d-SDNNF size and treewidth of DNFs.

To prove these two lower bounds, we leverage known results from knowledge
compilation and communication complexity~\cite{bova2016knowledge} (in
Section~\ref{sec:set}) of which we give a unified presentation. Specifically, we
show that Boolean functions captured by uOBDDs (resp., nOBDDs) and d-SDNNF (resp.,
SDNNF) variants can be represented via a small cover (resp., disjoint cover) of
so-called \emph{rectangles}. We also show two Boolean functions (\emph{set
covering} and \emph{set intersection}) which are known not to have any such
covers. We then bootstrap the lower bounds on these two functions to a general
lower bound in Section~\ref{sec:structured}, by rephrasing pathwidth and
treewidth to new notions of \emph{pathsplitwidth} and \emph{treesplitwidth},
which intuitively measure the performance of a variable ordering or v-tree. We
then show that, for DNFs and CNFs with a high pathsplitwidth (resp.,
treesplitwidth), we can find the corresponding hard function ``within'' the CNF
or DNF, and establish hardness.

Our last lower bound result is shown in Section~\ref{sec:unstructured},
where we lift the assumption that the compilation targets are structured:

\begin{result}[Corollary~\ref{cor:unstructured}]
  \label{res:lower_fbdd}
  For any monotone CNF $\phi$ of constant arity and degree, the size of the
  smallest nFBDD computing $\phi$ is $2^{\Omega(\pw(\phi))}$, and the size of the
  smallest DNNF computing $\phi$ is $2^{\Omega(\tw(\phi))}$.
\end{result}

This result generalizes Result~\ref{res:lower_obdd} and~\ref{res:lower_dsdnnf}
by lifting the structuredness assumption, but they only apply to CNFs (and not to
DNFs).
The proof of these results reuses the notions of pathsplitwidth and
treesplitwidth, along with a more involved combinatorial argument on the size of
rectangle covers.

The current article extends the conference article~\cite{amarilli2018connecting}
in many ways:
\begin{itemize}
  \item We added Section~\ref{sec:kc} which gives a systematic presentation of
    knowledge compilation classes and reviews known results that relate them.
  \item In Section~\ref{sec:result}, the upper bound result on uOBDDs
    (Result~\ref{res:upper}) was added, the results were rephrased in terms of
    width, and the size of the circuit has been improved\footnote{This
    observation is due to Stefan Mengel and is adapted from the recent article
    \cite{capelli2019tractable}.} to be linear in the
    number of variables like in~\cite{bova2017circuit}.
  \item The presentation of the lower bounds in Sections~\ref{sec:set}
    and~\ref{sec:structured} was restructured to clarify the connection with
    communication complexity. The lower bounds on OBDDs
    (Result~\ref{res:lower_obdd}) was extended to uOBDDs and nOBDDs.
  \item The bounds on unstructured representations in
    Section~\ref{sec:unstructured} are new.
  \item We include full proofs for all results.
\end{itemize}

%% file: preliminaries.tex
We give preliminaries on trees, hypergraphs, treewidth, and Boolean functions.

\myparagraph{Graphs, trees, and DAGs.} We use the standard notions of
directed and undirected graphs, of paths in a graph, and of cycles. All
graphs considered in the paper are finite.

A \emph{tree}~$T$ is an undirected graph that has no cycles and that is
connected (i.e., there exists exactly one path between any two different
nodes). Its \emph{size}~$|T|$ is its number of edges. A tree $T$ is
\emph{rooted} if it has a distinguished node $r$ called the \emph{root}
of $T$. Given two adjacent nodes $n_1, n_2$ of a rooted tree $T$ with
root $r$, if $n_1$ lies on the (unique) path from $r$ to $n_2$, we say
that $n_1$ is the \emph{parent} of $n_2$ and that $n_2$ is a \emph{child}
of $n_1$. A \emph{leaf} of $T$ is a node that has no child, and an
\emph{internal node} of $T$ is a node that is not a leaf. Given a set $U$
of nodes of~$T$, we denote the set of leaves of $U$ by $\LEAVES(U)$. A
node~$n'$ is a \emph{descendant} of a node $n$ in a rooted tree if $n
\neq n'$ and $n$ lies on the path from $n'$ to the root. For $n \in T$,
we denote by $T_n$ the subtree of~$T$ rooted at~$n$. A rooted tree is
\emph{binary} if all nodes have at most two children, and it is
\emph{full} if all internal nodes have exactly two children.
A rooted full binary tree is called \emph{right-linear} if the children of each internal node are ordered (we then talk of a \emph{left} or \emph{right child}), and if every right child is a leaf.

A \emph{directed acyclic graph} (or \emph{DAG})~$D$ is a directed graph
that has no cycles. A DAG~$D$ is \emph{rooted} if it has a distinguished
node~$r$ such that there is a path from~$r$ to every node in~$D$. A
\emph{leaf}
of~$D$ is a node that has no child.

\myparagraph{Hypergraphs, treewidth, pathwidth.}
A \emph{hypergraph} $H=(V,E)$ consists of a finite set of \emph{nodes}
(or \emph{vertices}) $V$ and of a set $E$ of \emph{hyperedges} (or simply
\emph{edges}) which are non-empty subsets of~$V$.
We always assume that hypergraphs have at least one edge.
For a node~$v$ of~$H$, we write $E(v)$
for the set of edges of~$H$ that contain $v$.
The \emph{arity} of~$H$, written
$\arity(H)$,
is the maximal size of an edge of~$H$.
The \emph{degree} of~$H$,
written $\degree(H)$,
is the maximal number of edges to which a
vertex belongs, i.e., $\max_{v\in V} \card{E(v)}$.

A \emph{tree decomposition} of a hypergraph $H = (V, E)$
is a rooted tree $T$, whose nodes~$b$ (called
\emph{bags}) are labeled by a subset $\dom(b)$ of~$V$,
and which satisfies:
\begin{enumerate}[(i)]
  \item for every
hyperedge $e \in E$, there is a bag $b \in T$ with $e \subseteq
\dom(b)$; 
\item for all $v \in V$, the set of bags $\{b \in T \mid v \in
  \dom(b)\}$ is a connected subtree of~$T$.
\end{enumerate}
For brevity, we often identify a bag $b$ with its domain $\dom(b)$.
The \emph{width} of~$T$ is $\max_{b\in T} \card{\dom(b)}-1$.
The \emph{treewidth} of~$H$, denoted $\tw(H)$, is the minimal width of a tree decomposition of~$H$.
Pathwidth (denoted $\pw(H)$) is defined similarly 
but with \emph{path decompositions}, tree decompositions where all nodes
have at most one child.

It is NP-hard to determine the treewidth of a hypergraph~$(V,E)$,
but 
we can compute a tree decomposition in linear time
when parametrizing by the treewidth. This can be done in time
$O(|V|\times 2^{(32+\epsilon)k^3})$ with the classical result
of~\cite{bodlaender1996linear}, or, using a recent algorithm by Bodlaender
et al.~\cite{bodlaender2016approximation}, in time $O(|V|\times2^{ck})$:
\begin{theorem}[\cite{bodlaender2016approximation}]
	\label{thm:bod-precise}
        There exists a constant $c$ such that, given a hypergraph
        $H=(V,E)$ and an integer $k \in \NN$, we can check in time
	$O(|V| \times 2^{ck})$ whether $\tw(H) \leq k$, and if yes
        output a tree decomposition of $H$ of width $\leq 5k+4$.
\end{theorem}

For simplicity, we will often assume that a tree decomposition is $v$-\emph{friendly}, for a node $v \in V$, meaning that:
	\begin{enumerate}
	      \item 
                it is a full binary tree, i.e., each node has exactly zero or two children;
              \item
                for every internal bag $b$ with children $b_l,b_r$ we have $b \subseteq b_l \cup b_r$;
              \item
                for every leaf bag $b$ we have $|b| \leq 1$;
	      \item the root bag of $T$ only contains the node $v$.
        \end{enumerate}

Assuming a tree decomposition to be $v$-friendly for a fixed $v$ can be
done without loss of generality:        

\begin{lemma}
	\label{lem:friendly}
	Given a tree decomposition $T$ of a hypergraph $H$ of width $k$ and a node $v$ of $H$, we can
        compute in time $O(k\times |T|)$ a $v$-friendly tree decomposition $T'$ of $H$ of width
        $k$.
\end{lemma}
\begin{proof}
	We first create a bag $b_\root$ containing only the node $v$, and make this bag the root of $T$ by connecting it to a bag of $T$ that contains $v$ (if there is no such bag then we connect $b_\root$ to 
	an arbitrary bag of $T$).
  Then, we make the tree decomposition binary (but not necessarily full) by
  replacing each bag $b$ with children $b_1, \ldots, b_n$ with $n > 2$ by a
  chain of bags with the same label as~$b$ to which we attach the children
  $b_1, \ldots, b_n$. This process is in time $O(|T|)$ and does not change the
  width.

  We then ensure the second and third conditions, by applying a transformation
  to leaf bags and to internal bags. We first modify every leaf bag $b$ containing
  more than one vertex by a chain of at most $k$ internal bags with leaves where
  the vertices are added one after the other. Then, we modify every
  internal bag $b$ that contains elements $v_1, \ldots, v_n$ not present in the
  union $D$ of its children: we replace $b$ by a chain of at most $k$ internal bags
  $b_1', \ldots, b_n'$ containing respectively $b, b \setminus \{v_n\}, b
  \setminus \{v_n, v_{n-1}\}, \ldots, D$, each bag having a child introducing
  the corresponding gate $v_i$. This is in time $O(k \times |T|)$, and again
  it does not change the width; further, the result of the process is a tree decomposition that satisfies the
  second, third and fourth conditions and is still a binary tree.

  The only missing part is to ensure that the tree decomposition is full, which
  we can simply do in linear time by adding bags with an empty label as a
  second children for internal nodes that have only one child. This is obviously
  in linear time, does not change the width, and does not affect the other
  conditions, concluding the proof.
\end{proof}

\myparagraph{Boolean functions.}
A (Boolean) \emph{valuation} of a set $V$ is a function $\nu: V \to \{0,
1\}$, which can also be seen as the set of elements of $V$ mapped to~$1$.
A \emph{Boolean function}~$\phi$ on variables~$V$ is a mapping
$\phi:2^V\to\{0,1\}$ that associates
to each valuation $\nu$ of~$V$ a Boolean value $\phi(\nu)$ in $\{0, 1\}$ called the
\emph{evaluation} of~$\phi$ according to~$\nu$.
We write $\#\phi$ the number of satisfying valuations of $\phi$.
Given two Boolean functions $\phi_1,\phi_2$, we write $\phi_1 \Rightarrow \phi_2$ when every satisfying valuation of $\phi$ also satisfies $\phi_2$.
We write $\bot$ the Boolean function that maps every valuation to~$0$.

Let $X,Y$ be two disjoint sets, $\nu$ a valuation on $X$ and $\nu'$ a valuation on $Y$. We denote by $\nu \cup \nu'$ the valuation on $X \cup Y$ such that $(\nu \cup \nu')(a)$ is $\nu(a)$ if $a \in X$
and is $\nu'(a)$ if $a\in Y$.
Let $\phi$ be a Boolean function on $V$, and $\nu$ be a valuation on a set $X \subseteq V$.
We denote by $\nu(\phi)$ the Boolean function on variables $V \setminus X$ such that, for any valuation $\nu'$ of $V \setminus X$, $\nu(\phi)(\nu') = \phi(\nu \cup \nu')$.
When $\nu$ is a Boolean valuation on $V$ and $X \subseteq V$, we denote by $\nu|_X$ the Boolean valuation on $X$ defined by $\nu|_X(x) = \nu(x)$ for all $x$ in $X$.

Two simple formalisms for representing Boolean functions are 
\emph{Boolean circuits} and formulas in \emph{conjunctive normal form}
or \emph{disjunctive normal form}. We will discuss more elaborate
formalisms, namely binary decision diagrams and decomposable normal
negation forms, in Section~\ref{sec:kc}.

\myparagraph{Boolean circuits.}
A \emph{(Boolean) circuit} $C = (G, W, g_\out, \mu)$ is a DAG
$(G, W)$ whose vertices~$G$ are called \emph{gates}, whose edges $W$ are called
\emph{wires},
where $g_\out \in G$
is the \emph{output gate},
and where each gate
$g \in G$ has a \emph{type} $\mu(g)$ among $\var$ (a \emph{variable
gate}), $\NOT$,
$\lor$, $\land$.
The \emph{inputs}
of a gate $g \in G$ is the set $W(g)$ of gates $g' \in G$
such that $(g', g) \in W$; the \emph{fan-in} of~$g$ is its number of inputs.
We require $\NOT$-gates to have fan-in~1 and
$\var$-gates to have fan-in~0.
The \emph{treewidth} of~$C$ is that of the hypergraph~$(G,
W')$, where $W'$ is $\{\{g,g'\} \mid (g,g') \in W\}$.
Its \emph{size} $|C|$ is the number of wires.
The set $C_\var$ of \emph{variable gates} of~$C$ are those of type~$\var$.
Given a valuation $\nu$ of~$C_\var$, we extend it to an
\emph{evaluation} of~$C$ by mapping each variable $g \in C_\var$
to~$\nu(g)$, and evaluating the other gates according to their type.
We recall the convention that $\land$-gates (resp., $\lor$-gates) with no input evaluate to $1$ (resp., $0$).
The Boolean function on~$C_\var$ \emph{captured} (or \emph{computed}, or \emph{represented}) by the circuit
is the one that maps~$\nu$ to the evaluation of~$g_\out$ under~$\nu$.
Two circuits are \emph{equivalent} if they capture the same function. 

\myparagraph{DNFs and CNFs.}
We also study other representations of Boolean functions, namely,
Boolean formulas in \emph{conjunctive normal form} (\emph{CNFs}) and in
\emph{disjunctive normal form} (\emph{DNFs}). 
A CNF (resp., DNF) $\phi$ on a set of variables $V$ is a conjunction (resp.,
disjunction) of
\emph{clauses},
each of which
is a disjunction (resp., conjunction)  of \emph{literals} on~$V$, i.e.,
variables of~$V$ (a \emph{positive} literal) or their negation
(a \emph{negative} literal).

A \emph{monotone CNF} (resp., monotone DNF) is one where all literals are
positive, in which case
we often identify a clause to the set of variables
that it contains.
We always assume that monotone CNFs and monotone DNFs are
\emph{minimized}, i.e., no clause 
is a subset
of another.
This ensures that every monotone Boolean function has a unique
representation as a monotone CNF (the disjunction of its prime
implicants), and likewise for monotone DNF.
In addition, when we consider a valuation $\nu$ of a subset $X$ of the
variables of a monotone CNF/DNF $\phi$, we always take the Boolean
function $\nu(\phi)$ to be the equivalent minimized monotone CNF/DNF.

We assume that monotone CNFs and DNFs always contain at least
one non-empty clause (in particular, they cannot represent constant functions).
Monotone CNFs and DNFs~$\phi$ are isomorphic to hypergraphs:
the vertices are the variables of~$\phi$, and the hyperedges are the clauses
of~$\phi$.
We often identify~$\phi$ with its hypergraph.
In particular, the \emph{pathwidth} $pw(\phi)$
and \emph{treewidth} $\tw(\phi)$ of~$\phi$, and its \emph{arity} and \emph{degree},
are defined as that of its hypergraph.

Observe that there is a connection between the treewidth and pathwidth of a
monotone CNF or DNF with the treewidth and pathwidth of the natural Boolean circuit computing
it. Namely:

\begin{observation}
  For any monotone CNF or DNF formula $\phi$, there is a circuit~$C_\phi$
  capturing~$\phi$ such that $\tw(C_\phi) \leq \tw(\phi)+2$
  and $\pw(C_\phi) \leq \pw(\phi)+2$.
\end{observation}

\begin{proof}
  Fix $\phi$ and define $C_\phi$ as follows if $\phi$ is a CNF: $C_\phi$ has one input gate $i_x$
  for each variable $x$ of $\phi$, one $\lor$-gate $v_K$ for each clause $K$ of
  $\phi$ whose inputs are the $i_x$ for each variable $x$ of~$K$, and one $\land$-gate $o$ 
  which is the output gate of $C_\phi$ and whose inputs are all the $v_K$. If
  $\phi$ is a DNF we replace $\lor$-gates by~$\land$-gates and vice-versa.

We claim that $\tw(C_\phi) \leq \tw(\phi)+2$. Indeed, given a tree
decomposition $T$ of the hypergraph of $\phi$ of width $\tw(\phi)$, we can
construct a tree decomposition $T'$ of $C_\phi$ as follows. First, we replace each bag
$b$ of $T$ in $T'$ by $b' = \{i_x \mid x \in b\} \cup \{o\}$. Second, for every
  clause $K$, we consider a bag $\beta(K)$ of~$T$ containing all variables of~$K$ (which must
  exist because $T$ is a tree decomposition of~$\phi$): we modify~$T'$ to
  make~$\beta$ injective by
  creating sufficiently many copies of each bag in the image of~$\beta$ and
  connecting them to the original bag. Third, we add $v_K$ to the domain of~$\beta(K)$ for each clause~$K$.

The connectedness of $T'$ follows by the connectedness of $T$ and by the fact
that for each clause $K$, $v_K$ appears only once in $T'$ and $o$ appears in
every bag of $T'$. Moreover, every edge of $C_\phi$ is covered by some bag of
$T'$: indeed, both the edges of the form $(o,v_K)$ and the edges of the form
$(v_K, i_x)$ are covered in the (only) bag of $T'$ containing $v_K$. It is then
immediate that~$T'$ has the prescribed width.

  The proof for pathwidth is the same, because if $T$ is a path then $T'$
  can also be constructed to be a path.
\end{proof}

Note that there is no obvious converse to this result, e.g., the family of
single-clause CNFs $x_1 \lor \cdots \lor x_n$ has unbounded treewidth but can be
captured by a circuit of treewidth~$1$. A finer connection can be made by
considering the \emph{incidence treewidth}~\cite{Szeider04b} of CNFs and DNFs,
but we do not investigate this in this work.

%% file: kc.tex
\input{figure}

We now review some representation formalisms for Boolean functions that
are used in knowledge compilation, based either on \emph{binary decision
diagrams} (also known as \emph{branching programs}
\cite{wegener2000branching}) or on Boolean circuits in \emph{decomposable
negation normal form} \cite{DBLP:journals/jacm/Darwiche01}; in the rest of the
paper
we will study
translations between
bounded-width Boolean circuits and these classes. The classes that we consider have all
been introduced in the literature (see, in particular,
\cite{DBLP:journals/jair/DarwicheM02} for the main ones) but we sometimes
give slightly different (but equivalent) definitions in order to see them
in a common framework. An element $C$ of a knowledge compilation class
$\mathcal{C}$ is associated with its \emph{size} $|C|$ (describing how compact it
is) and with the Boolean function $\phi$ that it \emph{captures}.

A summary of the classes considered is provided in
Fig.~\ref{fig:kc-classes}.
This figure also shows (with arrows) when a class can be compiled into another in
polynomial-time (i.e., when one can transform an element $C$ of class
$\mathcal{C}$ capturing a Boolean function~$\phi$ into $C'$ of class
$\mathcal{C}'$ capturing~$\phi$, in time polynomial in
$|C|$). All classes shown in Fig.~\ref{fig:kc-classes} are
unconditionally \emph{separated} and some (cf.\ double arrows) are \emph{exponentially separated}. 
Specifically, we say that a class $\mathcal{C}$ is
separated (resp., exponentially separated) from a class $\mathcal{C}'$ if
there exists a family of Boolean functions $(\phi_k)$ captured by
elements $(C_k)$ of class
$\mathcal{C}$ such that all families of class~$\mathcal{C}'$
capturing~$(\phi_k)$ have size $\Omega(|C_k|^\alpha)$ for all
$\alpha\geq 0$ (resp., of size $\Omega(2^{|C_k|^\alpha})$ for some
$\alpha>0$).

We first consider general classes, then \emph{structured} variants of
these classes, and further introduce the notion of \emph{width} of these
structured classes. When introducing classes of interest, we recall or
prove non-trivial polynomial-time compilation and separation results related to
that class.

\subsection{Unstructured Classes}

We start by defining general, unstructured classes, i.e., those in the
background of Fig.~\ref{fig:kc-classes}, namely (non-deterministic) 
free binary decision diagrams
and circuits in decomposable negation normal form.

\subsubsection{Free Binary Decision Diagrams}

A \emph{non-deterministic binary decision diagram} (or \emph{nBDD}) on a
set of variables $V = \{v_1, \ldots, v_n\}$ is a rooted DAG~$D$
with labels on edges and nodes, verifying:

\begin{enumerate}[(i)]
\item there are exactly two leaves (also called \emph{sinks}), one being
  labeled by~$0$ (the \emph{$0$-sink}), the other one by~$1$ (the
    \emph{$1$-sink});

\item internal nodes are labeled either by $\lor$ or by a variable
  of~$V$;

\item each internal node that is labeled by a variable has two outgoing edges, labeled $0$ and $1$.
\end{enumerate}

The \emph{size} $\card{D}$ of~$D$ is its number of edges.
Let $\nu$ be a valuation of $V$, and let~$\pi$ be a path in $D$ from the root to a sink of~$D$. 
	We say that $\pi$ is \emph{compatible} with~$\nu$ if for every node~$n$
        of $\pi$ that is labeled by a variable $x$ of $V$, the path $\pi$ goes through the outgoing edge of $n$ labeled by $\nu(x)$.
	Observe that multiple paths might be compatible with~$\nu$, because no condition is imposed on nodes of the path that are labeled by $\lor$; in other words,
	the behaviour at $\lor$ nodes is non-deterministic.
	An nBDD $D$ \emph{captures} a Boolean function $\phi$ on~$V$ defined as follows: 
        for every valuation~$\nu$ of $V$, if there exists a path $\pi$ from the root to the $1$-sink that is compatible with $\nu$, then $\phi(\nu)=1$, else $\phi(\nu)=0$.
        An nBDD is \emph{unambiguous} when, for every valuation~$\nu$, there exists at most one path $\phi$ from the root to the $1$-sink that is compatible with $\nu$.
        A~\emph{BDD} is an nBDD that has no $\lor$-nodes.

	The most general form of nBDDs that we will consider in this
        paper are \emph{non-deterministic free binary decision
        diagrams} (\emph{nFBDDs}): they are nBDDs such that for every path from the root to a leaf, no two nodes of that path are labeled by the same variable.
        In addition to the nFBDD class, we will also study the class
        \emph{uFBDD} of unambiguous nFBDDs, and the class \emph{FBDD} of nFBDDs
        having no $\lor$-nodes.
\begin{proposition}\label{prp:separation-fbdd}
  nFBDDs are exponentially separated from uFBDDs, and uFBDDs are
  exponentially separated from FBDDs.
\end{proposition}

\begin{proof}
  The exponential separation between nFBDDs and uFBDDs is shown
  in~\cite{bova2016knowledge}: Proposition~7 of~\cite{bova2016knowledge} shows that there
  exists an nFBDD of size $O(n^2)$ for the Sauerhoff function~\cite{Sauerhoff03}
  over $n^2$ variables, while Theorem~9
  of~\cite{bova2016knowledge}, relying on
  \cite[Theorem 4.10]{Sauerhoff03}, shows that any representation of this
  function as a d-DNNF (a formalism
  that generalizes uFBDD, see our Proposition~\ref{prp:fbdd-to-dnnf-fw})
  necessarily has size $2^{\Omega(n)}$.

  To separate uFBDDs from FBDDs, we rely on the proof
  of exponential separation of PBDDs and
  FBDDs in~\cite[Theorem 11]{bollig1997complexity} (see also \cite[Theorem
    10.4.7]{wegener2000branching}). Consider the Boolean function
  $\phi_n$ on $n^2$ variables that tests whether,
  in an~$n\times n$ Boolean matrix, either the number of $1$'s is odd and
  there is a row full of $1$'s, or the number of $1$'s is even and there
  is a column full of $1$'s. As shown in
  \cite{bollig1997complexity,wegener2000branching}, 
  an FBDD for $\phi_n$ has necessarily
  size $2^{\Omega(n^{1/2})}$. On the other hand, it is easy to construct an FBDD of
  size $O(n^2)$ to test if the number of $1$'s is odd and
    there is a row full of $1$'s (enumerating variables in row order),
    and to construct an FBDD of size $O(n^2)$ to test if the number of
    $1$'s is even and there is a column full of $1$'s (enumerating
    variables in column order). An uFBDD for~$\phi_n$ is then obtained by
    simply adding an $\lor$-gate joining these two FBDDs, since only one of these two functions can evaluate to~$1$ under a given
    valuation.
\end{proof}

\subsubsection{Decomposable Negation Normal Forms}
We say that a circuit $C$ is in \emph{negation normal form} (NNF)
if the inputs of $\NOT$-gates are always variable gates.
For a gate~$g$ in a Boolean circuit~$C$,
we write $\VARS(g)$ for the set of variable gates
that have a directed path to~$g$ in~$C$.
An $\land$-gate $g$ of $C$ is \emph{decomposable} if it has at most two inputs and if, in case it has two input gates
$g_1\neq g_2$, then we have $\VARS(g_1) \cap \VARS(g_2) = \emptyset$.
We call $C$ \emph{decomposable} if each $\land$-gate is.
We write DNNF for an NNF that is decomposable.
Some of our proofs will use the standard notion of a \emph{trace} in an NNF:

  \begin{definition}
    \label{def:trace}
	  Let $C$ be an DNNF and $g$ be a gate of $C$. A \emph{trace of $C$
          starting at $g$} is a set $\Xi$ of gates of~$C$ that is minimal by
          inclusion and where:
	\begin{itemize}
		\item We have $g \in \Xi$;
		\item If $g' \in \Xi$ and $g'$ is an $\land$-gate, then all
                  inputs of~$g'$ are in~$\Xi$, i.e., $W(g') \subseteq \Xi$;
		\item If $g' \in \Xi$ and $g'$ is an $\lor$-gate, then exactly one input of $g'$
                  is in $\Xi$;
		\item If $g' \in \Xi$ and $g'$ if a $\lnot$-gate with input
                  variable gate $g''$, then $g''$ is in $\Xi$.
	\end{itemize}
\end{definition}

Observe that a gate $g \in C$ is \emph{satisfiable} (i.e., there exists a valuation $\nu$ such that $g$ evaluates to $1$ under $\nu$) if and only if there exists a trace of $C$ starting at $g$.
Indeed, given such a trace, define the valuation $\nu$ that maps to $0$ all the
variables $x$ such that a $\lnot$-gate with input~$x$ is in $\Xi$, and to $1$ all the other
variables: this valuation clearly satisfies $g$, noting in particular that each
variable occurs at most once in~$\Xi$ thanks to decomposability.
Conversely, when $g$ is satisfiable, it is clear that one can obtain a trace
starting at $g$ whose literals (variable gates, and negations of the variables that are an input to a
$\lnot$-gate) evaluate to~$1$ under the witnessing valuation~$\nu$.
This means that we can check in linear time whether a DNNF is
\emph{satisfiable}, i.e., if it has an accepting valuation, by computing
bottom-up the set of gates at which a trace starts.

As we will later see, the tractability of satisfiability of DNNFs does not extend to some other
tasks (e.g., model counting or probability computation). For these tasks, a useful additional requirement on circuits is \emph{determinism}. An $\lor$-gate $g$ of $C$ is \emph{deterministic} if there
is no pair $g_1\neq g_2$ of input gates of~$g$ and valuation $\nu$ of $C_\var$ such
that 
$g_1$ and~$g_2$ both evaluate to~$1$ under~$\nu$.
A Boolean circuit is \emph{deterministic} if each $\lor$-gate is.
We write d-DNNF for an NNF that is both decomposable and deterministic.
Model counting and probability computation can be done in linear time for
d-DNNFs thanks to decomposability and determinism (in fact this does not even
use the restriction of being an NNF).

Observe that, while decomposability is a syntactical restriction that can be
checked in linear time, the determinism property is semantic, and it is
co-NP-complete to check if a given $\lor$-gate of a circuit is deterministic:
hardness comes from the fact that an arbitrary Boolean circuit $C$ is unsatisfiable iff $\lor(C,1)$ is deterministic.
This motivates the notion of \emph{decision} gates, which gives us a syntactic way
to impose determinism.
Formally, an $\lor$-gate is a \emph{decision} gate if it is of the form $(x \land
C_1) \lor (\lnot x \land C_2)$, for some variable $x$ and (generally non-disjoint) subcircuits $C_1,C_2$.
A \emph{dec-DNNF} is a DNNF where all $\lor$-gates are decision gates: it is in
particular a d-DNNF.

\begin{proposition}
  DNNFs are exponentially separated from d-DNNFs, and d-DNNFs are
  exponentially separated from dec-DNNFs.
\end{proposition}

\begin{proof}
  The exponential separation of DNNFs and d-DNNFs is in
  \cite[Proposition~7 and Theorem~9]{bova2016knowledge},
  by a similar argument to the proof of our
  Proposition~\ref{prp:separation-fbdd}.

  The exponential separation of d-DNNFs and dec-DNNFs is in
  \cite[Corollary 3.5]{beame2017exact}.
\end{proof}

\subsubsection{Connections between FBDDs and DNNFs}

We have presented our unstructured classes of decision diagrams (namely FBDDs,
uFBDDs, and nFBDDs), and of decomposable NNF circuits (dec-DNNF, d-DNNF, and
DNNF). We now discuss the relationship between these various classes. We first 
observe that nFBDDs (and their subtypes) can be compiled to DNNF (and their
subtypes):

\begin{proposition}\label{prp:fbdd-to-dnnf-fw}
  nFBDDs (resp., uFBDDs, FBDDs)
  can be compiled to
  DNNFs (resp., d-DNNFs, dec-DNNFs)
  in linear time.
\end{proposition}

\begin{proof}
We first describe the linear-time compilation of an nFBDD to a DNNF that
captures the same function: recursively rewrite every internal node $n$
labeled with variable $x$ by a circuit $(x\land D_0)\lor (\NOT x\land D_1)$,
where $D_0$ and $D_1$ are the (not necessarily disjoint) rewritings of the nodes
to which $n$ respectively had a 0-edge and a 1-edge. We note that the new
$\lor$-gate is a decision gate and the two $\land$-gates are
decomposable. Furthermore:
\begin{itemize}
  \item if $D$ is unambiguous, all $\lor$-gates in the rewriting are
    deterministic, so we obtain a d-DNNF;
  \item if $D$ is an FBDD, then $\lor$-gates are only introduced in the
    rewriting, so we obtain a dec-DNNF. \qedhere
\end{itemize}
\end{proof}

The proof above implies that nFBDDs (resp., uFBDDs, FBDDs) are the restriction of
DNNFs (resp., d-DNNFs, dec-DNNFs) to the case where $\land$-gates, in
addition to being decomposable, are also all \emph{decision}
$\land$-gates, i.e., $\land$-gates appearing as children of a decision
$\lor$-gate.

Unlike previous compilation results, Proposition~\ref{prp:fbdd-to-dnnf-fw} does
not come with an exponential separation:
we can compile in the other direction at a \emph{quasi-polynomial} cost, i.e., in time
$2^{O((\log n)^\alpha)}$ for some fixed $\alpha>0$:

\begin{proposition}\label{prp:fbdd-to-dnnf-bw}
  DNNFs (resp., d-DNNFs, dec-DNNFs)
  can be compiled to
  nFBDDs (resp., uFBDDs, FBDDs)
  in quasi-polynomial time.
\end{proposition}

\begin{proof}
  Quasi-polynomial compilation has been first 
  shown for dec-DNNFs and FBDDs in in~\cite[Corollary 3.2]{beame2013lower}. This result was extended
  in~\cite[Section 5]{beame2015new} to the compilation of DNNFs to
  nFBDDs. Finally, in
  \cite[Proposition 1]{bollig2018relative}, it is shown that the same
  compilation yields a uFBDD when applied to a d-DNNF.
\end{proof}

We will see in Proposition~\ref{prp:bdd-separated-nnf}
that these quasi-polynomial time compilations cannot be made
polynomial-time, which will conclude the separation of all classes in the
background of Fig.~\ref{fig:kc-classes}. We now move to structured classes, that
are in the foreground of Fig.~\ref{fig:kc-classes}.

\subsection{Structured Classes}

The classes introduced so far are \emph{unstructured}: there is no particular
order or structure in the way variables appear within an nFBDD, or within
a DNNF circuit. In this section, we introduce \emph{structured} variants
of these classes, which impose additional constraints on how 
variables are used. Such additional restrictions often help with the
tractability of some operations: for example, given two FBDDs $F_1,F_2$
capturing Boolean functions $\phi_1,\phi_2$, it is NP-hard to decide if $\phi_1 \land \phi_2$ if satisfiable~\cite[Lemma 8.14]{meinel2012algorithms}. 
By contrast, with the \emph{ordered binary decision
diagrams}~\cite{bryant1991complexity} (OBDDs) that we now define, we can perform
this task tractably: given two OBDDs $O_1$ and $O_2$ that are ordered in the
same way, we can compute in polynomial time an OBDD representing $O_1 \land
O_2$, for which we can then decide satisfiability.
We first present OBDDs, and we then present \emph{SDNNFs} which are the structured
analogues of DNNF.

\subsubsection{Ordered Binary Decision Diagrams}
	A \emph{non-deterministic ordered binary decision diagram}
        (\emph{nOBDD})
        is an nFBDD~$O$ with a total order $\mathbf{v} = v_{i_1}, \ldots,
        v_{i_n}$ on the variables which \emph{structures}~$O$, i.e.,
        for every path $\pi$ from the root of~$O$ to a leaf,
        the sequence of variables which labels the internal nodes of~$\pi$
        (ignoring $\lor$-nodes)
        is a subsequence of~$\mathbf{v}$.
	We say that the nOBDD is \emph{structured by $\mathbf{v}$}.
        We also define \emph{uOBDDs} as the unambiguous nOBDDs, and 
        \emph{OBDDs} as the nOBDDs without any $\lor$-node.

        Like in the unstructured case (Proposition~\ref{prp:separation-fbdd}),
        these classes are exponentially separated:

\begin{proposition}
  nOBDDs are exponentially separated from uOBDDs, and uOBDDs are
  exponentially separated from OBDDs.
\end{proposition}

\begin{proof}
	The exponential separation between nOBDDs and uOBDDs will follow from our lower bounds on uOBDDs. 
	Indeed, Corollary~\ref{cor:obdd_omega} shows a lower bound on the size of uOBDDs representing bounded-degree and bounded-arity monotone DNFs of 
	high pathwidth. 
	But there exists a family of DNFs $(\phi_n)_{n \in \NN}$ of bounded degree and arity whose treewidth (hence pathwidth) is linear in their size: for instance, 
	DNFs built from \emph{expander graphs} (see 
\cite[Theorem~5 and Proposition~1]{grohe2009treewidth}). 
        Hence, for such a family $(\phi_n)_{n \in \NN}$ we have that any uOBDD for~$\phi_n$ is of size $2^{\Omega(|\phi_n|)}$.
	By contrast, it is easy to see that any DNF $\phi$ can be represented as
        an nOBDD in linear time. To do so, fix an arbitrary variable order $\mathbf{v}$ of the variables of $\phi$.
	Any clause $K$ of $\phi$ can clearly be represented as a small OBDD with order $\mathbf{v}$. Taking the disjunction of all these OBDDs then yields an nOBDD equivalent to $\phi$ 
	of linear size.

	For the separation between uOBDDs and OBDDs, consider the \emph{Hidden Weighted Bit} function $\mathrm{HWB}_n$ on variables $V = \{x_1,\ldots,x_n\}$, defined for a valuation $\nu$ of $V$ by:
	\[\mathrm{HWB}_n(\nu) \defeq \begin{cases}
		0 & \text{if } \sum_{i=1}^n \nu(x_i) = 0; \\
		\nu(x_k) & \text{if } \sum_{i=1}^n \nu(x_i) = k \neq 0.
	\end{cases}
		\]
	Bryant~\cite{bryant1991complexity} showed that OBDDs for $\mathrm{HWB}_n$ have size $2^{\Omega(n)}$.
	By contrast, it is not too difficult to construct uOBDDs of
        polynomial size for $\mathrm{HWB}_n$.
        This was observed in~\cite[Theorem~10.2.1]{wegener2000branching} for
        nOBDDs, with a note
        \cite[Proof of Corollary~10.2.2]{wegener2000branching} that the
        constructed nBDDs are unambiguous. See also
	\cite[Theorem~3]{bova2016sdds},
	which covers the case of \emph{sentential decision diagrams} instead of
        uOBDDs.
\end{proof}

\subsubsection{Structured DNNFs}
For NNFs, as with BDDs, it is possible to introduce a notion of 
\emph{structuredness}, that goes beyond that of
decomposability. 

Here, the analogue of a total order of variables (that structured an OBDD) is what is called a \emph{v-tree}, which is a tree whose leaves correspond to variables. More formally, a~\emph{v-tree}~\cite{pipatsrisawat2008new} over a set~$V$ is a rooted full
  binary tree $T$ whose leaves
        are in bijection with $V$. We always identify each leaf with the associated
        element of~$V$. We will also use the notion of an \emph{extended v-tree}
        $T$~\cite{capelli2019tractable} over a set $V$, which is like a v-tree, except that
        there is only an injection between $V$ and $\LEAVES(T)$.
	That is, some leaves can correspond to no element
        of $V$: we call those leaves \emph{unlabeled} (and they can intuitively
        stand for constant gates in the circuit). 

	\begin{definition}
	A \emph{structured DNNF} (resp., \emph{extended
        structured DNNF}), denoted \emph{SDNNF} (resp.,
        \emph{extended SDNNF}), is a triple $(D,T,\rho)$
        consisting of:
		\begin{itemize}
			\item a DNNF $D$;
			\item a v-tree (resp., extended v-tree) $T$ over $D_\var$;
			\item a mapping $\rho$ labeling each $\land$-gate of $g$ with a node of~$T$ that satisfies the following:
				for every $\land$-gate $g$ of~$D$ with $1 \leq m \leq 2$ inputs $g_1, g_m$, the node $\rho(g)$ \emph{structures} $g$, i.e.,
        there exist $m$ distinct children $n_1, n_m$ of $\rho(g)$ such that 
         $\VARS(g_i) \subseteq \LEAVES(T_{n_i})$ for all $1 \leq i \leq m$.
		\end{itemize}
	\end{definition}

We also define d-SDNNF and dec-SDNNF as structured d-DNNF and dec-DNNF,
and define extended d-SDNNF and extended dec-SDNNF in the expected way.

As in the case of FBDDs and DNNFs, observe that an OBDD (resp., uOBDD, nOBDD) is a special type of dec-SDNNF (resp., d-SDNNF, SDNNF). 
Namely, the transformation described above
Proposition~\ref{prp:fbdd-to-dnnf-fw}, when applied to an OBDD (resp., uOBDD, nOBDD), yields a dec-SDNNF (resp., d-SDNNF, SDNNF) that is structured by a v-tree
that is right-linear (recall the definition from
Section~\ref{sec:preliminaries}). Hence, we have:

\begin{proposition}\label{prp:obdd-to-sdnnf}
  nOBDDs (resp., uOBDDs, OBDDs)
  can be compiled to
  SDNNFs (resp., d-SDNNFs, dec-SDNNFs)
  in linear time.
\end{proposition}

\begin{proof}
 Given the variable order $\mathbf{v} = v_1 \ldots v_n$ of
an nOBDD, we construct our right-linear v-tree $T$ as having a
root $r_1$, internal nodes $r_i$ with $r_i$ being the left child of
$r_{i-1}$ for $2 \leq i \leq n-1$, leaf nodes $v_1 \ldots v_{n-1}$ with
$v_i$ being the right child of $r_i$, and leaf node $v_n$ being the right
child of $r_{n-1}$. We then apply as-is the translation described in the proof
  of
  Proposition~\ref{prp:fbdd-to-dnnf-fw}.
\end{proof}

As in the unstructured case (Proposition~\ref{prp:fbdd-to-dnnf-fw}), there is no exponential separation
result: indeed, analogously to Proposition~\ref{prp:fbdd-to-dnnf-bw} in
the unstructured case, there exist quasi-polynomial compilations in the
other direction:

\begin{proposition}\label{prp:sdnnf-to-obdd}
  SDNNFs (resp., d-SDNNFs, dec-SDNNFs)
  can be compiled to
  nOBDDs (resp., uOBDDs, OBDDs)
  in quasi-polynomial time.
\end{proposition}

\begin{proof}
  Quasi-polynomial time compilation of a SDNNF into an nOBDD is proved in
  \cite[Theorem~2]{bollig2018relative}, by adapting the compilation of
  \cite{beame2015new} from DNNFs to nFBDDs. Furthermore,
  \cite[Proposition~2]{bollig2018relative} shows that the resulting nOBDD
  is unambiguous if the SDNNF is deterministic. But it is easy to see
  that the same compilation \cite[Simulation 2]{bollig2018relative} yields
  an OBDD if the input is a dec-SDNNF: indeed, in a dec-SDNNF there are
  no $\lor$-gates that are not decision gates, so no $\lor$-gates are
  produced in the output.
\end{proof}

\subsection{Comparing Structured and Unstructured Classes}
To obtain all remaining separations in Fig.~\ref{fig:kc-classes}, and justify
that no arrows are missing, we need two last results in which we will compare structured and
unstructured classes.

The first result describes the power of decomposable $\land$-gates as opposed to
decision gates: it shows that the least powerful class that has arbitrary
decomposable $\land$-gates (dec-SDNNF) cannot be compiled to the most
powerful class with decision $\land$-gates (nFBDD) without a
super-polynomial size increase.

\begin{proposition}
  \label{prp:bdd-separated-nnf}
  There exists a family of functions $(\phi_n)$ that has $O(n^2)$ dec-SDNNF
  but no nFBDD of size smaller than $n^{\Omega(\log(n))}$.
\end{proposition}

\begin{proof}
  In~\cite{razgon2014twnrop}, Razgon constructs for every $k$ a family
  of 2CNF $(\phi^k_n)_{n \in \mathbb{N}}$ such that $\phi_n^k$ has
  $n$ variables and treewidth $k$. He proves (\cite[Theorem
  1]{razgon2014twnrop}) a $n^{\Omega(k)}$ lower bound on the size of
  any nFBDD computing $\phi_n^k$ (Razgon refers to nFBDD as NROBP in
  his paper). It is known from~\cite[Section
  3]{DBLP:journals/jacm/Darwiche01} that one can compile any CNF with $n$
  variables and with treewidth $k$ into a dec-SDNNF of size
  $2^{\Omega(k)}n$. Thus, $\phi_n^k$ can be computed by a
  dec-SDNNF of size $2^{\Omega(k)}n$.

  Taking $k \colonequals \lfloor \log(n) \rfloor$ gives the desired separation:
  $\phi_n^k$ can be computed by a dec-SDNNF 
  of size $O(n^2)$ 
  but by no nFBDD of size smaller than $n^{\Omega(\log(n))}$.
\end{proof}

Proposition~\ref{prp:bdd-separated-nnf} implies that no
DNNF class in the upper level of Fig.~\ref{fig:kc-classes}
can be polynomially compiled into any BDD class in the
lower level of Fig.~\ref{fig:kc-classes}.

The second result describes the power of unstructured formalisms as opposed to
structured ones: it shows that the least powerful unstructured class (FBDD) cannot be compiled
to the most powerful structured class (SDNNF) in size less than
exponential.

\begin{proposition}
  \label{prp:structure-separation}
  FBDDs are exponentially separated from SDNNFs:
  there exists a family of functions
  $(\phi_n)$ that has FBDDs of size $O(n)$ but no SDNNF of size smaller than $2^{\Omega(n)}$.
\end{proposition}

\begin{proof}
	This separation was proved independently by Pipatsrisawat and Capelli in their PhD theses (see~\cite[Appendix D.2]{pipatsrisawat2010reasoning},
        and~\cite[Section 6.3]{capelli2016structural}).

	In his work, Pipatsrisawat considers the Boolean function
        \emph{circular bit shift} $CBS(S,X,Y)$: it is defined on a tuple
        $(S,X,Y)$ of variables with $N = s_1 \ldots s_k$, $X = x_1 \ldots
        x_{2^k}$, $Y = y_1 \ldots y_{2^k}$ for some $k \in \NN$, and it
	evaluates to $1$ on valuation~$\nu$ iff shifting the bits of $\nu(X)$ by
        $S$ (as written in binary) positions yields $\nu(Y)$.
        Pipatsrisawat shows that the CBS function
        on $n$~variables has an FBDD of
        size $O(n^2)$, but that any SDNNF for CBS has size $2^{\Omega(n)}$.

	The proof of Capelli uses techniques close to
	the ones used in Section~\ref{sec:structured}.
\end{proof}

Proposition~\ref{prp:structure-separation} implies that no
unstructured class (in the background of Fig.~\ref{fig:kc-classes})
can be polynomially compiled into any structured class (in the
foreground of Fig.~\ref{fig:kc-classes}).

Looking back at Fig.~\ref{fig:kc-classes}, we see that, indeed, all
classes are separated and no arrows are missing. The separation is exponential except when moving
(on the vertical axis in the figure) from
BDD-like classes to NNF-like classes, in which case we know (cf.
Propositions~\ref{prp:fbdd-to-dnnf-bw} and~\ref{prp:sdnnf-to-obdd}) that
quasi-polynomial compilations exist in the other direction.

\subsection{Completeness and Width}
\label{sec:completewidth}

Two last notions that will be useful for our results are the notions of \emph{completeness} and \emph{width} for structured classes.
Intuitively, completeness further restricts the structure of how variables are
tested in the circuit or BDD: in addition to the structuredness requirement, we
impose that no variables are ``skipped''. We will be able to assume completeness
without loss of generality, it will be guaranteed by our construction, and it
will be useful in our lower bound proofs.

On complete classes, we will additionally be able to define a notion of
\emph{width} that we will use to show finer lower bounds.

\myparagraph{Complete OBDDs.}
	An nOBDD $O$ on $V$ is \emph{complete} if every path from the root
        to a sink tests every variable of $V$.
	For $x \in V$, the {\em $x$-width} of a complete nOBDD
	$O$ is the number of nodes labeled with variable $x$.
	The \emph{width} of $O$ is $\max_{x \in V} x$-width of $O$.

	It is immediate that partially evaluating a complete nOBDD does not increase its width:
\begin{lemma}
	\label{lem:partial_evaluate_nOBDD}
	Let $O$ be a complete nOBDD (resp., uOBDD) on variables $V$, with order
        $\mathbf{v}$ and of width $\leq w$, and let $\phi$ be the Boolean function that $O$ captures. Let $X \subseteq V$, and 
	$\nu$ be a valuation of $X$. Then there exists a complete nOBDD (resp., uOBDD) $O'$, on variables $V \setminus X$, of order $\mathbf{v}|_{V \setminus X}$ and width $\leq w$, that computes
	$\nu(\phi)$.
\end{lemma}

\myparagraph{Complete SDNNFs.}
The notion of completeness and width of OBDDs extends naturally to SDNNFs.
Following~\cite{capelli2019tractable}, we say that a \mbox{(d-)}SDNNF (resp., extended
\mbox{(d-)}SDNNF) $(D,T,\rho)$ is \emph{complete} if $\rho$ labels \emph{every} gate of
$D$ (not just $\land$-gates) with a node of $T$ and the following conditions are
satisfied:
\begin{enumerate}
        \item The output gate of $D$ is an $\lor$-gate;
	\item For every variable gate $g$ of $D$, we have $\rho(g)=g$;
	\item For every $\NOT$-gate $g$ of $D$, letting $g'$ be the variable gate that feeds $g$, we have $\rho(g)=g'$; 
	\item For every $\lor$-gate $g$ of $D$, for any input $g'$ of $g$, the
          gate $g'$ is not an $\lor$-gate, and moreover we have $\rho(g') = \rho(g)$;
	\item For every $\land$-gate $g$ of $D$, for any input $g'$ of $g$, the
          gate $g'$ is an $\lor$-gate, and we have that $\rho(g')$ is a child of $\rho(g)$;
	\item For every $\land$-gate $g$ of $D$ and any two inputs $g' \neq g''$ of $g$, we have $\rho(g') \neq \rho(g'')$;
	\item For every $\land$-gate $g$ of $D$ such that $\rho(g)$ is an internal node of $T$, $g$ has exactly two inputs.
\end{enumerate}

For a node $n$ of $T$, the \emph{$n$-width} of a complete \mbox{(d-)}SDNNF (resp., extended complete) $(D,T,\rho)$ is the number of $\lor$-gates that are structured by $n$.
The \emph{width} of $(D,T,\rho)$ is the maximal $n$-width for a node of $T$.

One of the advantages of complete \mbox{(d-)}SDNNFs of bounded width is that we can
work with extended v-trees, and then compress their size in linear time, so that
the v-tree becomes non-extended and the size of the circuit is linear in the
number of variables. When doing so, the extended v-tree is modified in a way
that we call a \emph{reduction}:
\begin{definition}
	\label{def:vtree-refinement}
	Let $T$, $T'$ be two extended v-trees over variables $V$. We say that $T'$ is a \emph{reduction of $T$} if, for every internal node $n$ of $T$, there exists an internal node $n'$ of
	$T'$ such that $\LEAVES(T') \cap \LEAVES(T \setminus T_n) \subseteq \LEAVES(T' \setminus T'_{n'})$ and $\LEAVES(T') \cap \LEAVES(T_n) \subseteq \LEAVES(T'_{n'})$.
\end{definition}

We can now show how to compress extended complete \mbox{(d-)}SDNNFs:

\begin{lemma}[\cite{capelli2019tractable}]
	\label{lem:reduce_extended}
        Let $(D,T,\rho)$ be an extended complete \mbox{(d-)}SDNNF of width~$w$ on $n$
        variables. We can compute in linear time a complete \mbox{(d-)}SDNNF
        $(D',T',\rho')$ of width $w$ such that $T'$ is a reduction of $T$ and
        such that
	$|D'|$ is in $O(n \times w^2)$.
\end{lemma}
\begin{proof}
	We present a complete proof, inspired by the proof
        in~\cite[Lemma~4]{capelli2019tractable}.
	As a first prerequisite, we preprocess $D$ in linear time so that the number of $\land$-gates structured by a same node $n$ of $T$ is in $O(w^2)$.
	This can be done, as in~\cite[Observation~3]{capelli2019tractable}, by
        noticing that there can be at most $w^2$ inequivalent $\land$-gates that are structured by a node $n$.
	Indeed, this is clear if $n$ is a leaf, as such an $\land$-gate cannot
        have an input (so there is at most one inequivalent $\land$-gate). If
	$n$ is an internal node with children $n_1$ and $n_2$, then, by item (7) of the definition of being an extended complete SDNNF, every
	$\land$-gate structured by~$n$ has one input among the $\leq w$ $\lor$-gates structured by~$n_1$, and one input among
        the $\leq w$ $\lor$-gates structured by~$n_2$. Hence there are $w^2$ possible inequivalent $\land$-gates.
	We can then merge all the $\land$-gates that are equivalent, and obtain
        a complete \mbox{(d-)}SDNNF where for each node $n$ of the v-tree, at most
        $w^2$ $\land$-gates are structured by $n$.

	The second prerequisite is to  eliminate the gates that are not connected to the output of~$D$, and then to propagate the constants in the circuit (i.e., to evaluate it partially). 
	In other words, eliminate all gates (and their wires) that are not
        connected to the output of $D$, and then repeat the following until
        convergence:
	\begin{itemize}
		\item For every constant $1$-gate $g$ (i.e., an $\land$-gate with no input) and wire $g \rightarrow g'$, if $g'$ is an $\land$-gate then simply remove the wire $g \rightarrow g'$, and if $g'$ is an $\lor$-gate, then replace $g$ by a constant $1$-gate; then remove $g$ and all the wires connected to it.
	
		\item For every constant $0$-gate $g$ (i.e., an $\lor$-gate with
                  no input) and wire $g \rightarrow g'$, if $g'$ is an $\lor$-gate then simply remove the wire $g \rightarrow g'$, and if $g'$
			is an $\land$-gate, then replace $g$ by a constant $0$-gate; then remove $g$ and all the wires connected to it.
	\end{itemize}
	This again can be done in linear time (by a DFS traversal of the
        circuit, for instance), and it does not change the properties of the
        circuit or the captured function. Further, it ensures that
        $\lor$- and $\land$-gates of the
        resulting circuit always have at least one input, or that we get to one
        single constant gate (if the circuit captures a constant Boolean
        function): as this second case of constant functions is uninteresting, we assume
        that we are in the first case.
	We call the resulting circuit $C$. It is clear that $C$ is still structured by
	$T$ (by taking the restriction of~$\rho$ to the gates that have not been removed).

	Having enforced these prerequisites on~$C$, the idea is to eliminate unlabeled leaves $l$ in he v-tree one by one by merging the parent and the sibling of $l$. 
	Formally, whenever we can find in $T$ an unlabeled leaf $l$ with parent $n$ and sibling $n'$, we perform these two steps:
	\begin{enumerate}
		\item Remove from $T$ the leaf $l$ (and its parent edge)
                  noticing that no gate of~$C$ was structured by~$l$ because we
                  propagated the constants in the circuit in our second
                  preprocessing step; then replace the parent~$n$ in~$T$ by its
                  remaining child~$n'$ so that it is again binary and full.

		\item We now need to modify $C$ so that $C$ is an extended
                  complete \mbox{(d-)}SDNNF structured by the new v-tree. There is nothing to do in the case that $n'$ was an unlabeled leaf, because
			then no gate of $C$ was structured by $n'$, or even by $n$ (since we propagated constants). 
			In the case where $n'$ was a variable leaf or an internal node, then, for every $\lor$-gate $g$
			that was structured by~$n$, we compute the set $I_g$ of
                        gates $g'$ that were structured by~$n'$, that are not an $\lor$-gate, 
			and such that there is a path from $g'$ to $g$ in $C$.
                        Thanks to our first preprocessing step, the set~$I_g$
                        can be computed in time $O(w^2)$ as this bounds the
                        number of $\land$-gates
			structured by $n$ and by $n'$.
			Observe that gates in $I_g$ can be either $\land$-gates that were structured by $n'$ (in case $n'$ was an internal node), or  $\lnot$-gates or variable gates (in case $n'$ was a variable leaf).
			Now, remove from $C$ all the $\lor$-gates that were structured by $n'$, all the $\land$-gates that were structured by $n$, and all the edges connected to them. 
			For each $\lor$-gate $g$ that was structured by~$n$, set its new inputs
			to be all the gates in $I_g$. One can check that the resulting circuit captures the same function (this uses the fact that we propagated constants),
			and that determinism cannot be broken in case the original circuit was a d-SDNNF.
                        Moreover, $C$ is now an extended complete \mbox{(d-)}SDNNF structured by the new v-tree $T$.
	\end{enumerate}
	By iterating this process, we will end up with a v-tree $T'$ that is not
        extended, and the resulting circuit will be an equivalent complete
        \mbox{(d-)}SDNNF of width $\leq w$ and size $O(n \times w^2)$.
	The total time is linear since we spend $O(w^2)$ time to eliminate each single unlabeled leaf.
	Moreover it is clear that the final v-tree obtained is a reduction of
        the original v-tree, as the property is preserved by each elimination.
\end{proof}

Like for OBDDs (Lemma~\ref{lem:partial_evaluate_nOBDD}), we will use the fact
that partially evaluating a complete \mbox{(d-)}SDNNF cannot increase the width:
\begin{lemma}
	\label{lem:partial_evaluate_DNNF}
        Let $(D,T,\rho)$ be a complete \mbox{(d-)}SDNNF of width $\leq w$ over
        variables $V$, and let $\phi$ be the  Boolean function that~$D$ captures. Let $X \subseteq V$, and $\nu$ be a valuation of $X$.
        Then there exists a complete \mbox{(d-)}SDNNF $(D',T',\rho')$ of width $\leq w$ on variables $V \setminus X$ computing $\nu(\phi)$ such that $T'$ is a reduction of~$T$.
\end{lemma}
\begin{proof}
  We replace every leaf $l$ of $T$ that corresponds to a variable $x \in
  X$ by an unlabeled leaf, replace every variable gate $x$ in $D$ by a
  constant $\nu(x)$-gate, replace every $\lnot$-gate with input
  variable $x$ by a constant $(1-\nu(x))$-gate, and then propagate constants as in
  the second prerequisite in the proof of Lemma~\ref{lem:reduce_extended}.
        This yields an \emph{extended} complete \mbox{(d-)}SDNNF computing $\nu(\phi)$.
        We then conclude by applying Lemma~\ref{lem:reduce_extended}.
\end{proof}

\myparagraph{Making nOBDDs and SDNNFs complete.}
Imposing completeness on nOBDDs or SDNNFs is in fact not too restrictive, as we
can assume that OBDDs and SDNNFs are complete up to multiplying the size by
the number of variables:

\begin{lemma}
	\label{lem:complete_obdd_sdnnf}
        For any nOBDD (resp., SDNNF) $D$ on variables $V$, there exists an
        equivalent complete nOBDD (resp., SDNNF) of size at most $(\card{V}+1)
        \times \card{D}$.
\end{lemma}
\begin{proof}
  The result will follow from a more general completion result on unstructured
  classes given later in the paper (Lemma~\ref{lem:complete_circuit}); it is
  straightforward to observe that applying the constructions of that lemma yield
  structured outputs when the input representations are themselves structured.
\end{proof}

%% file: figure.tex
\begin{figure}[t]
  \centering
\footnotesize
\hspace*{-1cm}	\begin{tikzpicture}

	\coordinate (origin) at (0,0) {};

        \node[text width=8em,text centered] (and) at (0,3) {\textbf{decomposable\\$\pmb{\land}$-gates}};

        \coordinate (or) at (4.5,0) {};
        \path (or.north) ++(0,.02) node[anchor=south east] {\textbf{$\pmb{\lor}$-gates}};

	\node (structure) at (1.725,1.5) {\textbf{structured}};

	\draw[->] (origin) -- (and)
          node[pos=.85,auto] {arbitrary} node[pos=.15,auto] {decision};
	\draw[->] (origin) -- (or) 
          node[pos=.9,below] {arbitrary} 
          node[pos=.5,below] {deterministic}
          node[pos=.1,below] {decision};
        \draw[->] (origin) -- (structure)
          node[pos=.85,left] {no} node[pos=.4,left] {yes};

  \begin{scope}[xshift=-2.5cm,xscale=1.15]
	\node (000) at (7,0) {OBDD};
	\node (001) at (8,1) {FBDD};
	\node (010) at (7,2) {dec-SDNNF};
	\node (011) at (8,3) {dec-DNNF};
	\node (100) at (9,0) {uOBDD};
	\node (101) at (10,1) {uFBDD};
	\node (110) at (9,2) {d-SDNNF};
	\node (111) at (10,3) {d-DNNF};
	\node (200) at (11,0) {nOBDD};
	\node (201) at (12,1) {nFBDD};
	\node (210) at (11,2) {SDNNF};
	\node (211) at (12,3) {DNNF};

    \tikzset{every path/.style={->}}

\newcommand{\drawwithbg}[3][]{
    \draw[white,line width=3pt,opacity=.8] (#2) -- (#3);
    \draw[black,#1] (#2) -- (#3);
}

	\drawwithbg[->>]{000}{001}
	\drawwithbg{000}{010}
        \drawwithbg[->>]{010}{011}
	\drawwithbg{001}{011}
        \drawwithbg[->>]{000}{100}
	\drawwithbg[->>]{001}{101}
	\drawwithbg[->>]{010}{110}
	\drawwithbg[->>]{011}{111}
	\drawwithbg{100}{110}
        \drawwithbg[->>]{100}{101}
	\drawwithbg{101}{111}
        \drawwithbg[->>]{110}{111}
	\drawwithbg[->>]{100}{200}
	\drawwithbg[->>]{101}{201}
	\drawwithbg[->>]{110}{210}
	\drawwithbg[->>]{111}{211}
	\drawwithbg{200}{210}
        \drawwithbg[->>]{200}{201}
	\drawwithbg{201}{211}
        \drawwithbg[->>]{210}{211}
  \end{scope}
	\end{tikzpicture}
\caption{Dimensions of the knowledge compilation classes consider in this paper
(left), and diagram of the classes (right). Arrows indicate
polynomial-time compilation; all classes are separated and no arrows 
are missing (except those implied by transitivity). Double arrows
indicate that the classes are exponentially separated; when an arrow is
not double, quasi-polynomial compilation in the reverse direction exists.}
\label{fig:kc-classes}
\end{figure}
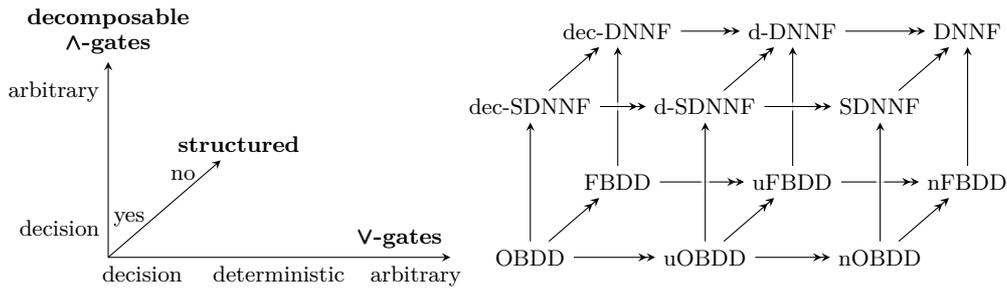

%% file: result.tex
In this section we study how to
compile Boolean circuits to d-SDNNFs (resp., uOBDDs), parameterized by the
treewidth (resp., pathwidth) of the input circuits. 
We first present our results in Section~\ref{sec:results_present}, then
show some examples of applications in Section~\ref{sec:applications},
before providing full proofs in Section~\ref{sec:proof}.

\subsection{Results}
\label{sec:results_present}

To present our upper bounds, we first review the 
independent result that was recently 
shown by Bova and Szeider~\cite{bova2017circuit} on compiling bounded-treewidth circuits to d-SDNNFs:

\begin{theorem}[{\cite[Theorem~3 and Equation~(22)]{bova2017circuit}}]
  \label{thm:bovaszeider}
  Given a Boolean circuit $C$ having \mbox{treewidth~$\leq k$}, there exists an
  equivalent d-SDNNF of size $O(f(k) \times |C_\var|)$, where $f$~is doubly
  exponential.
\end{theorem}

Their result has two drawbacks:
(i) it has a doubly exponential
dependency on the width; and (ii) it is nonconstructive, because
\cite{bova2017circuit} gives no time bound
on the computation, 
leaving open the
question of 
effectively compiling 
bounded-treewidth circuits to d-SDNNFs.
The nonconstructive aspect can easily be tackled by encoding in linear time the
input circuit $C$ into a relational instance of same treewidth, and then
use~\cite[Theorem~6.11]{amarilli2016tractable} to construct in linear time a
d-SDNNF representation of the provenance on $I$ of a fixed MSO formula
describing how to evaluate Boolean circuits
(see the conference
version of this paper~\cite{amarilli2018connecting} for more details).
This ``naïve'' approach computes a d-SDNNF in 
time $O(|C| \times f(k))$, but where $f$ is a superexponential function that does not address the first drawback. We show that we can get $f$ to be \emph{singly} exponential.

\myparagraph{Treewidth bound.}
Our main upper bound result addresses both drawbacks and shows that we can compile
in time linear in the circuit and singly exponential in the treewidth.
Our proof is independent from~\cite{bova2017circuit}.
Formally, we show:

\begin{theorem}
\label{thm:upper_bound}
  There exists a function $f(k)$ that is in $O(2^{(4+\epsilon)k})$ for any $\epsilon \geq 0$ such that the following holds.
  Given as input a Boolean circuit $C$ and tree decomposition $T$ of width $\leq k$ of $C$,
	we can compute a complete extended d-SDNNF equivalent to $C$ of width $\leq 2^{2 (k+1)}$
	in time $O\left(|T| \times f(k)\right)$. 
\end{theorem}

This result assumes that the tree decomposition is provided as input; but we can
instead use Theorem~\ref{thm:bod-precise} to obtain it. We can also apply 
Lemma~\ref{lem:reduce_extended} to the resulting circuit to get a proper
(non-extended) d-SDNNF and reduce its size so
that it only depends on the number of variables of the input circuit~$C$ (i.e.,
$\card{C_\var}$ rather than $\card{C}$),
which allows us to truly generalize Theorem~\ref{thm:bovaszeider}.
Putting all of this together, we get:

\begin{corollary}
	\label{cor:upper_bound}
There exists a constant $c \in \NN$ such that the following holds.
Given as input a Boolean circuit $C$ of treewidth $\leq k$,
	we can compute
	in time $O\left(|C| \times 2^{ck}\right)$
        a complete d-SDNNF equivalent to $C$ of width $O(2^{ck})$
        and size $O(\card{C_\var} \times 2^{ck})$.
\end{corollary}

However, Corollary~\ref{cor:upper_bound} is mainly of theoretical interest,
since the constant hidden in Theorem~\ref{thm:bod-precise} is huge.
In practice, one would first use a heuristic to compute a tree decomposition, and
then use our construction of Theorem~\ref{thm:upper_bound} on that decomposition.
We will prove Theorem~\ref{thm:upper_bound}
in Section~\ref{sec:proof}.

\myparagraph{Pathwidth bound.}
A by-product of our construction is that, in the special case where we start with a path decomposition, it turns out that the d-SDNNF computed is in fact an uOBDD.
The compilation of bounded-pathwidth Boolean circuits to OBDDs had already been studied in~\cite{jha2012tractability,amarilli2016tractable}: 
Corollary~2.13 of~\cite{jha2012tractability} shows that a circuit of pathwidth $\leq k$ has an equivalent OBDD of width $\leq 2^{(k+2)2^{k+2}}$, and \cite[Lemma~6.9]{amarilli2016tractable} justifies that the transformation can be made in polynomial time.
Our second upper bound result is that, by using uOBDDs instead of OBDDs, we can get a singly exponential dependency:

\begin{theorem}
\label{thm:upper_bound_pw}
  There exists a function $f(k)$ that is in $O(2^{(2+\epsilon)k})$ for any $\epsilon \geq 0$ such that the following holds.
  Given as input a Boolean circuit $C$ and path decomposition $P$ of width $\leq k$ of $C$,
	we can compute a complete uOBDD equivalent to $C$ of width $\leq 2^{2 (k+1)}$
	in time $O\left(|P| \times f(k)\right)$. 
\end{theorem}

While we do not know if the doubly exponential dependence on $k$ in~\cite{jha2012tractability} is tight for OBDDs, we will show in Section~\ref{sec:structured} that the singly exponential dependence for uOBDDs is indeed tight.

\subsection{Applications}
\label{sec:applications}

Theorem~\ref{thm:upper_bound} implies several consequences for bounded-treewidth
circuits. The first one deals with \emph{probability computation}: we are given
a \emph{probability valuation}~$\pi$ mapping each variable $g \in C_\var$ to a 
probability that~$g$ is true (independently from other variables),
and we wish to compute the probability
$\pi(C)$ 
that $C$ evaluates to true under~$\pi$, assuming that arithmetic operations (sum
and product) take unit time.
More formally, we define
        the \emph{probability} $\pi(\nu)$ of a valuation~$\nu: C_\var \to \{0,1\}$ as 
        \[\pi(\nu) \colonequals \left(\prod_{g \in C_\var,~\nu(g)=1}
        \pi(g)\right) \left(
        \prod_{g \in C_\var,~\nu(g)=0} (1-\pi(g))\right).\]
	The \emph{probability $\pi(C)$ of Boolean circuit $C$ with probability
	assignment~$\pi$} is then the total probability of the valuations that satisfy~$\phi$. Formally: 
	\[\pi(C) \colonequals \sum_{\nu \text{~satisfies~} \phi}
        \pi(\nu).\]

	When $\pi(x)=1/2$ for every variable, the probability computation
        problem simplifies to the \emph{model counting problem}, i.e., counting the number of satisfying valuations, noted $\#C$.
	Indeed, in this case we have $\#C = 2^{|C_\var|} \times \pi(C)$.
Hence, the probability computation problem is \#P-hard for arbitrary
circuits. However, it is tractable for deterministic decomposable circuits \cite{darwiche2001tractable}. 
Thus,
our result implies the following, where $\card{\pi}$
denotes the size of writing the probability valuation~$\pi$:

\begin{corollary}
  \label{cor:proba}
  Let $f(k)$ be the function from Theorem~\ref{thm:upper_bound}.
  Given a Boolean circuit $C$, a tree decomposition $T$ of width~$\leq k$ of~$C$, and
  a probability valuation~$\pi$ of~$C_\var$, we can compute~$\pi(C)$ in
	$O\left(\card{\pi} + \card{T} \times f(k)\right)$.
\end{corollary}

\begin{proof}
	Use Theorem~\ref{thm:upper_bound} to compute an equivalent d-SDNNF $(D,T',\rho)$; as
  $C$ and $D$ are equivalent, it is clear that $\pi(C) = \pi(D)$. Now, compute
  the probability $\pi(D)$ in linear time in~$D$ and~$\card{\pi}$ by a simple
  bottom-up pass, using the fact that $D$ is a d-DNNF
  \cite{darwiche2001tractable}.
\end{proof}

This improves the bound obtained when applying message passing
techniques~\cite{lauritzen1988local} directly on the bounded-treewidth input
circuit (as presented, e.g., in
\cite[Theorem~D.2]{amarilli2015provenance}). Indeed, message passing
applies to \emph{moralized} representations of the input: for each gate,
the tree decomposition must contain a bag containing all inputs of this gate
\emph{simultaneously}, which is problematic for circuits of large fan-in. 
Indeed, if the original circuit has a tree decomposition of width~$k$,
rewriting it to make it moralized 
will result in a tree decomposition of width~$3k^2$ 
(see~\cite[Lemmas~53
and~55]{amarilli2017combined}),
and
the bound of
\cite[Theorem~D.2]{amarilli2015provenance} then yields an overall complexity 
of~$O\big(|\pi|+|T|\times 2^{3k^2}\big)$ for message passing.
Our Corollary~\ref{cor:proba}
achieves a more favorable bound because Theorem~\ref{thm:upper_bound}
directly uses the associativity of $\land$ and $\lor$.
We note that the connection between message-passing
techniques and structured circuits has also been investigated by
Darwiche,
but his construction
\cite[Theorem~6]{darwiche2003differential} produces arithmetic circuits rather
than d-DNNFs, and it also needs the input to be moralized.

A second consequence concerns the task of \emph{enumerating} the accepting
valuations of circuits, i.e., producing them one after the other, with small
\emph{delay} between each accepting valuation. 
The valuations are concisely represented as \emph{assignments}, i.e., as a
set of variables that are set to true, omitting those that are set to false.
This task is of course NP-hard on
arbitrary circuits (as it implies that we can check whether an accepting
valuation exists), but was recently shown in~\cite{amarilli2017circuit}
to be feasible on d-SDNNFs with
linear-time preprocessing and delay linear in the Hamming weight of each
produced assignment. Hence, we have:

\begin{corollary}
  \label{cor:enum}
  Let $f(k)$ be the function from Theorem~\ref{thm:upper_bound}.
  Given a Boolean circuit $C$ and a tree decomposition $T$ of width~$\leq k$ of~$C$,
  we can enumerate the satisfying assignments of~$C$ with preprocessing in $O\left(\card{T}
	\times f(k)\right)$ and delay linear in the size of each produced
  assignment.
\end{corollary}

\begin{proof}
	Use Theorem~\ref{thm:upper_bound} to compute an equivalent d-SDNNF $(D,T',\rho)$, which
  has the same accepting valuations. We 
  conclude using \cite[Theorem~2.1]{amarilli2017circuit}.
\end{proof}

This corollary refines some existing results about enumerating the satisfying
valuations of some circuit classes with \emph{polynomial}
delay~\cite{DBLP:journals/jair/DarwicheM02}, and also relates to results on the
enumeration of monomials of arithmetic circuits~\cite{strozecki2010enumeration}
or of solutions to constraint satisfaction problems
(CSP)~\cite{creignou2011enumerating}, again with polynomial delay. It also relates to recent incomparable
results on constant-delay enumeration for classes of
DNF formulae~\cite{capelli2018enumerating}.

A third consequence concerns the tractability of \emph{quantifying variables} in bounded-treewidth circuits.
Let $\phi$ be a Boolean function on variables $V$, and let $X_1, \ldots X_n$ be disjoint subsets of $V$. 
A \emph{quantifier prefix} of length $n$ is a prefix of the form $\Pi \defeq Q_1 X_1 \ldots Q_n X_n$, where each $Q_i$ is either $\exists$ or $\forall$, with $Q_i \neq Q_{i+1}$.
Let $\Pi(\phi)$ be the Boolean function on variables $V \setminus (X_1 \cup \dots \cup X_n)$, with the obvious semantics.
Then~\cite{capelli2019tractable} shows:

\begin{theorem}[{\cite[Theorem~5]{capelli2019tractable}}]\label{thm:mainprojection}
There is an algorithm that, given a complete SDNNF $(D,T,\rho)$ on variables $V$ of width $k$ and $Z \subseteq V$, computes in time $2^{O(k)}|D|$ a complete d-SDNNF of width at most $2^k$ having a designated gate computing $\exists Z~D$ and another designated gate computing $\neg \exists Z~D$.
\end{theorem}

By iterating the construction of Theorem~\ref{thm:mainprojection} and using the identity $\forall X.D \equiv \neg \exists X \neg D$, one can easily get:
\begin{corollary}
	\label{cor:qbf}
	Let $\Pi \defeq Q_1 X_1 \ldots Q_n X_n$ be a quantifier prefix of length $n$ with $Q_n = \exists$. 
	There is an algorithm that, given a complete d-SDNNF $(D,T,\rho)$, computes in time $\exp^n(k) \times |D|$ a complete structured
	d-SDNNF of width $\leq \exp^n(k)$ representing $\Pi(D)$, where
        $\exp^{n} (k) = \underbrace{2^{\reflectbox{$\ddots$}^{2^{O(k)}}}}_{n}$.
\end{corollary}

We can then combine Corollary~\ref{cor:qbf} with our Theorem~\ref{thm:upper_bound} to show:

\begin{corollary}
Let $C$ be a Boolean circuit of treewidth $\leq k$, and let $\Pi$ be a quantifier prefix of length $n$ that ends with $\exists$.
	We can compute in time $\exp^{n+1} (k) \times |C|$ a d-SDNNF of width at most $\exp^{n+1} (k)$ representing $\Pi(C)$.
\end{corollary}

This generalizes the corresponding result of~\cite{capelli2019tractable}, which applies to bounded-treewidth CNFs instead of bounded-treewidth circuits.
(However, we note that~\cite[Theorem~10]{capelli2019tractable} also applies to CNFs of bounded \emph{incidence} treewidth, which can be smaller than the primal treewidth that we use in our article.) We refer to~\cite{capelli2019tractable} for a discussion of the related work on model counting of quantified formulas.

Other applications of Theorem~\ref{thm:upper_bound} include counting the
number of satisfying valuations of the circuit (a
special case of probability computation), MAP inference
\cite{fierens2015inference}, or random sampling of possible worlds (which
can easily be done on the d-SDNNF).

%% file: construction.tex
We first present in Section~\ref{sec:connecting-construction} the construction used for
Theorem~\ref{thm:upper_bound}, then prove in Section~\ref{sec:proofproof} that this construction is correct and can be done within the prescribed time bound.
We then explain how to specialize the construction to the case of bounded-pathwidth circuits and uOBBDs in Section~\ref{sec:construction_pathwidth}.

\subsection{Construction}
\label{sec:connecting-construction}

Let $C$ be the input circuit on $n$ variables, and $T$ the input tree decomposition of $C$ of width $\leq k$.
We start with prerequisites.
\myparagraph{Prerequisites.}
Let $g_\out$ be the output gate of $C$.
Thanks to Lemma~\ref{lem:friendly}, we can assume that $T$ is $g_\out$-friendly.
For every variable gate $x \in C_\var$, we choose a leaf bag $b_x$ of $T$ such that $\lambda(b_x) = \{x\}$. Such a leaf bag exists because $T$ is friendly (specifically, thanks to bullet points $2$ and $3$).
We say that $b_x$ is \emph{responsible for the variable gate $x$}.
We can obviously choose such a $b_x$ for every variable gate $x$ in linear time in $T$.

To abstract away the type of gates and their values in the construction,
we will talk of \emph{strong} and \emph{weak} values. Intuitively, a value is
\emph{strong} for a gate $g$ if any input $g'$ of~$g$ which carries this value
determines the value of~$g$; and \emph{weak} otherwise. Formally:
 
\begin{definition}
  Let $g$ be a gate and $c \in \{0, 1\}$:
  \begin{itemize}
    \item If $g$ is an $\land$-gate, we say that $c=0$ is \emph{strong} for~$g$ and $c=1$
      is \emph{weak} for~$g$;
    \item If $g$ is an $\lor$-gate, we say that $c=1$ is \emph{strong} for~$g$ and $c=0$
      is \emph{weak} for~$g$;
    \item If $g$ is a $\NOT$-gate, $c=0$ and $c=1$ are both \emph{strong}
      for~$g$;
    \item For technical convenience, if $g$ is a $\var$-gate, $c=0$ and $c=1$ are both \emph{weak} for~$g$.
  \end{itemize}
\end{definition}

If we take any valuation $\nu:C_\var\to\{0,1\}$ of the circuit $C = (G, W,
g_\out,
\mu)$,
and extend it to an evaluation $\nu:G\to\{0,1\}$, then $\nu$ will respect the
semantics of gates. In particular, it will \emph{respect strong values}: 
for any gate $g$ of~$C$, if $g$ has an input $g'$ for which $\nu(g')$ is a strong
value,
then $\nu(g)$ is determined by~$\nu(g')$, specifically, it is~$\nu(g')$ if~$g$ is
an $\lor$- or an $\land$-gate, and $1-\nu(g')$ if $g$ is a $\NOT$-gate. In our
construction, we will need to guess how gates of the circuit are evaluated,
focusing on a subset of the gates (as given by a bag of~$T$); we will then call
\emph{almost-evaluation} an assignment of these gates that respects strong
values. Formally:
\begin{definition}
	Let $U$ be a set of gates of~$C$.
        We call $\nu: U \to \{0, 1\}$ a \emph{$(C,U)$-almost-evaluation} if 
        it \emph{respects strong values}, i.e., for every gate $g \in U$, if
        there is an input~$g'$ of~$g$ in~$U$ such that $\nu(g')$ is a strong value for
        $g$, then $\nu(g)$ is determined from $\nu(g')$ in the sense above.
\end{definition}

Respecting strong values is a necessary condition for such an assignment to be extensible to a
valuation of the entire circuit. However, it is not sufficient: an
almost-evaluation $\nu$ may map
a gate $g$ to a strong value even though $g$ has no input that can justify this
value. This is hard to avoid: when we focus on the
set~$U$, we do not know about other inputs of~$g$. For now, let us call
\emph{unjustified} 
the gates of~$U$ that carry a strong value that is not justified by~$\nu$:

\begin{definition}
	Let $U$ be a set of gates of a circuit $C$ and $\nu$ a $(C,U)$-almost-evaluation.
        We call $g \in U$ \emph{unjustified} if $\nu(g)$ is a strong value
        for~$g$, but, for every input $g'$ of~$g$ in~$U$, the value $\nu(g')$ is
        weak for~$g$; otherwise, $g$ is \emph{justified}.
        The set of unjustified gates is written~$\UNF(\nu)$.
\end{definition}

Let us start to explain in a high-level manner
how to construct the d-SDNNF $D$ equivalent to the input
circuit~$C$ (we will later describe the construction formally). We do so by traversing
$T$ bottom-up,
and for each bag $b$ of $T$ we create gates~$G_b^{\nu,S}$ in~$D$,
where $\nu$~is a $(C,b)$-almost-evaluation
and $S$ is a subset of~$\UNF(\nu)$ which we call the \emph{suspicious gates}
of~$G_b^{\nu,S}$. We will connect the gates of~$D$
created for each internal bag~$b$ with the
gates created for its children in~$T$, in a 
way that we will 
explain later.
Intuitively, for a gate $G_b^{\nu,S}$ of~$D$,
the \emph{suspicious gates} $g$ in the set~$S$ are gates of~$b$
whose strong value is not justified by~$\nu$ (i.e., $g \in \UNF(\nu)$),
and is not justified either by any of the almost-evaluations at descendant bags
of~$b$ to which $G_b^{\nu,S}$ is connected.
We call \emph{innocent} the other gates of~$b$; hence, they are the gates that are 
justified in~$\nu$ (in particular, those who carry weak values),
and the gates that are unjustified in~$\nu$ but have been
justified by an almost-evaluation at a descendant bag $b'$ of~$b$. Crucially, in
the latter case, the gate $g'$ justifying the strong value in~$b'$ may no longer
appear in~$b$, making $g$ unjustified for~$\nu$; this is why we remember the set~$S$.

We still have to explain 
how we connect the gates $G_b^{\nu,S}$ of~$D$ to the gates
$G_{b_l}^{\nu_l,S_l}$ and~$G_{b_r}^{\nu_r,S_r}$ created for the children $b_l$
and $b_r$ of~$b$ in~$T$.
The first condition is that $\nu_l$ and~$\nu_r$ must \emph{mutually agree},
i.e., $\nu_l(g) = \nu_r(g)$ for all $g \in b_l \cap b_r$, and $\nu$ must then be
the union of~$\nu_l$ and~$\nu_r$, restricted to~$b$. 
We impose a second condition to prohibit suspicious gates from escaping
before they have been justified, which we formalize as \emph{connectibility} of
a pair $(\nu,S)$ at bag~$b$ to the parent bag of~$b$.
\begin{definition}
	\label{def:connectible}
	Let $b$ be a non-root bag, $b'$ its parent bag,
        and $\nu$ a $(C,b)$-almost-evaluation.
        For any set $S \subseteq \UNF(\nu)$,
        we say that $(\nu,S)$ is \emph{connectible} to~$b'$
        if $S \subseteq b'$, i.e., the suspicious gates of $\nu$ must still
        appear in~$b'$.
\end{definition}
If a gate $G^{\nu,S}_b$ is such that $(\nu,S)$ is not connectible to the parent
bag~$b'$, then this gate will not be used as input to any other gate, but we do
not try to preemptively remove these useless gates in the construction (but note that this will be taken care of at the end, when we will apply Lemma~\ref{lem:reduce_extended}).
We are now ready to give the formal definition that will be used to
explain how gates
are connected:
\begin{definition}
	\label{def:result}
	Let $b$ be an internal bag with children $b_l$ and $b_r$, 
        let $\nu_l$ and
        $\nu_r$ be respectively $(C,b_l)$ and $(C,b_r)$-almost-evaluations that
        mutually agree,
        and $S_l \subseteq \UNF(\nu_l)$
        and $S_r \subseteq \UNF(\nu_r)$ be sets of
	suspicious gates 
        such that both $(\nu_l,S_l)$ and $(\nu_r,S_r)$ are connectible to~$b$.
        The \emph{result} of $(\nu_l,S_l)$ and $(\nu_r,S_r)$ is then defined as
        the pair $(\nu,S)$ where:
	\begin{itemize}
		\item $\nu$ is defined as the
                  restriction of~$\nu_l \cup \nu_r$ to~$b$.
                \item $S \subseteq \UNF(\nu)$ is the new set of suspicious
                  gates, defined as follows.
                  A gate $g \in b$ is innocent (i.e., $g \in b \setminus S$) if it is
                  justified for $\nu$ or if it is innocent for some child.
			Formally, $b \setminus S \colonequals (b
                        \setminus\UNF(\nu)) \cup \big[b \cap \left[ (b_l \setminus S_l) \cup (b_r
                        \setminus S_r)\right]\big]$.
	\end{itemize}
        We point out that $(\nu,S)$ is not necessarily a $(C,b)$-almost-evaluation.
\end{definition}

\myparagraph{Construction.}
We now use these definitions to present the construction formally.
For every variable gate $g$ of $C$,
we create a corresponding variable gate $G^{g,1}$ of $D$,
and we create $G^{g,0} \colonequals \NOT(G^{g,1})$.
For every internal bag $b$ of $T$,
for each $(C,b)$-almost-evaluation $\nu$ and set
$S\subseteq \UNF(\nu)$ of suspicious
gates of $\nu$,
we create 
one $\lor$-gate $G_b^{\nu,S}$.
For every leaf bag $b$ of $T$, we create 
one $\lor$-gate $G_b^{\nu,S}$ for every $(C,b)$-almost-evaluation $\nu$,
where we set $S \colonequals \UNF(\nu)$; intuitively, in a leaf bag, an unjustified
gate is always suspicious (it cannot have been justified at a descendant bag).

Now, for each internal bag $b$ of~$T$ with children $b_l,b_r$,
for each pair of gates $G_{b_l}^{\nu_l,S_l}$ and~$G_{b_r}^{\nu_r,S_r}$ that are both
connectible to~$b$ and where $\nu_l$ and $\nu_r$ mutually agree, letting $(\nu,
S)$ be the result of~$(\nu_l,S_l)$ and~$(\nu_r,S_r)$, if $(\nu,S)$ is a $(C,b)$-almost-evaluation then we create 
a gate $G_b^{\nu_l,S_l,\nu_r,S_r}=\land(G_{b_l}^{\nu_l,S_l}, G_{b_r}^{\nu_r,S_r})$
and make it an input of~$G^{\nu,S}_b$.
We now explain where the variables gates are connected. 
For every leaf bag $b$ 
that is responsible for a variable gate~$x$ (i.e., $b$ is $b_x$), for $\nu \in \{\{x \mapsto 1\}, \{x \mapsto 0\}\}$,
we set the gate $G^{x,\nu(x)}$ 
to be the (only) input of the gate $G_b^{\nu,S}$.
Last, for every leaf bag $b$ 
that is not responsible for a variable gate, for every valuation $\nu$ of $b$, we create a constant $1$-gate (i.e., an $\land$-gate with no inputs), and we make it the (only) input of $G_b^{\nu,S}$.
The output gate of~$D$ is the gate $G^{\nu,\emptyset}_{b_\root}$ where
$b_\root$ is the root of~$T$
and $\nu$ maps~$g_\out$ to~$1$
(remember that $b_\root$ contains only $g_\out$).

We now construct the extended v-tree $T'$ together with the mapping $\rho$.
$T'$ has the same skeleton than $T$ (in particular, it is a full binary tree). For every node $b$ of $T$, let us denote by $b'$ the corresponding node of $T'$. 
For every leaf bag $b$ of $T$, $b'$ is either $x$ if $b$ is responsible for the variable $x$, or unlabeled otherwise.
For every bag $b$ of $T$ and every gate $g$ of the form $G_b^{\nu,S}$ or $G_b^{\nu_l,S_l,\nu_r,S_r}$, we take $\rho(g) = b'$.
For every leaf bag $b$ of $T$ that is not responsible for a variable, for any gate $g$ of the form $G_b^{\nu,S}$ (there can be either one (if $b$ is empty) or two of them), 
letting $g'$ be the (only) input of $g$  (i.e., $g'$ is a constant $1$-gate), we set $\rho(g') = b'$.
For every leaf bag $b$ of $T$ that is responsible for a variable $x$, we set 
$\rho(G^{x,1}) = \rho(G^{x,0}) = b'$.

\subsection{Proof of Correctness}
\label{sec:proofproof}
\input{proofproof}

\subsection{uOBDDs for Bounded-Pathwidth Circuits}
\label{sec:construction_pathwidth}

We now argue that our construction for Theorem~\ref{thm:upper_bound} can be specialized in the case of bounded-pathwidth circuits to compute uOBDDs, i.e., we prove Theorem~\ref{thm:upper_bound_pw}.
If the input circuit captures a constant Boolean function then there is no
difficulty, so we assume that this is not the case.

Let $C$ be the Boolean circuit with output gate $g_\out$, and $P$ be a path decomposition of $C$ of width $\leq k$.
It is clear that, by adapting Lemma~\ref{lem:friendly}, we can compute in linear time from $P$ a $g_\out$-friendly \emph{tree} decomposition $T$ that is further right-linear.
We then use the same construction as the one we use in Theorem~\ref{thm:upper_bound} with $T$.
This gives us an extended complete d-SDNNF $(D,T',\rho)$ of width $\leq 2^{2(k+1)}$ equivalent to $C$, where $T'$ is an extended v-tree that is right-linear.
It is easy to verify that we can compute this d-SDNNF in time $O(|T| \times f(k))$ for some function $f$ that is in $O(2^{(2+\epsilon)k})$ for any $\epsilon > 0$ (as opposed to $O(2^{(4+\epsilon)k})$ in the original construction). 
This is thanks to the fact that $T$ is right-linear, and because the leaf bags of a friendly tree decomposition can contain at most one gate of $C$.
We then use Lemma~\ref{lem:reduce_extended} to compress the extended complete d-SDNNF into a complete d-SDNNF $(D',T'',\rho')$ of width $\leq 2^{2(k+1)}$, again equivalent to $C$. 
By inspection of the proof of this Lemma, it is clear that $T''$ is a (non extended) v-tree that is right-linear. 
Again by inspection of the proof of Lemma~\ref{lem:reduce_extended}, we can see that we can rewrite every $\land$-gate $g$ of $D'$ to 
be of the form $g' \land x$ or $g' \land \lnot x$, for some variable $x$ and gate~$g'$: this is thanks to the fact that Lemma~\ref{lem:reduce_extended} starts by propagating constants, so that the 
right input of $g$ can only 
be equivalent to $x$ or to $\lnot x$
(there is an $\lor$-gate in between, which we remove).

We now justify that we can transform $D'$ into a complete uOBDD.
First, create the $0$-sink $\bot$ and the $1$-sink $\top$.
Then, we traverse the internal nodes $n$ of~$T'$ top-down and inductively define an uOBDD $O(g)$ for every $\lor$-gate $g$ of~$D'$ structured by $n$.
The intuition is that $O(g)$  captures the subcircuit rooted at $g$. 
Remember that $C$ captures a non-constant Boolean
function, which implies that there are no $\land$- or $\lor$-gates
without input left in $D'$.
We proceed as follows, letting $l$ be the right child of $n$, corresponding to variable $x$, and $n'$ being the left child of $n$:
\begin{itemize}
	\item For every $\lor$-gate $g$ of $D'$ structured by $n$ we do the following. 
		First, compute $A_g^0$, the set of gates $g'$ such that there
                exists an $\land$-gate of the form $g' \land \lnot x$ that is an input of $g$,
		and $A_g^1$, the set of gates $g'$ such that there exists an $\land$-gate of the form $g' \land x$ that is an input of $g$.
		Then define $O^0(g)$ to be $\bigvee_{g' \in A_g^0} O(g')$; and similarly $O^1(g) \defeq \bigvee_{g' \in A_g^1} O(g')$.
		Finally define $O(g)$ to consist of a node labeled by $x$, with outgoing $0$-edge to $O^0(g)$ and outgoing $1$-edge to $O^1(g)$.
\end{itemize}

There is one special case in this construction: we might not have defined $O(g)$ when $g$ is a variable gate (resp., a $\lnot$-gate) that is structured by the leftmost leaf of $T$.
In that case, we define $O(g)$ to consist of a node labeled by the corresponding variable, with outgoing $1$-edge to $\top$ (resp., $\bot$), and outgoing $0$-edge to $\bot$ (resp., $\top$).
It is then clear that, by considering the output gate $G_\out$ of $D'$, we have that $O(G_\out)$ is a complete uOBDD of width $\leq 2^{2(k+1)}$ that captures the same function as $D'$.

%% file: proofproof.tex
We now prove that $(D,T',\rho)$ is indeed an extended complete d-SDNNF equivalent to the initial circuit
  $C$, that its width is $\leq 2^{2 (k+1)}$, and that it can be constructed in time $O\left(|T|\times
  2^{(4+\epsilon)k}\right)$ for any~$\epsilon>0$. 

\subsubsection{$(D,T',\rho)$ is an extended complete SDNNF of the right width}

Negations only apply to the input gates, so $D$ is an NNF.
It is easy to check that $(D,T',\rho)$ satisfies the conditions of being a complete extended SDNNF. 
Now, for every leaf $l$ of $T'$, there at most two $\lor$-gates of $D$ that are structured by $l$ (remember that leaf bags of the friendly tree decomposition $T$ contain one or zero gates of $C$). 
For every internal node $n$ of $T'$ corresponding to a bag $b$ of $T$, the $\lor$-gates that are structured by $n$ are of the form $G^{\nu,S}_b$, so they are
at most $2^{2 (k+1)}$, which shows our claim about the width of $(D,T',\rho)$.

\subsubsection{$D$ is equivalent to $C$}

We now show that $D$ is equivalent to the original circuit $C$.
Recall the definition of a trace of a DNNF from Definition~\ref{def:trace}.
Our first step is to prove that traces have exactly one almost-evaluation
  corresponding to each descendant bag, and that these almost-evaluations mutually agree. 
\begin{lemma}
	\label{lem:only_one_agree2}
	Let $G^{\nu,S}_b$ a gate in $D$ and $\Xi$ be a 
	trace of $D$ starting at $G^{\nu,S}_b$.
	Then for any bag $b' \leq b$
        (meaning that $b'$ is $b$ or a descendant of $b$),
        $\Xi$ contains exactly one gate of the form~$G^{\nu',S'}_{b'}$.
	Moreover, over all $b' \leq b$, all the almost-evaluations of the gates $G^{\nu',S'}_{b'}$ that are in~$\Xi$ mutually agree.
\end{lemma}

\begin{proof}
	The fact that $\Xi$ contains exactly one gate $G^{\nu',S'}_{b'}$
        for any bag $b' \leq b$ is obvious by construction of $D$, as
        $\lor$-gates are assumed to have exactly one input in~$\Xi$.
	For the second claim, suppose by contradiction that not all the almost-evaluations of the gates $G^{\nu',S'}_{b'}$ that are in $\Xi$ mutually agree.
	We would then have $G^{\nu_1,S_1}_{b_1}$ and $G^{\nu_2,S_2}_{b_2}$ in $\Xi$ and $g \in b_1 \cap b_2$ such that $\nu_1(g) \neq \nu_2(g)$. 
	But because $T$ is a tree decomposition,
	$g$ appears in all the bags on the path from $b_1$ and $b_2$, and by construction the almost-evaluations of 
	the gates $G^{\nu',S'}_{b'}$ on this path that are in $\Xi$ mutually agree, a contradiction.
\end{proof}

Therefore, Lemma~\ref{lem:only_one_agree2} allows us to define the union of the almost-evaluations in such a trace:
\begin{definition}
	\label{def:gamma2}
	Let $G^\nu_b$ a gate in $D$
	and $\Xi$ be a trace of $D$ starting at $G^\nu_b$.
	Then $\gamma(\Xi) \defeq \bigcup_{G^{\nu',S'}_{b'} \in \Xi} \nu'$ (the union of the almost-evaluations in $\Xi$, which is a valuation 
	from $\bigcup_{G^{\nu',S'}_{b'} \in \Xi} b'$ 
	to $\{0,1\}$) is properly defined.
\end{definition}

We now need to prove a few lemmas about the behavior of gates that are innocent (i.e., not suspicious).
\begin{lemma}
\label{lem:behaviour}
	Let $G^{\nu,S}_b$ a gate in $D$ and $\Xi$ be a 
	trace of $D$ starting at $G^{\nu,S}_b$.
	Let $g \in b$ be a gate that is innocent ($g \notin S$). Then the following holds:
	\begin{itemize}
		\item If $\nu(g)$ is a weak value of $g$, then for every input $g'$ of $g$ that is in the domain of $\gamma(\Xi)$ (i.e., $g'$ appears in a bag $b' \leq b$),
			we have that $\gamma(\Xi)$ maps $g'$ to a weak value of $g$;
		\item If $\nu(g)$ is a strong value of $g$, then there exists an input $g'$ of $g$ that is in the domain of $\gamma(\Xi)$ such that 
			$\gamma(\Xi)(g')$ is $\nu(g)$ if $g$ is an $\land$-
                        or $\lor$-gate, and $\gamma(\Xi)(g')$ is $1-\nu(g)$
                        if $g$ is a $\NOT$-gate.
	\end{itemize}
\end{lemma}
\begin{proof}
	We prove the claim by bottom-up induction on $b \in T$. One can easily check that the claim is true when $b$ is a leaf bag, remembering that
	in this case we must have $S = \UNF(\nu)$ by construction (that is, all the gates that are unjustified are suspicious).
	For the induction case, let $b_l$, $b_r$ be the children of $b$.
	Suppose first that $\nu(g)$ is a weak value of $g$, and suppose for a contradiction that there is an input $g'$ of $g$ in the domain of $\gamma(\Xi)$ such that
	$\gamma(\Xi)(g')$ is a strong value of $g$. By the occurrence and connectedness properties of tree decompositions, there exists a bag $b' \leq b$ in which
	both $g$ and $g'$ occur. Consider the  gate $G^{\nu',S'}_{b'}$ that is in $\Xi$:
        by Lemma~\ref{lem:only_one_agree2}, this gate exists and is unique.
	By definition of $\gamma(\Xi)$ we have $\nu'(g')=\gamma(\Xi)(g')$.
	Because $\nu'$ is a $(C,b')$-almost-evaluation that maps $g'$ to a strong value of $g$, we must have that $\nu'(g)$ is also a strong value of $g$,
	thus contradicting our hypothesis that $\nu(g)=\gamma(\Xi)(g) = \nu'(g)$ is a weak value for $g$.

	Suppose now that $\nu(g)$ is a strong value of $g$. We only treat
        the case when $g$ is an $\lor$- or an $\land$-gate, as the case of a
        $\NOT$-gate is similar.
	We distinguish two sub-cases:

	\begin{itemize}
		\item $g$ is justified. Then, by definition of $\nu$ being a $(C,b)$-almost-evaluation, there must exist an input $g'$ of $g$ that is 
			also in $b$ such that $\nu(g')$ is a strong value of $g$, which proves the claim.
		\item $g$ is unjustified.
			But since $g$ is innocent ($g \notin S$),
			by construction (precisely, by the second item of Definition~\ref{def:result}) $g$ must then be innocent for a child of $b$.
			The claim then holds by induction hypothesis.\qedhere
	\end{itemize}
\end{proof}

Lemma~\ref{lem:behaviour} allows us to show that for a gate $g$ that is not a variable gate, letting $b$ be the topmost bag in which $g$ appears (hence, each input of $g$ must occur in some bag 
$b' \leq b$), if $g$ is innocent then for any trace $\Xi$ starting at a gate for bag $b$, $\gamma(\Xi)$ respects the semantics of $g$.
Formally, recalling that $W(g)$ denotes the inputs of~$g$:

\begin{lemma}
\label{lem:respects_semantics}
	Let $G^{\nu,S}_b$ a gate in $D$ and $\Xi$ be a 
	trace of $D$ starting at $G^{\nu,S}_b$.
	Let $g \in b$ be a gate that is not a variable gate and such that $b$ is the topmost bag in which $g$ appears (hence $W(g) \subseteq \text{domain}(\gamma(\Xi))$).
	If $g$  is innocent ($g \notin S$) then $\gamma(\Xi)$ respects the semantics of $g$, that is, 
	$\gamma(\Xi)(g)=\bigodot \gamma(\Xi)(W(g))$ where $\bigodot$ is the type of $g$.
\end{lemma}
\begin{proof}
	Clearly implied by Lemma~\ref{lem:behaviour}.
\end{proof}

We need one last lemma about the behavior of suspicious gates, which intuitively tells us that if we have already seen all the input gates of a gate $g$
and $g$ is still suspicious, then $g$ can never escape:
\begin{lemma}
\label{lem:suspicious_not_innocent_propagate}
	Let $G^{\nu,S}_b$ a gate in $D$ and $\Xi$ be a 
	trace of $D$ starting at $G^{\nu,S}_b$.
	Let $g$ be a gate such that the topmost bag $b'$ in which $g$ appears is $\leq b$, and consider the unique gate of the form
	$G^{\nu',S'}_{b'}$ that is in $\Xi$.
	If $g \in S'$ then $b'=b$ (hence $G^{\nu,S}_b = G^{\nu',S'}_{b'}$
        by uniqueness).
\end{lemma}
\begin{proof}
Let $g \in S'$. Suppose by contradiction that $b' \neq b$.
Let $p$ be the parent of $b'$ (which exists because $b' < b$).
It is clear that by construction $(\nu',S')$ is connectible to $p$ (recall Definition~\ref{def:connectible}), hence $g$ must be in $p$,
contradicting the fact that $b'$ should have been the topmost bag in which $g$ occurs.
Hence $b'=b$.
\end{proof}

We now have all the results that we need to show that $D \implies C$, i.e. that, for every valuation~$\chi$ of the variables of~$C$, if $\chi(D) = 1$ then $\chi(C) = 1$ (we see $\chi$ both as a valuation of the variables of $C$, and as a valuation of the (corresponding) variables of $D$). We prove a stronger result.
Given a valuation $\chi$ of the variable gates of $D$ and a gate $g$ of $D$, we say that a \emph{trace of $D$ starting at $g$ according to~$\chi$} 
is a trace of $D$ starting at $g$ such that $\chi$ satisfies every gate in $\Xi$. We show:

\begin{lemma}
	\label{lem:corresponds}
	Let $\chi$ be a valuation of the variable gates, $G^{\nu,\emptyset}_{\root(T)} \in D$ a gate that evaluates to~$1$ under $\chi$,
	and $\Xi$ a trace 
	of $D$ starting at $G^{\nu,\emptyset}_{\root(T)}$ according to~$\chi$. Then $\gamma(\Xi)$ corresponds to the evaluation $\chi$ of $C$.
\end{lemma}
\begin{proof}
  We prove by induction on $C$ (as its graph is a DAG) that for all $g \in C$, $\gamma(\Xi)(g) = \chi(g)$. When $g$ is a variable gate, 
consider the leaf bag $b_g$ that is responsible of $g$,
  and consider
 the gate~$G^{\nu',S'}_{b_g}$ that is in $\Xi$: this gate exists and is unique according to Lemma~\ref{lem:only_one_agree2}. 
	This gate evaluates to $1$ under $\chi$ (because it is in the trace),
 which is only possible if $G^{g,\nu'(g)}$ evaluates to $1$ under $\chi$, hence by construction we must have $\nu'(g)=\chi(g)$ and then $\gamma(\Xi)(g)=\chi(g)$.
 When $g$ is a constant gate (an $\lor$- or $\land$-gate with no inputs), consider the topmost bag $b'$ in which $g$ appears, and consider the unique $G_{b'}^{\nu',S'}$ that is in $\Xi$.
	According to Lemma~\ref{lem:respects_semantics}, we have that $\gamma(\Xi)(g) = \chi(g)$.
 Now suppose that $g$ is an internal gate of type $\bigodot$, and consider the topmost bag $b'$ in which $g$ appears. 
 Consider again the unique $G^{\nu',S'}_{b'}$ that is in~$\Xi$. 
	By induction hypothesis, we know that $\gamma(\Xi)(g') = \chi(g')$ for every input $g'$ of $g$.
We now distinguish two cases:
\begin{itemize}
	\item $b' = \root(T)$.
		Therefore
		by Lemma~\ref{lem:respects_semantics} we know that $\gamma(\Xi)$ respects the semantics of~$g$, which means that  
		$\gamma(\Xi)(g) = \bigodot \gamma(\Xi)(W(g)) = \bigodot \chi(W(g)) = \chi(g)$ (where the second equality comes from the induction hypothesis
		and the third equality is just the definition of the evaluation $\chi$ of $C$), which proves the claim.
	\item $b' < \root(T)$.
          But then by Lemma~\ref{lem:suspicious_not_innocent_propagate} we must have $g \notin S'$
		 (because otherwise we should have $b' = \root(T)$ and then $S' = \emptyset$, a contradiction), that is $g$ is innocent for $G^{\nu',S'}_{b'}$.
		Therefore, again by Lemma~\ref{lem:respects_semantics}, it must be the case that
		$\gamma(\Xi)$ respects the semantics of~$g$, and we can again
                show that $\gamma(\Xi)(g) = \chi(g)$, concluding the proof.
                \qedhere
\end{itemize}
\end{proof}

This indeed implies that $D \implies C$: let $\chi$ be a valuation of the variable gates and suppose $\chi(D) = 1$. Then by definition of the output of 
$D$, it means that the gate $G^{\nu,\emptyset}_{\root(T)}$ such that $\nu(g_\out)=1$ evaluates to $1$ under $\chi$.
But then, considering a trace $\Xi$ of $D$ starting at~$G^{\nu,\emptyset}_{\root(T)}$ according to~$\chi$, we have that $\chi(g_\out) = \gamma(\Xi)(g_\out) = \nu(g_\out) = 1$.

To show the converse ($C \implies D$), one can simply observe the following phenomenon:

\begin{lemma}
	\label{lem:follows2}
	Let $\chi$ be a valuation of the variable gates. Then for every
        bag $b \in T$, the gate~$G^{\chi_{|b},S}_b$ evaluates to $1$
        under $\chi$, where $S$ is the set of gates $g \in \UNF(\nu)$
        such that for every input $g'$ of $g$ that appears in some bag $b'
        \leq b$, then $\chi(g')$ is
        a weak value of $g$.
\end{lemma}
\begin{proof}
	Easily proved by bottom-up induction.
\end{proof}

Now suppose $\chi(C)=1$. By Lemma~\ref{lem:follows2} we have that
$G^{\chi_{|\root(T)},\emptyset}_{\root(T)}$ evaluates to $1$ under~$\chi$, and because $\chi(g_\out)=1$ we have that $\chi(D)=1$.
This shows that $C \implies D$.
Hence, we have proved that $D$ is equivalent to $C$.

\subsubsection{$D$ is deterministic}
We now prove
that $D$ is deterministic, i.e., that every $\lor$-gate in $D$ is deterministic.
Recall that the only $\lor$-gates in $D$ are the gates of the form $G^{\nu,S}_{b}$.
We will in fact prove that for every valuation $\chi$ and every $\lor$-gate of $D$, there exists at most one trace of $D$ starting at that gate according to $\chi$, which clearly implies that all the
$\lor$-gates are deterministic. 

We start by proving the following lemma:

\begin{lemma}
\label{lem:innocent_or_suspect_but_not_innocent}
	Let $G^{\nu,S}_b$ be a gate in $D$ and $\Xi$ be a 
	trace of $D$ starting at $G^{\nu,S}_b$. 
	Let $g \in b$. Then the following is true:

	\begin{itemize}
		\item if $g$ is innocent ($g \notin S$) and $\nu(g)$ is a strong value of $g$,
			then there exists an input $g'$ of $g$ in the domain of $\gamma(\Xi)$ such that $\gamma(\Xi)(g')$ is a strong value for~$g$.
		\item if $g \in S$, then for every input $g'$ of $g$ that is in the domain of $\gamma(\Xi)$, we have that $\gamma(\Xi)(g')$ is a weak value for $g$. 
	\end{itemize}
\end{lemma}
\begin{proof}
  We prove the two claims independently:

	\begin{itemize}
		\item Let $g \in b$ such that $g \notin S$ and $\nu(g)$ is a strong value for~$g$.
			Then the claim directly follows from the second item of Lemma~\ref{lem:behaviour}.
                      \item We prove the second claim via a bottom-up induction on $T$.
			When $b$ is a leaf then it is trivially true because $g$ has no input $g'$ in $b$ because $|b| \leq 1$ because
			$T$ is friendly.
			For the induction case, let $G^{\nu_l,S_l}_{b_l}$ and $G^{\nu_r,S_r}_{b_r}$ be the (unique) gates in $\Xi$ corresponding to the children
			$b_l, b_r$ of~$b$. By hypothesis we have $g \in S$.
			By definition of a gate being suspicious, we know that $\nu(g)$ is a strong value for~$g$.
			To reach a contradiction, assume that there is an input $g'$ of $g$ in the domain of $\gamma(\Xi)$ such that
			$\gamma(\Xi)(g')$ is a strong value for~$g$. Clearly this $g'$ is not in $b$, because $g$ is unjustified by $\nu$ (because $S \subseteq \UNF(\nu)$).
			Either $g'$ occurs in a bag $b_l' \leq b_l$, or it occurs in a bag $b_r' \leq b_r$. 
			The two cases are symmetric, so we assume that we are in the former.
                        As $g \in b$ and $g' \in b_l'$, by the properties of tree decompositions and because $g' \notin b$,
			we must have $g \in b_l$.
			Hence, by the contrapositive of the induction hypothesis on $b_l$ applied to $g$, we deduce that $g \notin  S_l$.
			But then by the second item of Definition~\ref{def:result}, $g$ should be innocent for $G^{\nu,S}_b$, that is $g \notin S$,	
			which is a contradiction. \qedhere
	\end{itemize}
\end{proof}

We are ready to prove that traces starting at $\lor$-gates are unique (according to a valuation of the variable gates).
Let us first introduce some useful notations:
Let $U$, $U'$ be sets of gates, $\nu$, $\nu'$ be valuations having the same domain. 
We write $(\nu,U) = (\nu',U')$ to mean $\nu = \nu'$ and $U = U'$, and for
$g$ in the domain of $\nu$ we write
$(\nu,U)(g) = (\nu',U')(g)$ to mean that $\nu(g)=\nu'(g)$ and that we have $g \in U$ iff $g \in U'$.
We show the following:

\begin{lemma}
	Let $\chi$ be a valuation of the variable gates, $G^{\nu,S}_{b}$ be a gate of $D$.
	Then there exists at most one trace of $D$ starting at 
	$G^{\nu,S}_{b}$ according to $\chi$.
\end{lemma}
\begin{proof}
	Fix the valuation $\chi$.
	We will prove the claim by bottom-up induction on $T$. 
	The case when $b$ is a leaf is trivial because gates of the form $G^{\nu,S}_{b}$, for $b$ a leaf of $T$, have either:

	\begin{itemize}
		\item exactly one input $g'$ that is a constant $1$-gate;
		\item or, exactly one input that is a variable gate $G^{x,1}$, if $b$ is responsible of $x$ and $\nu(g)=1$;
		\item or, exactly one input $G^{x,0}$ that is the $\NOT$-gate $\NOT(G^{x,1})$, if $b$ is responsible of $x$ and $\nu(g)=0$.
	\end{itemize}

	In all cases, there can be at most one trace.
	For the inductive case, let $b$ be an internal bag with children $b_l$ and $b_r$.
	By induction hypothesis for every $G^{\nu_l,S_l}_{b_l}$ (resp., $G^{\nu_r,S_r}_{b_r}$), there exists at most one trace
       	of $D$ starting at $G^{\nu_l,S_l}_{b_l}$ (resp., $G^{\nu_r,S_r}_{b_r}$) according to $\chi$.
	Hence, if by contradiction there are at least two traces of $D$
        starting at $G^{\nu,S}_{b}$ according to $\chi$, it can only be because
        $G^{\nu,S}_{b}$ is not deterministic, i.e., because
	for at least two different inputs of $G^{\nu,S}_{b}$,
	there is a trace that starts at this input according to $\chi$,
	say $G^{\nu_l,S_l,\nu_r,S_r}_b$ and $G^{\nu'_l,S'_l,\nu'_r,S'_r}_b$ with $(\nu_l,S_l) \neq (\nu'_l,S'_l)$ or $(\nu_r,S_r) \neq (\nu'_r,S'_r)$. 
	We can suppose that it is $(\nu_l,S_l) \neq (\nu'_l,S'_l)$, since the other case is symmetric.
	Hence there exists $g_0 \in b_l$ such that 
	$(\nu_l,S_l)(g_0) \neq (\nu'_l,S'_l)(g_0)$.
	Let $\Xi_l$ be the trace of $D$ starting at $G^{\nu_l,S_l}_{b_l}$ and $\Xi'_l$ be the trace of $D$ starting at $G^{\nu'_l,S'_l}_{b_l}$ (according to $\chi$).
	We observe the following simple fact about $\Xi_l$ and $\Xi'_l$:

	\begin{enumerate}
		\item[(*)] for any $g$, if $\gamma(\Xi_l)(g) \neq \gamma(\Xi'_l)(g)$ then $g \notin b$. 
	\end{enumerate}

	Indeed otherwise we should have $\nu(g) = \nu_l(g) = \gamma(\Xi_l)(g)$ and $\nu(g) = \nu'_l(g) = \gamma(\Xi'_l)(g)$, 
	which is impossible.

	Now we will define an operator $\theta$ that takes as input a gate $g$ such that $(\gamma(\Xi_l), S_l)(g) \neq (\gamma(\Xi'_l), S'_l)(g)$, and outputs another
	gate $\theta(g)$ which is an input of $g$ and such that again $(\gamma(\Xi_l), S_l)(\theta(g)) \neq (\gamma(\Xi'_l), S'_l)(\theta(g))$.
	This will lead to a contradiction because for any $n \in \NN$, starting with $g_0$ and applying $\theta$ $n$ times consecutively
	we would obtain a path of $n$ mutually distinct gates (because $C$ is acyclic), but
	$C$ has a finite number of gates.

	Let us now define $\theta$: let $g$ such that $(\gamma(\Xi_l), S_l)(g) \neq (\gamma(\Xi'_l), S'_l)(g)$.
	We distinguish two cases:

	\begin{itemize}
		\item We have $(\gamma(\Xi_l), S_l)(g) \neq (\gamma(\Xi'_l), S'_l)(g)$ because $\gamma(\Xi_l)(g) \neq \gamma(\Xi'_l)(g)$. Then by (*), 
			we know for sure that $g \notin b$.
			Therefore the topmost bag $b'$ in which $g$ occurs is $\leq b_l$.
			Let $G^{\nu',S'}_{b'}$ be the gate in $\Xi_l$ and
                        $G^{\nu'',S''}_{b'}$ the gate in $\Xi'_l$ (they exist and are unique by Lemma~\ref{lem:only_one_agree2}). 
			We again distinguish two subcases:
			\begin{itemize}
				\item $g$ is a variable gate. But then it is clear that, by considering the bag $b''$ that is responsible for $g$, we have $b'' \leq b'$, and then
					that $\gamma(\Xi_l)(g) = \chi(g) =  \gamma(\Xi'_l)(g)$, a contradiction.
				\item $g$ is not a variable gate. Observe that by Lemma~\ref{lem:suspicious_not_innocent_propagate} we must have $g \notin S'$ 
			and $g \notin S''$, because otherwise
			we should have $b'=b$, which is not true.
			But then, by Lemma~\ref{lem:respects_semantics}
			 we know that both $\gamma(\Xi_l)$ and $\gamma(\Xi'_l)$ respect the semantics of $g$.
			But we have $\gamma(\Xi_l)(g) \neq \gamma(\Xi'_l)(g)$, so there must exist an input $g'$ of~$g$ such that $\gamma(\Xi_l)(g') \neq \gamma(\Xi'_l)(g')$.
			We can thus take $\theta(g)$ to be $g'$.
			\end{itemize}
		\item We have $(\gamma(\Xi_l), S_l)(g) \neq (\gamma(\Xi'_l),
                  S'_l)(g)$ because (without loss of generality) $g \notin S_l$ and $g \in S'_l$.
			Observe that this implies that $g \in b_l$, and that $\nu'_l(g)$ is a strong value for~$g$.
			We can assume that $\nu_l(g)=\nu'_l(g)$, as otherwise we would have $\gamma(\Xi_l)(g) \neq \gamma(\Xi'_l)(g)$, which is a case already
			covered by the last item. Hence $\nu_l(g)$ is also a strong value for $g$, but we have $g \notin S_l$, 
			so by the first item of Lemma~\ref{lem:innocent_or_suspect_but_not_innocent} we know that
			there exists an input $g'$ of $g$ that occurs in some bag $\leq b_l$ and such that $\gamma(\Xi_l)(g')$ is a strong value for~$g$.
			We show that $\gamma(\Xi'_l)(g')$ must in contrast be a weak value for $g$, so that we can take $\theta(g)$ to be $g'$ and conclude the proof.
			Indeed suppose by way of contradiction that $\gamma(\Xi'_l)(g')$ is a strong value for~$g$. By the contrapositive of the second item
			of Lemma~\ref{lem:innocent_or_suspect_but_not_innocent}, we get that $g \notin S'_l$, which contradicts our assumption.
	\end{itemize}

	Hence we have constructed $\theta$, which shows a contradiction, which means that in fact we must have $G^{\nu_l,S_l,\nu_r,S_r}_b = G^{\nu'_l,S'_l,\nu'_r,S'_r}_b$, so that 
	$G^{\nu,S}_{b}$ is deterministic, which proves that there is at most one trace of $D$ starting at $G^{\nu,S}_{b}$ according to $\chi$, which
	was our goal.
\end{proof}
This concludes the proof that $D$ is deterministic, and thus that $D$ is a d-SDNNF equivalent to $C$.

\subsubsection{Analysis of the Running Time}
We last check that the construction can be performed in time $O(|T| \times f(k))$, for some function $f(k)$ that is in $O(2^{(4+\epsilon)k})$ for any $\epsilon > 0$:
\begin{itemize}
	\item From the initial tree decomposition $T$ of $C$, in time
		$O(k|T|)$ we computed the $g_\out$-friendly tree decomposition $T_{\mathrm{friendly}}$ of size
          $O(k|T|)$;
	\item In linear time in $T_{\mathrm{friendly}}$, for every variable $x \in C_\var$ we selected a leaf bag $b_x$ of $T$ such that $\lambda(b_x) = \{x\}$;
	\item We can clearly compute the v-tree and the mapping in linear time in $T_{\mathrm{friendly}}$;
	\item For each bag $b$ of $T_{\mathrm{friendly}}$ we have $2^{2|b|}
          \leq 2^{2k+2}$ different pairs of a valuation $\nu$ of $b$ and of a subset $S$ of $b$, and checking 
if $\nu$ is a $(C,b)$-almost-evaluation and if $S$ is a subset of the unjustified gates of $\nu$ can be done in
polynomial time in $|b| \leq k+1$ (we access the inputs and the type of each
    gate in constant time from $C$), hence we pay
    $O(|T_{\mathrm{friendly}}| \times p(k) \times 2^{2k})$ to create the
    gates of the form $G^{\nu,S}_b$, for some polynomial $p$;
	\item We constructed  and connected in time $O(|T_{\mathrm{friendly}}| \times p'(k) \times 2^{4k})$ the gates of the form $G_b^{\nu_l,S_l,\nu_r,S_r}$ (the polynomial
		is for testing if the result of $(\nu_l,S_l)$ and of $(\nu_r,S_r)$ is an almost-evaluation);
	\item In time $O(|T_{\mathrm{friendly}}|)$ we connected the gates of the form $G^{\nu,S}_{b}$ for $b$ a leaf to their inputs.
\end{itemize}
Hence, the total time is indeed in $O(|T| \times f(k))$, for some function $f(k)$ that is in $O(2^{(4+\epsilon)k})$ for any $\epsilon > 0$.

%% file: set.tex
In this section, we start our presentation of our lower bound results.
Our upper bound in Section~\ref{sec:result} applied to arbitrary
Boolean circuits; however,
our lower bounds in this section and the next one will already apply to
much weaker formalisms for Boolean functions, namely, monotone DNFs and monotone CNFs. 

We first review some existing lower bounds on the compilation of
monotone CNFs and DNFs 
into OBDDs and d-SDNNFs. 
Bova and Slivovsky have
constructed a family of CNFs of bounded degree whose OBDD representations are
exponential~\cite[Theorem~19]{bova2017compiling},
following 
an earlier result of this type by Razgon~\cite[Corollary~1]{razgon2014obdds}. 
The result of Bova and Slivovsky is as follows:

\begin{theorem}[{\cite[Theorem~19]{bova2017compiling}}]
\label{thm:obddlowerbova}
There is a class of monotone CNF formulas of bounded degree and arity such that every formula $\phi$ in this class has OBDD size at least 
$2^{\Omega(|\phi|)}$.
\end{theorem}

By a similar approach, Bova, Capelli, Mengel, and Slivovsky show an exponential
lower bound on the size of d-SDNNF representing a given family of DNFs~\cite[Theorem~14]{bova2016knowledge}.
However, these bounds
apply to well-chosen families of
Boolean functions. 
We adapt some of these techniques to show a more general result. First, our lower
bounds will apply to \emph{any} monotone DNF or monotone CNF, not to one
specific family.
Second, our lower bounds apply to more expressive classes of binary decision
diagrams than OBDDs, namely, uOBDDs and nOBDDs (recall their definitions from
Section~\ref{sec:kc}).
Third, we obtain finer lower bounds on SDNNFs thanks to our new notion of
width.

In essence, our result is shown by observing that
the families of functions used in~\cite{bova2017compiling,bova2016knowledge} occur ``within'' any
bounded-degree, bounded-arity monotone CNF or DNF.
Here is the formal definition of these two families:

\begin{definition}
	\label{def:SCOV_SINT}
  Let $n \in \NN$ and consider two disjoint tuples $X = (x_1,\dots,x_n)$ and
  $Y = (y_1,\dots,y_n)$. The {\em set covering CNF}
$\SCOV_n(X,Y)$ is the monotone CNF:
\[ \SCOV_n(X,Y)=\bigwedge_{i=1}^n  x_i \vee y_i.\]
Similarly, the {\em set intersection DNF} $\SINT_n(X,Y)$ is
the monotone DNF:
\[ \SINT_n(X,Y)=\bigvee_{i=1}^n x_i \wedge y_i.\]
As the order chosen on $X$ and $Y$ does not matter, we will often abuse notation
  and consider them as sets rather than tuples.
\end{definition}

In this section, we prove lower bounds on the size of representations of the
functions $\SCOV_n$ and $\SINT_n$. We will then show 
in Section~\ref{sec:structured}
how to extend these bounds to arbitrary monotone CNFs/DNFs.

Our lower bounds on the representations of $\SCOV_n(X,Y)$ and $\SINT_n(X,Y)$
will only apply to some specific variable orderings.
Indeed, observe that both $\SCOV_n(X,Y)$ and $\SINT_n(X,Y)$ can easily be
represented by complete OBDDs of size $O(n)$ with the variable ordering $\mathbf{v} \defeq x_1 y_1 \ldots x_n y_n$.
The idea of our bounds is that ``inconvenient'' variable orderings (or
``inconvenient'' v-trees) can force OBDD (or SDNNF) representations of $\SCOV$
and $\SINT$ to be of exponential size.
We formalize our notion of inconvenient variable orderings and v-trees as follows:

\begin{definition}
	\label{def:cut}
Let $V$ be a set of variables, $X$ and $Y$ be two disjoint subsets of $V$.
	We say that a total order $\mathbf{v} =  v_1, \ldots, v_{|V|}$ of $V$ \emph{cuts $(X,Y)$} if there exists $1 \leq i \leq n$ such that $X \subseteq \mathbf{v}_{<i}$ and
	$Y \subseteq \mathbf{v}_{\geq i}$.
	Similarly, we say that a v-tree $T$ over $V$ \emph{cuts $(X,Y)$} if there exists a node $n$ of $T$ such that $X \subseteq \LEAVES(T_n)$ and
	$Y \subseteq \LEAVES(T \setminus T_n)$.
\end{definition}

In the rest of this section, we show the following two theorems. The
first theorem applies to OBDDs, and the second generalizes it to
SDNNFs.

\begin{theorem}
	\label{thm:obdd_lower_scov_sint}
	Let $\mathbf{v}$ be a total order that cuts $(X,Y)$. 
	Then the width of any complete nOBDD (resp., complete uOBDD) structured
        by $\mathbf{v}$ that computes $\SCOV_n(X,Y)$ (resp., $\SINT_n(X,Y)$) is $\geq 2^n -1$.
\end{theorem}

\begin{theorem}
	\label{thm:sdnnf_lower_scov_sint}
	Let $T$ be a v-tree that cuts $(X,Y)$. 
	Then the width of any complete SDNNF (resp., complete d-SDNNF)
        structured by $T$ that computes $\SCOV_n(X,Y)$ (resp., $\SINT_n(X,Y)$) is $\geq 2^n -1$.
\end{theorem}

These results are proved in two steps, presented in the next two sections. First, we show that Boolean
functions computed by SDNNF or nOBDD (resp., d-SDNNF or uOBDD) of width
$w$ can be decomposed as a disjunction (resp., exclusive disjunction) of at most $w$
very simple Boolean functions known as {\em rectangles}.
We then appeal to known results about these functions that
show that $\SCOV_n$ (resp., $\SINT_n$) cannot be decomposed as a disjunction
(resp., exclusive disjunction) of less than $2^n-1$ rectangles,
which implies the desired lower bound.

\subsection{Rectangle Covers for Compilation Targets}

Towards our desired bounds on the size of compilation targets, we start by
formalizing the notion of decomposing Boolean functions as a
\emph{rectangle cover}. 

\begin{definition}
Let $V$ be a set of variables and $(X,Y)$ be a partition of $V$. A
  \emph{$(X,Y)$-rectangle} is a Boolean function $R : 2^V \rightarrow \{0,1\}$
such that there exists $R_X \colon 2^X \rightarrow \{0,1\}$ and
$R_Y \colon 2^Y \rightarrow \{0,1\}$ such that $R = R_X \wedge
  R_Y$. In other words, for any valuation~$\nu$ of~$V$, we have $R(\nu) = 1$ iff
  $R_X(\restr{\nu}{X}) = 1$ and 
  $R_Y(\restr{\nu}{Y}) = 1$.

For any Boolean function $f \colon 2^V \rightarrow \{0,1\}$, a {\em $(X,Y)$-rectangle
  cover} of $f$ is a set $S$ of \mbox{$(X,Y)$-rectangles} such that
  $f = \bigvee_{R \in S} R$. The {\em size} $\card{S}$ of $S$ is the number of
  rectangles. We say that $S$ is
{\em disjoint} if for every $R,R' \in S$, we have $R \wedge R' = \bot$.
\end{definition}

Connections between rectangle covers and compilation target sizes have
already been successfully used to prove lower bounds,
see~\cite{wegener,beame2015new,bova2016knowledge}. We adapt these
results to relate the size of rectangle covers to the width
of compilation targets in our context. We give proofs for these results that are
essentially self-contained: their main difference with existing proofs is that our proofs apply to
our notion of width whereas existing results generally apply to size.

We start by relating the width of OBDDs with the size of rectangle covers:

\begin{theorem}
  \label{thm:obdd_has_rectangle_cover}
  Let $O$ be a complete nOBDD on variables $V$ structured by the total
  order $\mathbf{v} = v_1, \ldots, v_{|V|}$. Let $1 \leq i \leq n$,
  let $\mathbf{v_{<i}} = \{v_j \mid j < i\}$ and
  let $\mathbf{v_{\geq i}} = \{v_j \mid j \geq i\}$. 
  There exists a
  $(\mathbf{v_{<i}}, \mathbf{v_{\geq i}})$-rectangle cover $S$ of $O$
  whose size is at most the $v_i$-width of $O$.
  Moreover, if $O$ is
  an uOBDD, then $S$ is disjoint.
\end{theorem}
\begin{proof}
  Let $v$ be a node of $O$ that tests variable $v_i$, and let $R_v$ be the
  set of valuations of $2^{\mathbf{v}}$ that are accepted in $O$ by a path going
  through $v$. We claim that $R_v$ is a
  $(\mathbf{v_{<i}}, \mathbf{v_{\geq i}})$-rectangle.
  Indeed, let $R_{\mathbf{v_{<i}}}$ be the set of valuations of
  $\mathbf{v_{<i}}$ compatible with some path in~$O$ from the root to~$v$, and
  let $R_{\mathbf{v_{\geq i}}}$ be the set of valuations of
  $\mathbf{v_{\geq i}}$ compatible with some path in~$O$ from~$v$ to the
  $1$-sink. Any valuation of~$R_v$ can clearly be written as the union of one valuation
  from $R_{\mathbf{v_{<i}}}$ and of one valuation from $R_{\mathbf{v_{\geq
  i}}}$. Conversely, given any pair  $\nu_{\mathbf{v_{<i}}}$ and $\nu_{\mathbf{v_{\geq
  i}}}$
  of valuations of these two sets, we can
  combine any two witnessing paths for these valuations to obtain a path in~$O$
  that witnesses that $\nu_{\mathbf{v_{<i}}} \cup \nu_{\mathbf{v_{\geq
  i}}}$ is in~$R_v$. Hence, it is indeed the case that $R_v = 
R_{\mathbf{v_{<i}}} \land R_{\mathbf{v_{\geq
  i}}}$.
  
  Consider now the set $O_i$ of gates of $O$ that test $v_i$. Since $O$ is
  complete, every accepting path of $O$ contains a gate of $O_i$. It
  follows that $\bigcup_{v \in O_i} R_v$ is a rectangle cover of $O$, and its
  size is at most~$|O_i|$, i.e., the $v_i$-width of~$O$.

  Now, if $O$ is an uOBDD, then let $\nu$ be a satisfying valuation of
  $O$. Since $O$ is unambiguous, there exists a unique path $\pi$ in
  $O$ compatible with $\nu$.  Moreover, since $O$ is complete, $\pi$
  contains exactly one node $v$ that is labeled by $v_i$, so that
  $\nu$ is in $R_v$, and not in any other $R_{v'}$ for $v' \neq v$.
  Hence, $\bigcup_{v \in O_i} R_v$ is a disjoint rectangle cover of
  $O$.
\end{proof}

We now generalize this result from OBDDs to SDNNFs. To do so, we use the connections of \cite[Theorem~13]{bova2016knowledge} and
\cite[Theorem~3]{pipatsrisawat2010lower} between rectangle covers and SDNNF
size, and adapt them to our notion of width:

\begin{theorem}[{\cite[Theorem~13]{bova2016knowledge} and
    \cite[Theorem~3]{pipatsrisawat2010lower}}]
  \label{thm:has_rectangle_cover}
  Let $(D,T,\rho)$ be a complete SDNNF on variables~$V$. Let $n \in
  T$.
  There is a 
  $(\LEAVES(T_n),
  \LEAVES(T \setminus T_n))$-rectangle cover $S$
  of~$D$ whose size is at most the $n$-width of $(D,T,\rho)$. 
  Moreover, if $D$ is a d-SDNNF, then $S$ is disjoint.
\end{theorem}
\begin{proof}
  Recall the notion of trace (Definition~\ref{def:trace}).
  Given an $\lor$-gate $g$ of $D$ structured by $n$, we define $R_g$ to
  be the set of valuations of $2^V$ that are accepted in $D$ by a
  trace going through $g$. Then $R_g$
  defines a $(\LEAVES(T_n), \LEAVES(T \setminus T_n))$-rectangle.
  Intuitively, any trace going through~$g$ defines a trace on the variables
  $\LEAVES(T_n)$ that starts at gate~$g$, and one ``partial'' trace on
  $\LEAVES(T \setminus T_n)$ where $g$ is also used as a leaf; conversely, any pair
  of such traces can be combined to a complete trace in~$D$ that goes
  through~$g$. The precise argument
  is given in \cite[Theorem~1]{bova2016knowledge}, where traces are called 
	\emph{certificates}.

  Consider now the set~$D_n$ of $\lor$-gates structured by $n$. Since $D$ is
  complete, every satisfying valuation of $D$ has a corresponding
  trace containing a gate in $D_n$; this uses the facts that in complete d-SDNNFs,
  no input of an $\land$-gate is an $\land$-gate, and an $\land$-gate structured by an internal node of the v-tree has exactly two children.
  It follows that
  $S = \bigcup_{g \in D_n} R_g$ is a
  $(\LEAVES(T_n), \LEAVES(T \setminus T_n))$-rectangle cover of~$D$ of
  size $|D_n|\leq w$.

  Now, if $D$ is a d-SDNNF, it is not hard to see that every
  satisfying assignment has exactly one accepting trace: otherwise, considering
  any topmost gate where the two traces differ, we see that this gate must be a
  $\lor$-gate where the two traces witness a violation of determinism. Moreover, if
  $\nu$ is a satisfying valuation of $R_g$ for some $g \in D_n$ then
  its unique trace contains $g$ and cannot contain
  another $g' \in D_n$: otherwise this would imply that one $\land$-gate is not
  decomposable, or that there is an $\lor$-gate having an $\lor$-gate as input
  which would be a violation of completeness. Thus $\nu$ is not in $R_{g'}$. In other words, $S$ is disjoint.
\end{proof}

\subsection{Rectangle covers of $\SCOV$ and $\SINT$}
\label{sec:rcCNF}

The second step of the proof of Theorems~\ref{thm:obdd_lower_scov_sint} and
\ref{thm:sdnnf_lower_scov_sint} is to observe that $\SCOV_n(X,Y)$ (resp.,
$\SINT_n(X, Y)$) does not have small $(X,Y)$-rectangle covers (resp., disjoint
covers). This is a folklore result in communication complexity: see, e.g., 
\cite[Section~3]{sherstov2014thirty} for the bound on $\SCOV_n(X, Y)$.
For completeness, we state and prove the result here:

\begin{theorem}
  \label{thm:sint_rect_lowerbound}
  Let $S$ be a $(X,Y)$-rectangle cover (resp., disjoint $(X, Y)$-rectangle
  cover) of $\SCOV_n(X,Y)$ (resp., of $\SINT_n(X,Y)$). Then
  $|S| \geq 2^n-1$.
\end{theorem}
\begin{proof}
  We first prove our claim for $\SCOV_n(X,Y)$ as in~\cite{sherstov2014thirty}.
  For any valuation $\nu$ of~$X$, we denote by~$\overline{\nu}$ the valuation
  of~$Y$ defined by $\nu(y_i) \colonequals \neg \nu(x_i)$ for all~$1 \leq i \leq
  n$. Consider the set $\calF \colonequals \{\nu \cup \overline{\nu} \mid \nu: X \to
  \{0, 1\}\}$: we have $\card{\calF} = 2^n$, it is clear that $\SCOV_n(X, Y)$ evaluates
  to true on each valuation of~$\calF$, and we will now show that $\calF$ is a
  \emph{fooling set} in the terminology of~\cite{sherstov2014thirty}.
  Specifically, consider the $(X, Y)$-rectangle
  cover $S$ and let us show that every rectangle contains at
  most one valuation of~$\calF$, which implies the bound. Assume by contradiction
  that some rectangle $R_X \land R_Y$ contains two valuations $\nu \cup
  \overline{\nu}$ and $\nu' \cup \overline{\nu'}$ of~$\calF$ for $\nu \neq \nu'$, so that $R_X(\nu) =
  R_X(\nu') = R_Y(\overline{\nu}) = R_Y(\overline{\nu'}) = 1$. This implies that
  the rectangle $R_X \land R_Y$ must also contain $\nu \cup \overline{\nu'}$ and
  $\nu' \cup \overline{\nu}$. However, as $\nu \neq \nu'$, there is $1
  \leq i \leq n$ where $\nu(x_i) \neq \nu'(x_i)$. The first case is that we have $\nu(x_i) = 1$ but
  $\nu'(x_i) = 0$, in which case $\overline{\nu}(y_i) = 0$ and
  $\overline{\nu'}(y_i) = 1$, but then $\SCOV_n(X, Y)$ evaluates to~$0$ on
  $\nu' \cup \overline{\nu}$, a contradiction. The second case is that we have
  $\nu(x_i) = 0$ but $\nu'(x_i) = 1$ and we conclude symmetrically using
  $\nu \cup \overline{\nu'}$.
  Thus we have shown that any rectangle of~$S$ can contain
  at most one valuation from $\calF$, so that $\card{S} \geq \card{\calF} \geq 2^n$,
  in particular $\card{S} \geq 2^n - 1$.

The proof of our claim for $\SINT_n(X,Y)$ can be found
  in~\cite{bova2016knowledge}. More precisely, it is exactly Theorem~15
  of~\cite{bova2016knowledge}, together with the second-last sentence in the
  proof of Proposition~14 of~\cite{bova2016knowledge}.
\end{proof}

We are now ready to prove Theorem~\ref{thm:obdd_lower_scov_sint} and
Theorem~\ref{thm:sdnnf_lower_scov_sint}:
\begin{proof}[Proof of Theorem~\ref{thm:obdd_lower_scov_sint}]
	Let $O$ be a complete nOBDD (resp., complete uOBDD) computing $\SCOV_n(X,Y)$ (resp., $\SINT_n(X,Y)$), where $O$ is structured by $\mathbf{v} = v_1,\ldots,v_{|V|}$ that cuts $(X,Y)$.
	Let $1 \leq i \leq n$ witnessing that $\mathbf{v}$ cuts $(X,Y)$, and let $w_i$ be the $v_i$-width of $O$.
	By Theorem~\ref{thm:obdd_has_rectangle_cover}, $\SCOV_n(X,Y)$ (resp., $\SINT_n(X,Y)$) has a $(X,Y)$-rectangle cover (resp., disjoint rectangle cover) of size $\leq w_i$.
	Hence, by Theorem~\ref{thm:sint_rect_lowerbound}, we have $w_i \geq 2^n -1$. 
	But then this implies that the width of $O$ is also $\geq 2^n -1$.
\end{proof}

The proof of Theorem~\ref{thm:sdnnf_lower_scov_sint} is similar, using
Theorem~\ref{thm:has_rectangle_cover} instead of
Theorem~\ref{thm:obdd_has_rectangle_cover}.

%% file: structured.tex
In this section we extend the lower bounds of the previous section from the
specific functions $\SCOV_n$
and $\SINT_n$, to obtain lower bounds that apply to
any family of CNFs and DNFs. Specifically, we will show:

\begin{theorem}
	\label{thm:obdd_lower_main}
	Let $\phi$ be a monotone CNF (resp., monotone DNF) of pathwidth~$k$,
        arity~$a$ and degree~$d$. 
	Then the width of any complete nOBDD (resp., any complete uOBDD) computing $\phi$ is $\geq 2^{\frac{k}{a^3d^2}} -1$.
\end{theorem}

\begin{theorem}
	\label{thm:SDNNF_lower_main}
	Let $\phi$ be a monotone CNF (resp., monotone DNF) of treewidth~$k$,
        arity~$a$ and degree~$d$. 
	Then the width of any complete SDNNF (resp., any complete d-SDNNF) computing $\phi$ is $\geq 2^{\frac{k}{3a^3d^2}} -1$.
\end{theorem}

Together with our upper bounds, this implies the following when we have a
constant bound on arity and
degree:

\begin{corollary}
  \label{cor:obdd_theta}
  For any monotone CNF $\phi$ (resp., monotone DNF $\phi$) of constant arity and
  degree, the width of the
  smallest complete nOBDD (resp., uOBDD) computing~$\phi$ is $2^{\Theta(\pw(\phi))}$.
\end{corollary}

\begin{corollary}
  \label{cor:SDNNF_theta}
  For any monotone CNF $\phi$ (resp., monotone DNF $\phi$) of constant arity and
  degree, the width of the
  smallest complete SDNNF (resp., d-SDNNF) computing~$\phi$ is $2^{\Theta(\tw(\phi))}$.
\end{corollary}

The completeness assumption can be lifted using the completion results
(Lemma~\ref{lem:complete_obdd_sdnnf}), to show a lower bound on representations
that are not necessarily complete. However, if we do this, we no longer have a
definition of width, so the lower bound is on the \emph{size} of the
representation (and thus is no longer tight).

\begin{corollary}
  \label{cor:obdd_omega}
  For any monotone CNF $\phi$ (resp., monotone DNF $\phi$) of constant arity and
  degree, the size of the
  smallest nOBDD (resp., uOBDD) computing~$\phi$ is $2^{\Omega(\pw(\phi))}$.
\end{corollary}

\begin{corollary}
  \label{cor:SDNNF_omega}
  For any monotone CNF $\phi$ (resp., monotone DNF $\phi$) of constant arity and
  degree, the size of the
  smallest SDNNF (resp., d-SDNNF) computing~$\phi$ is $2^{\Omega(\tw(\phi))}$.
\end{corollary}

In the case of OBDDs, it is easy to generalize
width to
non-complete OBDDs such that we can lift the
completeness assumption in Corollary~\ref{cor:obdd_theta}: see Theorem~15
of~\cite{amarilli2018connecting}. However, this result
in~\cite{amarilli2018connecting} is only shown for OBDDs (not nOBDDs or uOBDDs),
and it is open whether it extends to these larger classes. As for
SDNNFs, we leave to future work the task of generalizing width to non-complete
circuits and showing comparable bounds.

To prove Theorem~\ref{thm:obdd_lower_main} and \ref{thm:SDNNF_lower_main}, we
will explain how we can find $\SCOV_n$ (resp., $\SINT_n$) in any monotone CNF
(resp., DNF) of high pathwidth/treewidth. To this end, we first present a general notion of
a set of clauses being \emph{split} by two variable subsets. We will
then show Theorem~\ref{thm:obdd_lower_main}, and last show Theorem~\ref{thm:SDNNF_lower_main}.

\subsection{From Split Sets of Clauses to $\SCOV_n$ and $\SINT_n$}

To prove Theorem~\ref{thm:obdd_lower_main} and
\ref{thm:SDNNF_lower_main}, we will need to find a subset of the clauses of the
formula which is \emph{split} by two subsets of variables. In this subsection,
we introduce the corresponding notions. We first define the notions of \emph{split
clauses}, which we introduce in terms of hypergraphs because the definition is
the same for DNF and CNF.

\begin{definition}
	\label{def:split}
Let $H = (V, E)$ be a hypergraph, and let $X', Y'$ be two disjoint subsets of~$V$.
  We say that a set $E' \subseteq E$ of hyperedges is \emph{split} by~$(X', Y')$
  if every $e \in E'$ intersects~$X'$ and~$Y'$ nontrivially, i.e., $e \cap
  X' \neq \emptyset$ and $e \cap Y' \neq \emptyset$.
\end{definition}

If we can find a set $K'$ of clauses of a monotone CNF (resp., monotone DNF) that are
split by some pair $(X', Y')$ of disjoint variable subsets, then we can use it to find a partial
valuation that yields $\SCOV_n(X, Y)$ (resp., $\SINT_n(X, Y)$) for some $X
\subseteq X'$ and $Y \subseteq Y'$, where the number $n$ of extracted clauses depends on the number of clauses
in~$K'$ and on the arity and degree. Formally:

\begin{proposition}
  \label{prp:splitclauses}
  Let $\phi$ be a monotone CNF (resp., monotone DNF) with variable set~$V$,
  arity~$a$ and degree~$d$. Assume that there are two disjoint subsets $X', Y'$
  and a set $K'$ of clauses of~$\phi$ such that $K'$ is split by $(X', Y')$. 
  Let $n \colonequals \left\lfloor\frac{\card{K'}}{a^2 \times d^2}\right\rfloor$.
  Then we can find $X \subseteq X'$ and $Y \subseteq Y'$ such that $\card{X} =
  \card{Y} = n$, and a valuation $\nu$ of~$V \setminus (X \cup Y)$ such that
  $\nu(\phi) = \SCOV_n(X, Y)$ (resp., $\nu(\phi) = \SINT_n(X, Y)$).
\end{proposition}

We prove Proposition~\ref{prp:splitclauses} in the rest of this subsection.
Intuitively, the idea is to use the clauses of $K'$ to achieve $\SCOV_n(X, Y)$
or $\SINT_n(X, Y)$, and assign the other variables to eliminate them from the
clauses and eliminate the other clauses. Our ability to do this will rely on
the monotonicity of~$\phi$, but it will also require a careful choice of a
subset of clauses of~$K'$ that are ``independent'' in some sense: they should be
pairwise disjoint and any pair of them should never intersect a common clause.
We formalize this as an independent set in an \emph{exclusion graph} constructed
from the hypergraph of~$\phi$:

\begin{definition}
  The \emph{exclusion graph} of a hypergraph $H = (V, E)$ is the graph $G_H$
  whose vertices are the edges $E$ of~$H$, and where two edges $e \neq e'$ are
  adjacent if (1.) $e$ and $e'$ both intersect some edge $e'' \in E$, or if (2.) $e$ and
  $e'$ intersect each other: note that case (2.) is in fact covered by case (1.) by
  taking $e'' \colonequals e'$.
  Equivalently, $e$ and $e'$ are adjacent in~$G_H$ iff they are at
  distance~$\leq 4$ in the so-called incidence graph of~$H$. 
  Formally, we can define $G_H = (E, \{\{e, e'\}
  \in E^2 \mid e \neq e' \land \exists e'' \in E, (e \cap e'') \neq \emptyset
  \land (e' \cap e'') \neq \emptyset\}$. 
\end{definition}

Remember that an
\emph{independent set} of a graph~$G=(V,E)$ is a subset $S$ of~$V$ such that
no two elements of~$S$ are adjacent in $G$. We will use the following
easy lemma on independent sets:

\begin{lemma}
  \label{lem:indepset}
  Let $G = (V, E)$ be a graph and let $V' \subseteq V$.
  Then $G$ has an independent set $S \subseteq V'$ of size at least
  $\left\lfloor \frac{\card{V'}}{\degree(G) + 1} \right\rfloor$.
\end{lemma}

\begin{proof}
  We construct the independent $S$ set with the following trivial algorithm: start
  with $S \colonequals \emptyset$ and, while $V'$ is non-empty, pick an arbitrary
  vertex $v$ in~$V'$, add it to~$S$, and remove $v$ and all its neighbors
  from~$G$ and from~$V'$. It is clear that this algorithm terminates and adds the prescribed
  number of vertices to~$S$, so all that remains is to show that $S$ is an
  independent set at the end of the algorithm. This is initially true for $S =
  \emptyset$; let us show that it is preserved throughout the algorithm. Assume
  by way of contradiction that, at a stage of the algorithm, we add a vertex $v$
  to~$S$ and that it stops being an independent set. This means that $S$
  contains a neighbor $v'$ of~$v$ which must have been added earlier; but when
  we added $v'$ to~$S$ we have removed all its neighbors from~$G$, so we have
  removed $v$ and we cannot add it later, a contradiction. Hence, the algorithm
  is correct and the claim is shown.
\end{proof}

To use this lemma, let us bound the degree of~$G_H$ using the degree and arity of~$H$:
\begin{lemma}
  \label{lem:exclusiondegree}
  Let $H$ be a hypergraph. We have $\degree(G_H) \leq (\arity(H) \times \degree(H))^2
  - 1$.
\end{lemma}

\begin{proof}
  Any edge~$e$ of~$H$ contains $\leq \arity(H)$ vertices, each of which occurs
  in $\leq \degree(H)-1$ edges that are different from~$e$, so any edge $e$
  of~$H$ intersects at most $n \colonequals \arity(H) \times (\degree(H)-1)$
  edges different from~$e$. Hence, the degree of~$G_H$ is at most $n + n^2$
  (counting the edges that intersect $e$ or those at distance~$2$ from~$e$).
  Now, we have $n + n^2 = n(n+1)$, and as $\degree(H) \geq 1$ and $\arity(H)
  \geq 1$ (because we assume that hypergraphs contain at least one non-empty
  edge), the degree of~$G_H$ is $< \arity(H) \times \degree(H) \times (1 +
  \arity(H) \times (\degree(H) - 1))$, i.e., it is indeed $< (\arity(H) \times
  \degree(H))^2$, which concludes.
\end{proof}

We are now ready to show Proposition~\ref{prp:splitclauses}.

\begin{proof}[Proof of Proposition~\ref{prp:splitclauses}]
  Let $\phi$ be the
monotone formula with variable set $V$, arity~$a$, and degree~$d$, fix the sets
$X'$ and~$Y'$ and the set $K'$ of split clauses,
and let 
$n \colonequals \left\lfloor \frac{\card{K'}}{a^2\times d^2} \right\rfloor$.
Let $G_\phi$ be the exclusion
graph of~$\phi$ (seen as a hypergraph of clauses). By
Lemma~\ref{lem:exclusiondegree}, the graph $G_\phi$ has degree~$\leq a^2 \times
d^2 - 1$, so by Lemma~\ref{lem:indepset} it has an independent set $K'' \subseteq K'$
with $\card{K''} \geq n$. Let us pick any subset $K$ of~$K''$ that has
	cardinality exactly~$n$: $K$ is still an independent set of~$G_\phi$, and $K$ is still split by $(X',Y')$.

Let us now define $X$ by choosing one element of~$X'$ in each clause of~$K$, and
define~$Y$ accordingly. We have $X \subseteq X'$, $Y \subseteq Y'$, so that $X'$
and $Y'$ are disjoint; and the cardinality of~$X$ and~$Y$ is exactly~$n$, because
the clauses of~$K$ are pairwise disjoint as none of them are adjacent
in~$G_\phi$.

Let us now define the valuation~$\nu$ of~$V \setminus (X \cup Y)$. We first let
$Z \colonequals \bigcup K$ be the set of variables occurring in the clauses
  of~$K$: we have $X \cup Y \subseteq Z$. Let $\nu_1$ be the partial valuation 
  that assigns all variables of~$V \setminus Z$ to~$1$ if $\phi$ is a CNF and to~$0$
  if~$\phi$ is a DNF, and consider the formula $\nu_1(\phi)$: it consists
  precisely of the clauses of~$\phi$ that only use variables of~$Z$ (in particular those
  of~$K$), because the other clauses evaluate to true (if $\phi$ is a CNF) or
  false (if $\phi$ is a DNF) and so are simplified away.

  Let us now observe that in fact $\nu_1(\phi)$ precisely consists of the
  clauses of~$K$. Indeed, the clauses of~$K$ are clearly in~$\nu_1(\phi)$, and
  for the converse let us assume by contradiction that $\phi$ contains a clause
  $e''$ that only uses variables of~$Z$ but is not a clause of~$K$. We know that
  $e''$ cannot be the empty clause because we have disallowed it, so it contains
  a variable of $Z$, which means that it intersects a clause $e$
  of~$K$. Now, by hypothesis $e''$ is not a clause of~$K$, and as $\phi$ is
  minimized we know that $e''$ cannot be a subset of~$e$, which means that it
  must contain some variable of~$Z$ which is not in~$e$, hence it must intersect
  some other clause $e' \neq e$ of~$K$. Hence, $e''$ intersects both $e$
  and~$e'$, which is impossible as $K$ is an independent set of~$G_\phi$ but
  $e''$ witnesses that $e$ and $e'$ are adjacent in~$G_\phi$. We conclude that
  $\nu_1(\phi)$ precisely consists of the clauses of~$K$.

  Now, let us consider the partial valuation $\nu_2$ of $\nu_1(\phi)$ that assigns all
  variables of $Z \setminus (X \cup Y)$ to~$0$ (if $\phi$ is a DNF) or to~$1$
  (if $\phi$ is a CNF).
  Let $\nu \colonequals \nu_1 \cup \nu_2$.
  It is clear that the clauses of $\nu(\phi)$ are
  the intersection of the clauses of~$\nu_1(\phi)$ with $X \cup Y$, i.e., the
  intersection of the clauses of~$K$ with $X \cup Y$. Now, the definition of~$K$
  ensures that each clause contains exactly one variable of~$X$ and one
  variable of~$Y$, with each variable occurring in exactly one clause. Thus, it
  is the case that $\nu(\phi) = \SCOV_n(X, Y)$ or $\nu(\phi) = \SINT_n(X, Y)$,
  which concludes the proof.
\end{proof}

\subsection{Proof of Theorem~\ref{thm:obdd_lower_main}: Pathwidth and OBDDs}

In this section, we explain how we can obtain $\SCOV_n$ (resp., $\SINT_n$) by applying
a well-chosen partial valuation to any monotone CNF (resp., monotone DNF). The
key result is:

\begin{proposition}
	\label{prp:reduce_pathwidth}
	Let $\phi$ be a monotone CNF (resp., monotone DNF) of pathwidth $\geq k$ on variables $V$. 
	Then, for any variable ordering $\mathbf{v}$ of $V$, there exist disjoint subsets $X,Y$ of $V$ such that $\mathbf{v}$ cuts $(X,Y)$
	and a valuation $\nu$ of $V \setminus (X \cup Y)$ such that $\nu(\phi) = \SCOV_l(X,Y)$ (resp., $\SINT_l(X,Y)$) for some $l \geq \frac{k}{a^3d^2}$.
\end{proposition}

Thanks to this result, we can extend Theorem~\ref{thm:obdd_lower_scov_sint} to
arbitrary monotone CNFs/DNFs, which is what we need to prove Theorem~\ref{thm:obdd_lower_main}:
\begin{proof}[Proof of Theorem~\ref{thm:obdd_lower_main}]
	Let $O$ be a complete nOBDD  (resp., complete uOBDD) computing~$\phi$,
        and let $\mathbf{v}$ be its order on the variables~$V$. 
	By Proposition~\ref{prp:reduce_pathwidth}, there exist disjoint subsets $X,Y$ of $V$ such that $\mathbf{v}$ cuts $(X,Y)$
	and a valuation $\nu$ of $V \setminus (X \cup Y)$ such that $\nu(\phi) = \SCOV_l(X,Y)$ (resp., $\SINT_l(X,Y)$) for $l \geq \frac{k}{a^3d^2}$.
	By applying Lemma~\ref{lem:partial_evaluate_nOBDD} to~$O$, we know that there is a complete nOBDD (resp., complete uOBDD) $O'$ on variables $X \cup Y$ with order $\mathbf{v'} = \mathbf{v}|_{X \cup Y}$ computing $\SCOV_l(X,Y)$ (resp., $\SINT_l(X,Y)$), whose width is no greater than that of $O$.
	Now, it is clear that $\mathbf{v'}$ still cuts $(X,Y)$, so that by Theorem~\ref{thm:obdd_lower_scov_sint} the width of~$O'$, and hence that of $O$, is $\geq 2^{\frac{k}{a^3d^2}} -1$.
\end{proof}

Hence, in the rest of this subsection, we prove
Proposition~\ref{prp:reduce_pathwidth}.

\myparagraph{Pathsplitwidth.}
The first step of the proof of Proposition~\ref{prp:reduce_pathwidth} is to rephrase the
bound on pathwidth, arity, and degree, in terms of a bound on the performance of variable
orderings. Intuitively, a good variable ordering is one which does not \emph{split}
too many clauses. Formally:

\begin{definition}
	\label{def:pathsplitwidth}
	Let $H=(V,E)$ be a hypergraph, and
        $\mathbf{v} = v_1, \ldots, v_{|V|}$ be an ordering on the variables of~$V$.
        For $1 \leq i \leq \card{V}$,
        we let $\spl_i(\mathbf{v}, H)$ be the set of hyperedges $e$ of~$H$ that
        contain both a variable at or before $v_i$, and a variable strictly
        after~$v_i$, i.e.,
        $\spl_i(\mathbf{v}, H) \colonequals \{e \in E \mid \exists l \in \{1,
        \ldots, i\} \text{~and~}
        \exists r \in \{i+1, \ldots, \card{V}\} \text{~such that~} \{v_l, v_r\} \subseteq~e\}$.
Note that $\spl_{|V|}(\mathbf{v},H)$ is always empty.

        The \emph{pathsplitwidth} of $\mathbf{v}$
        relative to $H$ is the maximum size of the split, formally,
        $\psw(\mathbf{v}, H) \colonequals \max_{1 \leq i \leq |V|} |
        \spl_i(\mathbf{v},H) |$.
        The \emph{pathsplitwidth} $\psw(H)$ of $H$ is then the
        minimum
        of $\psw(\mathbf{v}, H)$
        over all variable orderings $\mathbf{v}$ of~$V$.
\end{definition}

In other words, $\psw(H)$ is the smallest integer $n\in \NN$ such that,
for any variable ordering~$\mathbf{v}$ of the nodes of $H$, there is a moment at which $n$
hyperedges of $H$ are split by $(\mathbf{v_{\leq i}},\mathbf{v_{> i}})$, in the sense of Definition~\ref{def:split}. 
We note that the pathsplitwidth of $H$ is exactly the \emph{linear
branch-width}~\cite{exploring2017nordstrand} of the dual hypergraph of $H$,
but we introduced pathsplitwidth because it
fits our proofs better.
This being said, the definition of pathsplitwidth is also reminiscent of
pathwidth, and we can indeed connect the two (up to a factor of the
arity):

\begin{lemma}
	\label{lem:cw_pw}
        For any hypergraph $H=(V,E)$, we have
        $\pw(H) \leq \arity(H) \times \psw(H)$.
\end{lemma}

\begin{proof}
  Let $H = (V, E)$ be a hypergraph, and
  let $\mathbf{v}$ be an enumeration of the nodes of~$H$ witnessing that $H$ has
  pathsplitwidth $\psw(H)$.
  We will construct a path decomposition of~$H$ of width $\leq \arity(H) \times
  \psw(H)$.
  Consider the path $P = b_1, \cdots, b_{|V|}$ and the labeling function
  $\lambda$
  where $\lambda(b_i) \colonequals \{v_i\} \cup \bigcup
  \spl_i(\mathbf{v},H) $ for $1 \leq i \leq |V|$. Let us show that $(P,
  \lambda)$ is a path 
  decomposition of $H$: once this is established, it is clear that its width
  will be $\leq \arity(H) \times \psw(H)$.

  First, we verify the occurrence condition. Let $e \in E$.
  If $e$ is a singleton $\{v_i\}$ then $e$ is included in $b_i$. 
Now, if $|e| \geq 2$, then let $v_i$ be the first element of~$e$
  enumerated by~$\mathbf{v}$. We have $e\in\spl_i(\mathbf{v},H)$,
  and therefore $e$ is included in~$b_i$.

Second, we verify the connectedness condition. Let $v$ be a vertex of $H$,
  then by definition $v \in b_i$ iff $v=v_i$ or there exists $e \in \spl_i(\mathbf{v},H)$
with $v \in e$. We must show that the set~$T_v$ of the bags that contain
  $v$ forms a connected subpath in
  $P$. To show this, first observe that for every $e \in E$, letting $\spl(e) =
  \{v_i\mid 1\leq i<|V| \land e \in \spl_{i}(\mathbf{v},H)\}$, then $\spl(e)$ is clearly a
  connected segment of $\mathbf{v}$. Second, note that for every $e$ with $v \in
  e$, 
  then either $v \in \spl(e)$ or $v$ and the connected subpath $\spl(e)$ are adjacent (in the case where $v$
  is the last vertex of~$e$ in the enumeration). Now, by definition $T_v$
  is the union of the $b_{v'}$ for $v' \in \spl(e)$ with $v \in e$ and of~$b_i$, so it is a 
  union of connected subpaths which all contain $b_i$ or are adjacent to it:
  this establishes that $T_v$ is a connected subpath, which shows in turn that $(T,
  \lambda)$ is a path decomposition, concluding the proof.
\end{proof}

  For completeness with the preceding result, we note that the following also
  holds, although we do not use it (the proof is in the extended version of \cite{amarilli2018connecting}):

\begin{lemma}
        For any hypergraph $H$, it is the case that $\psw(H) \leq \degree(H) \times (\pw(H) + 1)$.
\end{lemma}

We are finally ready to prove Proposition~\ref{prp:reduce_pathwidth}:

\begin{proof}[Proof of Proposition~\ref{prp:reduce_pathwidth}]
	Let $\phi$ be the monotone CNF (resp., monotone DNF) on variables $V$ having pathwidth $\geq k$, $a$ be its arity, $d$ its degree,
and let $\mathbf{v}$ be a variable ordering of~$V$. Remember that we identify $\phi$ with its associated hypergraph.
By Lemma~\ref{lem:cw_pw}, the pathsplitwidth $k'$ of $\phi$ is $\geq \frac{k}{a}$.
By definition of pathsplitwidth, there exists $1 \leq i \leq |V|$ such that
$\spl_i(\mathbf{v}, \phi) \geq k'$.
	Let $X' \defeq \mathbf{v_{\leq i}}$ and $Y' \defeq \mathbf{v_{> i}}$, and let
	$K' \defeq \spl_i(\mathbf{v}, \phi)$. Then by definition, $K'$ is split by $(X',Y')$.
	Hence by
	Proposition~\ref{prp:splitclauses}, letting 
	$n\defeq \left\lfloor\frac{k'}{a^2 \times d^2}\right\rfloor \geq  \left\lfloor\frac{k}{a^3 \times d^2}\right\rfloor$, we can find $X \subseteq X'$ and $Y \subseteq Y'$ of size $n$ 
	and a valuation $\nu$ of $V \setminus (X \cup Y)$ such that $\nu(\phi) = \SCOV_n(X,Y)$ 
	(resp., $\nu(\phi) = \SINT_n(X,Y)$), which is what we wanted.
\end{proof}

\subsection{Proof of Theorem~\ref{thm:SDNNF_lower_main}: Treewidth and SDNNFs}

In this section we show our general lower bound relating the treewidth of CNFs/DNFs to the width of equivalent (d)-SDNNFs.
We proceed similarly to the previous section, and start by showing
the analogue of Proposition~\ref{prp:reduce_pathwidth} for treewidth and v-trees:

\begin{proposition}
	\label{prp:reduce_treewidth}
	Let $\phi$ be a monotone CNF (resp., monotone DNF) of treewidth $\geq k$ on variables $V$. 
	Then, for any v-tree $T$ of $V$,
        there exist disjoint subsets $X,Y$ of $V$ such that $T$ cuts $(X,Y)$
	and a valuation $\nu$ of $V \setminus (X \cup Y)$ such that $\nu(\phi) = \SCOV_l(X,Y)$ (resp., $\SINT_l(X,Y)$) for some $l \geq \frac{k}{3a^3d^2}$.
\end{proposition}

This allows us to extend Theorem~\ref{thm:sdnnf_lower_scov_sint} to arbitrary monotone CNFs/DNFs and to prove Theorem~\ref{thm:SDNNF_lower_main}:

\begin{proof}[Proof of Theorem~\ref{thm:SDNNF_lower_main}]
	Let $(D,T,\rho)$ be a complete SDNNF  (resp., complete d-SDNNF) on $V$ computing $\phi$. 
	By Proposition~\ref{prp:reduce_treewidth}, there exist disjoint subsets $X,Y$ of $V$ such that $T$ cuts $(X,Y)$
	and a valuation $\nu$ of $V \setminus (X \cup Y)$ such that $\nu(\phi) = \SCOV_l(X,Y)$ (resp., $\SINT_l(X,Y)$) for $l \geq \frac{k}{3a^3d^2}$.
	By Lemma~\ref{lem:partial_evaluate_DNNF}, there exists a complete SDNNF (resp., complete d-SDNNF) $(D',T',\rho')$ on variables $X \cup Y$ computing $\SINT_l(X,Y)$ whose width is no greater than that of $(D,T,\rho)$, and such that $T'$ is a reduction of $T$.
	Now, it is clear that $T'$ still cuts $(X,Y)$ (by definition of $T'$ being a reduction of $T$), so that by Theorem~\ref{thm:sdnnf_lower_scov_sint} the width of $(D',T',\rho')$, and hence that of $(D,T,\rho)$, is $\geq 2^{\frac{k}{3a^3d^2}} -1$.
\end{proof}

Hence, in the rest of this subsection, we prove Proposition~\ref{prp:reduce_treewidth}.

\myparagraph{Treesplitwidth.} Informally, treesplitwidth is to v-trees what
pathsplitwidth is to variable orders: it bounds the ``best performance'' of
any v-tree.
\begin{definition}
	\label{def:treesplitwidth}
	Let $H=(V,E)$ be a hypergraph, and $T$ be a v-tree over $V$. For any
        node~$n$ of $T$,
	we define $\spl_n(T,H)$ as the set of hyperedges $e$ of~$H$
        that contain both a variable in~$T_n$ and one outside~$T_n$
        (recall that $T_n$ denotes the subtree of~$T$ rooted
        at~$n$).
        Formally $\spl_n(T,H)$ is defined as the following set of hyperedges:
	\[ \{e \in E \mid \exists v_{\mathrm i} \in \LEAVES(T_n)
        \text{~and~} \exists
        v_{\mathrm o} \in \LEAVES(T \setminus T_n) \text{~such~that~}
	\{v_{\mathrm i}, v_{\mathrm o}\} \subseteq e\}\]

	The \emph{treesplitwidth} of $T$ relative to $H$ is $\tsw(T, H)
        \colonequals \max_{n \in T} |\spl_n(T,H)|$.
        The \emph{treesplitwidth} $\tsw(H)$ of~$H$ is then the
        minimum
         of $\tsw(T, H)$
        over all v-trees $T$ of~$V$.
\end{definition}

We note that the treesplitwidth of $H$ is exactly the
\emph{branch-width}~\cite{robertson1991obstructions} of the dual hypergraph of
$H$, but treesplitwidth is more convenient for our proofs.
As with pathsplitwidth and pathwidth (Lemma~\ref{lem:cw_pw}), we can
bound the treewidth of a hypergraph by its treesplitwidth:

\begin{lemma}
	\label{lem:tsw_tw}
	For any hypergraph $H=(V,E)$, we have
        $\tw(H) \leq 3 \times \arity(H) \times \tsw(H)$.
\end{lemma}

\begin{proof}
  Let $H = (V, E)$ be a hypergraph, and $T$ a v-tree over~$V$ witnessing
  that $H$ has treesplitwidth $\tsw(H)$.
	We will construct a tree decomposition $T'$ of $H$ of width $\leq 3
        \times 
        \arity(H) \times \tsw(H)$.
        The skeleton of $T'$ is the same as that of $T$. Now, for each node $n
        \in T$, we call $b_n$ the corresponding bag of $T'$, and we define the
        labeling $\lambda(b_n)$ of~$b_n$.

        If $n$ is an internal node of $T$ with children $n_l,n_r$ (recall
        that v-trees are assumed to be binary), then 
        we define
        $\lambda(b_n) \colonequals \bigcup \spl_n(T,H) \cup \bigcup \spl_{n_l}(T,H) \cup \bigcup \spl_{n_r}(T,H)$, 
	and if $n$ is a variable $v \in V$ (i.e., $n$ is a leaf of $T$) then
        $\lambda(b_n) \colonequals \{v\}$.
        It is clear that the width of $P$ is $\leq \max(3 \times \arity(H) \times
        \tsw(H),1) -1 \leq 3 \times \arity(H) \times \tsw(H)$.
	
        The occurrence condition is verified: let $e$ be an edge of
        $H$. If $e$ is a singleton edge~$\{v\}$
	then it is included in $b_v$. If $|e| \geq 2$ then there must exists a node $n \in T$ such that $e \in \spl_n(T,H)$.
	If $n$ is an internal node of $T$ then $e \subseteq \bigcup \spl_n(T,H) \subseteq b_n$, and if
	$n$ is a leaf node of $T$ then it must have a parent $p$ (since
        $e$ is split), and $e \subseteq \bigcup \spl_n(T,H) \subseteq b_p$.

        Connectedness is proved in the same way as in the proof of
        Lemma~\ref{lem:cw_pw}: for a given vertex $v\in V$, the nodes
        of~$T$ where each edge~$e$
        containing~$v$ is split is a connected subtree of~$T$ without its root
        node: more
        precisely, they are all the ancestors of a leaf in~$e$ 
        strictly lower than their the least
        common ancestor. Adding the missing root to each
        such subtree and unioning them all will result in the subtree of all
ancestors of a vertex adjacent to~$v$ (included $v$~itself) up to their
  least common ancestor~$a$.
        Consequently, the set of nodes of~$T'$ containing~$v$ is
        a connected subtree of~$T'$, rooted in~$b_a$.
\end{proof}

We are now ready to prove Proposition~\ref{prp:reduce_treewidth}, similarly to
the way we proved Proposition~\ref{prp:reduce_pathwidth}:

\begin{proof}[Proof of Proposition~\ref{prp:reduce_treewidth}]
  Let $\phi$ be the monotone CNF (resp., monotone DNF) on variables $V$ having
  treewidth $\geq k$, $a$ be its arity, $d$ its degree, and let $T$ be a v-tree
  of~$V$. By Lemma~\ref{lem:tsw_tw}, the treesplitwidth $k'$ of~$\phi$ is $\geq
  \frac{k}{3a}$. By definition of treesplitwidth, there exists $n \in T$ such
  that $\card{\spl_n(T, \phi)} \geq k'$. Let $X'$ be the variables in~$T_n$, let
  $Y'$ be the variables outside $T_n$, and let $K' \colonequals \spl_n(T,
  \phi)$. Then by definition, $K'$ is split by~$(X', Y')$. Hence by
  Proposition~\ref{prp:splitclauses}, letting $n \colonequals
  \left\lfloor\frac{k'}{a^2\times d^2}\right\rfloor \geq
  \left\lfloor\frac{k}{3 a^3\times d^2}\right\rfloor$, we can find $X \subseteq
  X'$ and $Y \subseteq Y'$ of size~$n$
	and a valuation $\nu$ of $V \setminus (X \cup Y)$ such that $\nu(\phi) = \SCOV_n(X,Y)$ 
	(resp., $\nu(\phi) = \SINT_n(X,Y)$), which is what we wanted.
\end{proof}

%% file: unstructured.tex
Theorem~\ref{thm:SDNNF_lower_main} gives an exponential lower
bound on the size of \emph{structured} DNNFs computing monotone CNF
formulas of treewidth $k$. Intuitively, the high treewidth of the CNF makes it possible to find a
large instance of $\SCOV$ for some set of variables that are
cut by the v-tree, and this implies a lower bound on the size of the SDNNF.
However, this argument crucially depends on the fact that the whole circuit is
structured by the same v-tree.

It turns out that we can nevertheless extend our exponential lower bound on
monotone CNF of treewidth $k$; but this requires a completely different proof
technique as we cannot isolate a single bad partition of variables anymore. As
in the rest of our work, the same argument applies to decision diagrams with
pathwidth.

As in the previous sections, our lower bounds in this section will apply to
so-called \emph{complete} circuits and decision diagrams.
However, the definitions of completeness in Section~\ref{sec:completewidth} were
only given for structured classes. We now give these missing definitions:

\begin{definition}
  \label{def:un_complete}
  We say that an nFBDD is \emph{complete} if, for any root-to-sink path $\pi$, all
  variables are tested along $\pi$, i.e., occur as the label of a node of~$\pi$. For DNNFs, recalling the definition of a \emph{trace} (see
  Definition~\ref{def:trace}), we say that a DNNF is \emph{complete} if, for
  any trace $\Xi$ starting at the output gate, all variable gates are in~$\Xi$.
\end{definition}

Observe that a complete nOBDD is indeed complete when seen as an nFBDD, and a
complete SDNNF is also complete when seen as an DNNF. 

Our main results in this section 
are the following analogues of
Corollaries~\ref{cor:obdd_omega} and~\ref{cor:SDNNF_omega}, where
the structuredness assumption is lifted:

\begin{theorem}
\label{thm:nFBDD_lower_main}
For any monotone CNF $\phi$ of constant arity and degree,
the size of the smallest complete nFBDD computing $\phi$ is $2^{\Omega(\pw(\phi))}$.
\end{theorem}

\begin{theorem}
\label{thm:DNNF_lower_main}
For any monotone CNF $\phi$ of constant arity and degree,
the size of the smallest complete DNNF computing $\phi$ is $2^{\Omega(\tw(\phi))}$.
\end{theorem}

Like in the previous section, we can in fact lift the completeness assumption
because one can make any nFBDD or any DNNF complete by only increasing its size
with a factor linear in the number of variables:

\begin{lemma}
  \label{lem:complete_circuit}
  For any nFBDD (resp. DNNF) $D$ on variables set $V$, there exists an equivalent complete nFBDD (resp. DNNF) of size at most $(|V|+1) \times |D|$.
\end{lemma}

The proof of Lemma~\ref{lem:complete_circuit} for nFBDD can be
straightforwardly adapted from~\cite[Lemma 6.2.2]{wegener2000branching} where it is
shown for FBDD. As for Lemma~\ref{lem:complete_circuit} for DNNF, it can be
shown similarly to the way that DNNF are made \emph{smooth}
in~\cite{darwiche2001tractable} (after Definition~4).

Combining Lemma~\ref{lem:complete_circuit} with
Theorems~\ref{thm:nFBDD_lower_main} and~\ref{thm:DNNF_lower_main} allows us to remove the completeness
assumption as we did for Corollary~\ref{cor:obdd_omega}
and~\ref{cor:SDNNF_omega} (since $(|V|+1)\times |D| \leq |D|^2$, and because $\sqrt{2^{\Omega(k)}}$ is also $2^{\Omega(k)}$):

\begin{corollary}
  \label{cor:unstructured}
  For any monotone CNF $\phi$
  of constant arity
  and degree, the size of the smallest nFBDD computing
  $\phi$ is
  $2^{\Omega(\pw(\phi))}$.
  Likewise, the size of the smallest DNNF computing~$\phi$ is 
  $2^{\Omega(\tw(\phi))}$.
\end{corollary}

We note that, in contrast with our results in Sections~\ref{sec:set}
and~\ref{sec:structured}, the above results \emph{only} apply to
CNFs. By contrast, for DNFs, there is no hope of showing a lower bound
on DNNFs, because any DNF is in particular a DNNF. A natural analogue
would be a lower bound on \emph{d-DNNFs} representations of DNFs, but
this is a long-standing open problem in knowledge
compilation~\cite{DBLP:journals/jair/DarwicheM02,bova2016knowledge}.

Corollary~\ref{cor:unstructured} generalizes a lower bound of
Razgon~\cite{razgon2014twnrop} where he constructs a family
$(F_n)_{n\in \mathbb{N}}$ of monotone 2CNFs such that for every $n$,
$F_n$ has $n$ variables and treewidth $k$. He proves that every nFBDD
computing $F_n$ has size at least $n^{\Omega(k)}$. One can actually
observe that $F_n$ has pathwidth $\Omega(k \log(n))$, making Razgon's
lower bound a consequence of Corollary~\ref{cor:unstructured}. This
observation is actually crucial in Razgon's reasoning and this is
exactly this fact that makes his proof works. Our lower bound is more
general though as it works for \emph{any} monotone CNF and
can be lifted to treewidth and DNNF. The proof technique is however
roughly the same and rely on a similar technical result,
Lemma~\ref{lem:small_fraction} here, which corresponds
to~\cite[Theorem 4]{razgon2014twnrop}. We give a simpler and
more generic presentation of the proof.

The proofs of Theorem~\ref{thm:nFBDD_lower_main} and
Theorem~\ref{thm:DNNF_lower_main} follow the same structure. We will first
present the proof of Theorem~\ref{thm:nFBDD_lower_main} and then explain how the
argument adapts from nFBDDs to DNNFs.

Both results rely on a bound on the number of satisfying valuations that a
rectangle can cover when its underlying partition splits (remember Definition~\ref{def:split}) a large number of
clauses. The result can be understood as a
variant of Proposition~\ref{prp:splitclauses}, coupled with a generalization of
the notion of \emph{fooling sets} in the proof of
Theorem~\ref{thm:sint_rect_lowerbound}.

\begin{lemma}
  \label{lem:small_fraction}
  Let $\phi$ be a monotone CNF with variable
  set~$V$, arity~$a$ and degree~$d$. Let $X,Y$ be a partition of $V$
  and $R$ be an $(X,Y)$-rectangle such that $R \Rightarrow \phi$.
  Assume that there exists a set $K'$ of clauses of~$\phi$ such that $K'$ is
  split by $(X, Y)$.  Let
  $n \colonequals \left\lfloor\frac{\card{K'}}{a^2 \times
      d^2}\right\rfloor$.  Then we can bound the number of satisfying valuations
  of~$R$ as follows: $\#R \leq (1+\alpha_{d,a})^{-n}\times \#\phi$ where
  $\alpha_{d,a}=2^{-a^2 d} > 0$.
\end{lemma}
\begin{proof}
  We start by extracting a subset $K \subseteq K'$ of size $n$ from
  $K'$ such that $K$ is an independent set of $G_\phi$, the
  exclusion graph of $\phi$, as in the beginning of the proof of
  Proposition~\ref{prp:splitclauses}.

  Let $C$ be a clause of $K$. We denote by~$C \cap X$ the (disjunctive) clause
  formed of the variables in~$X$ that occur in~$C$.
  We claim that one of the following holds:
  \begin{itemize}    
  \item every satisfying valuation $\nu$ of $R$ satisfies $C \cap
    X$; or
  \item every satisfying valuation $\nu$ of $R$ satisfies $C \cap Y$.
  \end{itemize}

  Indeed, assume toward a contradiction that there exist two
  satisfying valuations $\nu_1,\nu_2$ of $R$ such that $\nu_1$ does
  not satisfy $C \cap X$ and $\nu_2$ does not satisfy $C \cap
  Y$. 
  Consider now the valuation
  $\nu \colonequals \restr{{\nu_1}}{X} \cup \restr{{\nu_2}}{Y}$. 
  As $R$ is a rectangle, it is easy to see that $\nu$ should satisfy~$R$,
  hence~$\phi$. However, 
  by definition $\nu$ does not satisfy~$C$, hence it does not satisfy~$\phi$, a
  contradiction.
  We point out here that this claim would not hold if we had a DNF formula instead of CNF.

  Thus, for every clause $C \in K$, all satisfying valuations of $R$
  satisfy either $C \cap X$ or $C \cap Y$.
  Let $K_X$ and $K_Y$ be a partition of~$K$ such that all clauses in~$K_X$
  (resp., in~$K_Y$) are
  such that all satisfying valuations of~$R$ satisfy $C \cap X$ (resp., $C \cap
  Y$): if there is any clause such that both conditions hold, assign it to~$K_X$
  or to~$K_Y$ arbitrarily.
  
  This definition ensures that every satisfying valuation $\nu$ of~$R$ satisfies $C
  \cap X$ for each clause $C \in K_X$ and that $\nu$ satisfies $C \cap Y$ for each
  clause $C \in K_Y$; what is more, as $R \Rightarrow \phi$, the valuation $\nu$
  must also satisfy all clauses in $\phi$. This means that we have $R
  \Rightarrow \psi(K_X,K_Y)$ where $\psi$ is defined as follows (up to
  minimization, i.e., removing clauses that are supersets of other clauses):
  \[\psi(K_X,K_Y) \colonequals 
  \phi
  \cup \left\{C
  \cap X \mid C \in K_X\right\} \cup \left\{C \cap Y \mid C \in K_Y\right\}.\]
  
We will now show that $\#\psi(K_X,K_Y) \leq (1+\alpha_{d,a})^{-n}\times \#\phi$. This
is enough to conclude the proof since $R \Rightarrow \psi(K_X,K_Y)$, that is,
$\#R \leq \#\psi(K_X,K_Y)$.

The proof is by induction on the size of $K_X \cup K_Y$. 
	More precisely, we will show
that given $C \in K_X$, it holds:
  \[ \#\psi(K_X,K_Y) \leq (1+\alpha_{d,a})^{-1}\times \#\psi(K_X \setminus \{C\},
  K_Y). \tag{*}\] The same is true for $C \in K_Y$, 
  but as the argument is symmetric, we only explain it for $K_X$. 
	This will be enough to show that $\#\psi(K_X,K_Y) \leq (1+\alpha_{d,a})^{-n}\times \#\phi$, since iteratively applying (*) to~$K_X$ and~$K_Y$ until $K_X \cup K_Y$ is
empty, and recalling that $K_X \cup K_Y = K$, we get:
\[\#\psi(K_X,K_Y) \leq (1+\alpha_{d,a})^{-(\card{K_X \cup
      K_Y})}\times\#\psi(\emptyset,\emptyset) = (1+\alpha_{d,a})^{-n}\times\#\phi. \]

	Hence, let us
  fix $C \in K_X$ and show (*).
  To do so, we will consider the satisfying valuations 
  of $\psi(K_X \setminus \{C\},K_Y)$ and distinguish those that happen to
  satisfy $C \cap X$ and those who do not. (Note that these valuations are not
  necessarily satisfying valuations of~$R$.) Specifically, let us call $S$ the
  set of satisfying valuations of~$\psi(K_X \setminus \{C\},K_Y)$ which do not
  satisfy~$C \cap X$. Observe that all the other satisfying valuations
  of~$\psi(K_X \setminus \{C\}, K_Y)$ also satisfy $\psi(K_X, K_Y)$. Hence, we
  have:
  \[\#\psi(K_X \setminus \{C\},K_Y) = \card{S} + \#\psi(K_X,K_Y)\tag{**}.\]This implies
  that to show (*), in combination with (**) it suffices to show:
  \[\#\psi(K_X,K_Y) \leq  \frac{\card{S}}{\alpha_{d,a}}.\tag{***}\]
  
  To show (***), we define a function $f$ from the satisfying valuations of~$\psi(K_X,K_Y)$
to~$S$, and show that each valuation in~$S$ has at most $1/\alpha_{d,a}$ preimages
by~$f$. To define~$f$, let us map every satisfying valuation $\nu$ of
$\psi(K_X,K_Y)$ to~$\nu' = f(\nu)$ defined as follows:
\begin{itemize}
  \item set $\nu'(x) \colonequals 0$ for every $x \in C \cap X$,
  \item set $\nu'(z) \colonequals 1$ for every $z \in V \setminus (C \cap X)$ that appears together in a
clause with a variable of $C \cap X$;
    \item set $\nu'(z') \colonequals \nu(z')$ for every other~$z'$.
\end{itemize}

We first show that $f$ is indeed a function that maps to~$S$, in other words:

\begin{claim}
  For any satisfying valuation $\nu$ of~$\psi(K_X, K_Y)$, letting $\nu'
  \colonequals f(\nu)$, we have $\nu' \in S$.
\end{claim}

\begin{proof}
  First, by definition, $\nu'$ does
not satisfy $C \cap X$ since $\nu'(x) = 0$ for every $x \in C \cap
X$. Now we have to show that $\nu'$ also satisfies
$\psi(K_X \setminus \{C\},K_Y)$. First observe that $\nu'$ satisfies $C$ since $C$ has
	at least one variable $y$ in $Y$ and by definition (precisely, by the second item), $\nu'(y) = 1$.

Now let $C'$ be a clause of $\phi$. Then either
$C'$ was satisfied by $\nu$ thanks to a variable $z$ such that
$\nu'(z) = \nu(z) = 1$. In this case, $C'$ is still satisfied by
$\nu'$. Otherwise, it may be that there is an $x \in C \cap X$ such
that $x \in C'$ and $\nu(x) = 1$. In this case, it is not guaranteed
anymore that $C'$ is satisfied by $\nu'$. However, since $\phi$ is
minimized, there must exist $z \in C'$ such that $z \notin C$. By
definition, $z$ appears in the same clause $C'$ as $x$, meaning that
	$\nu'(z) = 1$ (again by the second item) and thus $C'$ is satisfied by $\nu'$.

Finally let $C' \in (K_X \setminus \{C\})$ (the case $C' \in K_Y$ is
similar). We have to check that $C' \cap X$ is satisfied by $\nu'$
as this clause appears in $\psi(K_X \setminus \{C\}, K_Y)$. Since $C'$ is
also in $\psi(K_X,K_Y)$, $\nu$ satisfies $C'$, thus, there exists a
variable $y$ of $C'$ such that $\nu(y) = 1$. Now, we claim that
$\nu'(y) = 1$ too. Indeed, by definition of $K=K_X \cup K_Y$, all
clauses of $K$ form an independent set in $G_\phi$. Thus, $C'$ does
not share any variable with $C$. In other words,
	$\nu'(y) = \nu(y) = 1$ (by the third item), that is, $\nu'$ satisfies $C'$. Thus, we have
  established the claim.
\end{proof}

Now it remains to count the maximal number of preimages of a valuation $\nu' \in
S$ by~$f$.
To define $\nu'$ from~$\nu$ we only changed the valuation of
variables that occur in $C$ or in a clause
with a non-empty intersection with $C$.
There are at most $a^2d$ such variables: at most $a$ variables in $C$, and for
each of these variables, we have at most $d-1$ other clauses of size $a$
incident to it, that is, we change the value of at most $a+a(d-1)(a-1) \leq da^2$ variables.
Thus, given a
valuation $\nu'$ of $S$, there are at most $2^{a^2d} = 1/\alpha_{d,a}$ satisfying valuations of
$\psi(K_X \setminus \{C\},K_Y)$ whose image by~$f$ is $\nu'$. In other words, we
have shown~(***).

Combining this inequality with (**), we get:
\[ \#\psi(K_X,K_Y) \leq (1+\alpha_{d,a})^{-1} \times\#\psi(K_X \setminus
  \{C\},K_Y).\]

This concludes the proof of Lemma~\ref{lem:small_fraction}.
\end{proof}

Now that we have proved our technical lemma, we are ready to prove
our main result on nFBDDs, Theorem~\ref{thm:nFBDD_lower_main}:

\begin{proof}[Proof of Theorem~\ref{thm:nFBDD_lower_main}]
  Let $V$ be the variables of $\phi$, let $k \colonequals \pw(\phi)$, let $a$ be the arity
  of $\phi$, and let $d$ be the degree of $\phi$.
  Let $D$ be a complete nFBDD computing
  $\phi$.
 
  Let $\nu$ be a satisfying valuation of $\phi$.
  By definition of $D$, there exists a root-to-sink path~$\pi$ in $D$
  compatible with $\nu$. As we assumed $D$ to be complete,
  this path induces a total order on $V$. By
  Lemma~\ref{lem:cw_pw}, since $\phi$ is of pathwidth $k$, it is also
  of pathsplitwidth at least $k/a$. In other words,
  we know that there exists a set of clauses $K'$ of
  size at least $k/a$ and a gate~$g(\nu)$ in $\pi$ such that every clause of
  $K'$ has at least one variable tested before $g(\nu)$ in $\pi$ and one
  variable tested after $g(\nu)$ in $\pi$.
  In other words, letting $X$ (resp., $Y$) be the variables of $V$ tested before
  $g(\nu)$ (resp., after $g(\nu)$), we have that $K'$ is split by
  $(X, Y)$.
  
  We denote by $R_{g(\nu)}$ the set of satisfying valuations
  of $\phi$ having a compatible path going through
  $g(\nu)$. Observe now that, similarly to 
  how we proved Theorem~\ref{thm:obdd_has_rectangle_cover},
  $R_{g(\nu)}$ is an $(X,Y)$-rectangle such that
  $R_{g(\nu)} \Rightarrow \phi$: this uses in particular the fact that, although
  all valuations of~$R_{g(\nu)}$ may not test the variables in the same order,
  we know that all variables of~$X$ are tested before~$g(\nu)$, and all
  variables of~$Y$ are tested after~$g(\nu)$, thanks to the fact that each
  variable can be tested only once along root-to-sink paths in an nFBDD.
  Now, by
  Lemma~\ref{lem:small_fraction}, the number of valuations in~$R_{g(\nu)}$ is 
  $\leq (1+\alpha_{d,a})^{-n}\times \#\phi$
  where $n \colonequals {k \over a^3d^2}$.

  Now, observe that given $\nu \models \phi$, we can always find such
  a gate $g(\nu)$ inducing a rectangle $R_{g(\nu)}$ such that
  $\nu \models R_{g(\nu)}$ and such that the previous inequality holds,
  i.e., $\card{R_{g(\nu)}}$ is small wrt
  $\#\phi$. Let $S = \{g(\nu) \mid \nu \models \phi\}$. We thus have:
  $\phi = \bigvee_{g \in S} R_{g}$. Thus:
  \begin{align}
    \#\phi & \leq \sum_{g \in S} \#R_{g} \\
           & \leq \card{S}(1+\alpha_{d,a})^{-n}\times\#\phi \\
           & \leq \card{D}(1+\alpha_{d,a})^{-n}\times\#\phi \text{ since } S \subseteq D.
  \end{align}

  Simplifying by $\#\phi$ gives
  $\card{D} \geq (1+\alpha_{d,a})^{n} = 2^{\Omega(\pw(\phi))}$ since $d$
  and $a$ are constants.
\end{proof}

We have shown our result on nFBDDs, Theorem~\ref{thm:nFBDD_lower_main}. We now
explain how we can adapt the proof to DNNFs and show
Theorem~\ref{thm:DNNF_lower_main}:

\begin{proof}[Proof of Theorem~\ref{thm:DNNF_lower_main}]
  Let $V$ be the variables of~$\phi$, let $k \colonequals \tw(\phi)$, let $a$
  and $d$ be the arity and degree of~$\phi$, and let $D$ be a complete DNNF
  computing~$\phi$.

  For any satisfying valuation $\nu$ of~$\phi$, we consider a trace $\Xi$ of~$D$
  starting at the output gate that witnesses that $\nu$ satisfies $D$. As $D$ is complete, all variables of~$V$ occur in~$\Xi$. By
  Lemma~\ref{lem:tsw_tw}, we know that $\phi$ has treesplitwidth $\geq
  \frac{k}{3 a}$. Hence, we can define a gate $g(\nu)$ of~$\Xi$ such
  that every clause of~$K'$ has at least one variable in a leaf which a
  descendant of~$g(\nu)$ in~$\Xi$ (we call~$X$ the set of such variables) and
  one variable in a leaf which is not a descendant of~$g(\nu)$ in~$\Xi$ (we call
  $Y$ the set of such variables), i.e., $K'$ is split by~$(X, Y)$.

  We denote again by $R_{g(\nu)}$ the set of satisfying valuations of~$\phi$
  having a trace where~$g(\nu)$ occurs. 
  Similarly to how we proved Theorem~\ref{thm:has_rectangle_cover}, it is
  again the case that $R_{g(\nu)}$ is an $(X, Y)$-rectangle such that
  $R_{g(\nu)} \Rightarrow \phi$: this again uses the
  fact that, even though the different traces using $g(\nu)$ may have a very
  different structure, decomposability ensures that the variables of~$X$ must
  occur as descendants of~$g(\nu)$ in~$\Xi$ and the variables of~$Y$ must occur
  as non-descendants of~$g(\nu)$ in~$\Xi$.  
  Now, 
  Lemma~\ref{lem:small_fraction} ensures that the number of valuations in~$R_{g(\nu)}$ is 
  $\leq (1+\alpha_{d,a})^{-n}\times \#\phi$
  where $n \colonequals {k \over 3 a^3d^2}$.

  We conclude exactly as in Theorem~\ref{thm:nFBDD_lower_main} by considering
  the gates $g(\nu)$ for all satisfying valuations~$\nu$ and writing the
  corresponding rectangle cover. This establishes the desired bound.
\end{proof}

%% file: conclusion.tex
We have shown tight connections between 
structured circuit classes and width measures on circuits.
We constructively rewrite bounded-treewidth (resp., bounded-pathwidth) circuits
to d-SDNNFs (resp., uOBDDs) in time linear
in the circuit and singly exponential in the treewidth, and show matching lower
bounds for arbitrary monotone CNFs or DNFs under degree and arity assumptions,
also for CNFs in the unstructured case.
Our upper bound results imply the tractability of several tasks (probability
computation, enumeration, quantification, etc.) on bounded-treewidth and
bounded-pathwidth circuits, whereas our lower bounds show that pathwidth and
treewidth \emph{characterize} compilability to these classes. Our lower bounds
also have consequences for probabilistic query evaluation as described in the
conference version~\cite{amarilli2016tractable}.

Our work also raises a number of open questions. We leave to future work a more
thorough study of the relationships between the knowledge compilation classes
that we investigated, or their relationship to other classes such as 
\emph{sentential decision diagrams} (SDD)~\cite{darwiche2011sdd} which we did
not consider. Indeed, one interesting question is whether our upper bound result 
Theorem~\ref{thm:upper_bound} could be modified to construct SDDs, or whether
this is impossible and SDDs and d-SDNNFs can thus be separated. A related
question would be to characterize the bounds on the compilation of
bounded-pathwidth circuits to OBDDs (not uOBDDs): this can be done with doubly
exponential complexity in the pathwidth by the results of
\cite[Corollary~2.13]{jha2012tractability} but it is unclear whether this is
tight.
Another intriguing question is whether we could improve our main lower bound of Theorem~\ref{thm:upper_bound} by compiling to d-DNNFs that are not necessarily structured; could we then consider a less restrictive parameter than the treewidth of the original circuit?

In terms of our lower bounds, the main questions would be to investigate more
general languages, e.g., where the arity or degree are not bounded, or where the
functions are not monotone. There is also the question of proposing convincing
width definitions for non-complete circuit formalisms, so as to remove all
completeness assumptions from our results. Last, there is of course the
tantalizing question of showing a lower bound for \emph{unstructured} representations
of DNF formulae, i.e., an analogue for DNF and d-DNNF of the results of
Section~\ref{sec:unstructured}, that would match the results shown in 
Section~\ref{sec:structured} for the structured case. This relates to the open
problem in probabilistic databases of whether safe queries have tractable
lineages~\cite{jha2013knowledge,monet2018towards}.